\documentclass[letterpaper,12pt]{article}
\usepackage{amsmath}
\usepackage{amssymb}
\usepackage{amsfonts}
\usepackage{setspace}
\usepackage{bbm}
\usepackage{color}
\usepackage{enumerate}
\usepackage{mathrsfs}
\usepackage[longnamesfirst]{natbib}
\usepackage[bottom]{footmisc}
\usepackage{graphicx}
\usepackage{subcaption}
\usepackage{hyperref}
\usepackage[shortlabels]{enumitem}
\usepackage[usenames,dvipsnames,svgnames,table]{xcolor}
\usepackage{multirow}
\usepackage[margin=1.1in]{geometry}
\usepackage{tikz}
\usepackage[utf8]{inputenc}
\usepackage{pgfplots}
\usepackage{pgfgantt}
\usepackage{pdflscape}
\pgfplotsset{compat=newest}
\pgfplotsset{plot coordinates/math parser=false}
\usetikzlibrary{topaths,calc}
\usepackage{enumitem}

\newlist{steps}{enumerate}{1}
\setlist[steps, 1]{label = Step \arabic*:}

\newtheorem{theorem}{Theorem}[section]
\newtheorem{lemma}{Lemma}[section]

\newtheorem{assumption}{Assumption}

\newtheorem{corollary}{Corollary}[section]

\newtheorem{remark}{Remark}[section]

\newenvironment{proof}[1][Proof]{\noindent \textbf{#1.} }{\  \rule{0.5em}{0.5em}}
\newcommand{\mb}[1]{\mathbb{#1}}

\newcommand{\mf}[1]{\mathbf{#1}}
\newcommand{\mr}[1]{\mathrm{#1}}

\newcommand{\mc}[1]{\mathcal{#1}}

\newcommand{\ul}[1]{\underline{#1}}
\newcommand{\ol}[1]{\overline{#1}}

\newcommand{\T}[0]{\mathcal{T}}

\hypersetup{
     colorlinks   = true,
     citecolor    = Blue,
     urlcolor     = Black,
     linkcolor    = Blue
}

\makeatletter
\renewcommand\paragraph{\@startsection{paragraph}{4}{\z@}%
                                    {0pt \@plus1ex \@minus.2ex}%
                                    {-1em}%
                                    {\normalfont\normalsize\bfseries}}
\makeatother

\usepackage{float}

\newcommand{\E}{\mathbb{E}}
\newcommand{\R}{\mathbb{R}}
\graphicspath{{pic/}}
\DeclareGraphicsExtensions{.pdf,.jpg,.png}
\newcommand{\bpm}{\begin{pmatrix}}
\newcommand{\epm}{\end{pmatrix}}

\usepackage{bibunits}
\usepackage[longnamesfirst]{natbib}

\begin{document}

\defaultbibliography{refecon}
\defaultbibliographystyle{chicago}
\begin{bibunit}

\author{%
{Xiaohong Chen\thanks{%
 Cowles Foundation for Research in Economics, Yale University. \texttt{xiaohong.chen@yale.edu}}
 \quad \quad
 Timothy Christensen\thanks{%
 Department of Economics, University College London. \texttt{t.christensen@ucl.ac.uk}}
 \quad \quad
 Sid Kankanala\thanks{%
 Department of Economics, Yale University. \texttt{sid.kankanala@yale.edu}}
 }
}

\title{%
Adaptive Estimation and Uniform Confidence Bands for~Nonparametric~Structural~Functions~and~Elasticities\footnote{Authors are in alphabetical order. We are grateful to Francesca Molinari, two anonymous referees, Richard Nickl and Yixiao Sun for helpful suggestions, and to Rodrigo Ad{a}o, Costas Arkolakis and Sharat Ganapati for sharing their data. We thank participants of numerous workshops and the 2022 IAAE, 2022 AMES China, KEA2022, 2022 Toulouse Conference on Estimation and Inference in Econometric Models, and 2022 CIREQ Montreal Econometrics Conferences for comments. The data-driven choice of sieve dimension in this paper is based on and supersedes Section 3 of the preprint \href{https://arxiv.org/abs/1508.03365v1}{\texttt{arXiv:1508.03365v1}} \citep{chen2015arxiv}. The research is partially supported by the Cowles Foundation Research Funds (Chen) and the National Science Foundation under Grant No.~SES-1919034 (Christensen).}
}

\date{First version: July 24, 2021; Revised version: November 22, 2023}

\maketitle

\begin{abstract}
\singlespacing
\noindent
We introduce two data-driven procedures for optimal estimation and inference in nonparametric models using instrumental variables.
The first is a data-driven choice of sieve dimension for a popular class of sieve two-stage least squares estimators.
When implemented with this choice, estimators of both the structural function $h_0$ and its derivatives (such as elasticities) converge at the fastest possible (i.e., minimax) rates in sup-norm.
The second is for constructing uniform confidence bands (UCBs) for $h_0$ and its derivatives.
Our UCBs guarantee coverage over a generic class of data-generating processes and contract at the minimax rate, possibly up to a logarithmic factor.
As such, our UCBs are asymptotically more efficient than UCBs based on the usual approach of undersmoothing.
As an application, we estimate the elasticity of the intensive margin of firm exports in a monopolistic competition model of international trade.
Simulations illustrate the good performance of our procedures in empirically calibrated designs.
Our results provide evidence against common parameterizations of the distribution of unobserved firm heterogeneity.

\medskip

\noindent \textbf{Keywords:} Honest and adaptive uniform confidence bands, minimax rate-adaptive estimation, nonparametric instrumental variables, nonparametric estimation of elasticities, international trade.

\end{abstract}

\thispagestyle{empty}

\pagenumbering{arabic}

\newpage

\section{Introduction}

With easier access to large data sets, there is increasing interest in estimating flexible, nonparametric structural functions and their derivatives, such as elasticities or other marginal effects. In many applications, the structural function $h_0$ is identified by a conditional moment restriction
\begin{equation}\label{eq:npiv}
 \E[Y - h_0(X) |W]  = 0 \:\: (\mbox{almost surely}) ,
\end{equation}
where $Y$ (a scalar) and/or some elements of $X$ (a vector) are endogenous, $W$ is a vector of instrumental variables, and the conditional distribution of $(X,Y)$ given $W$ is otherwise unspecified. Examples include consumer demand \citep{blundell2007semi,BHP}, demand for differentiated products \citep{BerryHaile2014,Compiani}, and international trade \citep{ACD,AAG}.\footnote{Other applications include causal inference \citep{WangZhiTT2018} and reinforcement learning \citep{chen2022well,grettonRL2021}. Model (\ref{eq:npiv}) also nests nonparametric regression when $W=X$, in which case $h_0$ is the conditional mean of $Y$ given $X$.} Uniform confidence bands (UCBs) are very helpful for inferring the true shape, slope, or curvature of $h_0$, as they graphically convey sampling uncertainty about the estimated structural function and its derivatives.

In applications involving policy counterfactuals, researchers care about estimating and constructing UCBs for $h_0$ or its derivatives. For instance, \citeauthor{AAG} (\citeyear{AAG}, AAG hereafter) derive (\ref{eq:npiv}) via a semiparametric gravity equation for the intensive margin of firm exports in a monopolistic competition model based on \cite{Melitz2003}. In that context, the derivative of $h_0$ is the elasticity of the intensive margin of firm-level exports to changes in bilateral trade costs. Moreover, \cite{Compiani} performs policy experiments using nonparametric estimates of price elasticities in differentiated product demand models.

As is the case for almost all nonparametric and machine learning (ML) methods, researchers must choose tuning parameters---such as bandwidths, sieve dimensions, or penalty parameters---when estimating or performing inference on $h_0$ and its derivatives.
Poor choice of tuning parameters can lead to estimators that converge unnecessarily slowly and confidence bands with poor coverage. But ``good'' choices of tuning parameters typically require knowledge of key model regularities, such as the smoothness of $h_0$ and the strength of the instruments, which are unknown ex ante. It is therefore important to have data-driven methods that \emph{adapt} to unknown model regularities and yield estimators and confidence bands with desirable properties.
Data-driven methods for choosing tuning parameters also help to improve the transparency of nonparametric and ML methods, removing a degree of freedom with which the researcher can manipulate results.
Unfortunately, popular methods for choosing tuning parameters for nonparametric regression, such as standard cross validation, may not be valid in models with endogeneity---see Section~\ref{sec:cv}.

In this paper, we propose simple, data-driven procedures for choosing tuning parameters for estimating and constructing UCBs for $h_0$ and its derivatives. Our methods are developed for the popular class of sieve nonparametric IV estimators.\footnote{See \cite{ai2003efficient}, \cite{newey2003instrumental}, \cite{blundell2007semi}, and \cite{horowitz2011}.} That is, $h_0$ is approximated by a linear combination of several basis functions (e.g., B-splines), with the coefficients estimated by Two Stage Least Squares (TSLS) regression of $Y$ on the basis functions of $X$, using functions of $W$ as instruments (see Section~\ref{sec:npiv} for a detailed description). The key tuning parameter to be chosen by a researcher is the number of basis functions, say $J$, used to approximate $h_0$. If $J$ is too small, then estimators may be badly biased and UCBs may under-cover. But if $J$ is too large, estimators may be very noisy and UCBs
may be uninformatively wide. Before precisely stating our theoretical results in Section~\ref{s:theory}, we describe our methods and their practical importance.

\medskip

\paragraph{Our Methods and the Practical Implications.}

Our first contribution is a data-driven choice of sieve dimension, which we denote by $\tilde J$. This choice is simple to compute.
Under suitable regularity conditions, we show that sieve estimators implemented with $\tilde J$, which we denote $\hat h_{\tilde J}$, converge at the fastest possible (i.e., minimax) rate in sup-norm.\footnote{We focus on the sup-norm rather than $L^2$ norm (i.e., mean-square error)  primarily because our objective is to construct UCBs for $h_0$ and its derivatives. The sup-norm is essential for this purpose, as we require the entire function (or its derivatives) to lie inside the bands with desired coverage probability. The sup-norm also provides a stronger, more informative sense in which the estimator is converging as it measures the maximal, rather than average, error over the support of $X$.}
That is, the maximum error over the support of $X$, namely
\[
 \sup_x |\hat h_{\tilde J}(x) - h_0(x)|,
\]
vanishes as fast as possible---among \emph{all} estimators of $h_0$---as the sample size increases, uniformly over a class of data-generating processes (DGPs), for both nonparametric IV and nonparametric regression models. Formally, we refer to $\tilde J$ as \emph{sup-norm rate-adaptive}: it adapts to features of the DGP that are unknown ex ante, such as the smoothness of $h_0$ and strength of the instruments, so that the resulting estimator $\hat h_{\tilde J}$ converges as fast as possible in sup-norm.
We further show that the same data-driven choice $\tilde J$ is sup-norm rate-adaptive for estimating \emph{derivatives} of $h_0$ as well.\footnote{This is in contrast to kernel estimation, in which different bandwidths must be used for rate-adaptive estimation of a function and its derivatives.} Hence, $\tilde J$ should be very useful for researchers interested in estimating elasticities or other marginal effects. We illustrate this usefulness in our empirical application revisiting AAG, where we use $\tilde J$ to estimate the elasticity of the intensive margin of firm-level exports from aggregate bilateral trade data. 
We also demonstrate the good performance of $\tilde J$ across a variety of simulation designs for both nonparametric IV estimation and nonparametric regression.

Our second main contribution is a data-driven approach to constructing UCBs for $h_0$ and its derivatives. The term ``uniform'' indicates that the entire function lies within the bands with desired asymptotic coverage probability. The UCBs for $h_0$ and its derivatives are also simple to compute and have strong theoretical justification. They are \emph{honest} in the sense that they guarantee coverage for $h_0$ and its derivatives uniformly over a generic class of DGPs, and \emph{adaptive} in the sense that they contract at, or within a logarithmic factor of, the minimax rate.
As such, they provide efficiency improvements relative to UCBs based on the usual approach of undersmoothing, in which a sub-optimally large $J$ is chosen in the hope that bias is negligible relative to sampling variation. Of course, in empirical work, a researcher does not know the true function, and therefore doesn't know which  $J$ is truly large enough that sampling uncertainty dominates bias.

Our UCBs for $h_0$ and its derivatives are useful for inferring the true shape of the structural function and its derivatives. They complement existing approaches for testing shape restrictions, as they allow the researcher to read off the shape of the function without imposing a specific null (e.g. monotone increasing) a priori. In our empirical application to AAG we construct UCBs for the elasticity of the intensive margin of firm exports. As emphasized by AAG, this is an important, policy-relevant function yet its shape is not restricted by theory in a nonparametric setting. Our UCBs exclude constant functions and downwards-sloping functions. Hence, they provide evidence against the Pareto specification for unobserved firm productivity used by \cite{Chaney}, which leads to a constant elasticity, as well as other parameterizations used, e.g., by \cite{EKK}, \cite{HMT}, and \cite{MelitzRedding}, for which the elasticity is downwards-sloping. Empirically-calibrated simulation studies based on the models of \cite{Chaney} and \cite{HMT} demonstrate valid coverage of our UCBs for $h_0$ and its derivatives and efficiency improvements relative to undersmoothing.

\medskip

\paragraph{Related Literature and our Theoretical Contributions.}

Early work on nonparametric IV estimation includes \cite{newey2003instrumental}, \cite{hall2005nonparametric}, \cite{blundell2007semi}, \cite{darolles2011nonparametric}, \cite{horowitz2011} and others.

We complement prior work by \cite{horowitz2014adaptive} for near-adaptive estimation of $h_0$ in $L^2$ norm, \cite{breunig2016adaptive} for near-adaptive estimation of linear functionals of $h_0$, and \cite{breunigchen2021} for adaptive estimation of quadratic functionals of $h_0$. Our procedure builds on the bootstrap-based implementation of Lepski's method of \cite{chernozhukov2014anti} for kernel density estimation and \cite{spokoiny2019} for linear regression with Gaussian errors. But our procedure does not follow easily from theirs due to several challenges present in the conditional moment restriction (\ref{eq:npiv}), in which $h_0$ is identified by $\E[Y|W] = \E[h_0(X)|W]$ (a.s.).
The degree of difficulty of inverting $\E[h_0(X) |W]$ to recover $h_0$ is a nonparametric notion of instrument strength and plays an important role in determining minimax rates for estimators of $h_0$ and its derivatives.\footnote{See \cite{hall2005nonparametric}, \cite{chen2011rate}, and \cite{chen2018optimal} for minimax rates for nonparametric IV estimation. When the conditional density of $X$ given $W$ is continuous, these rates are slower than the corresponding rates for nonparametric regression.}
While adaptive procedures for nonparametric density estimation or regression deal only with unknown smoothness of the estimand, our procedures must also deal with the unknown degree of difficulty of the inversion problem. The literature has typically classified the difficulty of the inversion problem into ``mild'' and ``severe'' regimes. Minimax rates in the mild regime are achieved by a choice of sieve dimension that balances bias and sampling uncertainty, much like standard nonparametric problems. But minimax rates in the severe regime are obtained by a bias-dominating choice of sieve dimension. Our procedure for data-driven choice of sieve dimension delivers the minimax sup-norm rate for $h_0$ and its derivatives across the whole spectrum of models, from nonparametric regression to nonparametric IV models in the severe regime.

Our procedure improves significantly on and supersedes a modified Lepski procedure from Section 3 of \cite{chen2015arxiv} on sup-norm rate-adaptive estimation of~(\ref{eq:npiv}). Ours uses a multiplier bootstrap to avoid selection of several constants and performs much better in practice. Moreover, our rate-adaptivity guarantees encompass nonparametric regression and nonparametric IV in both mild and severe regimes.

Recent work on (non data-driven) UCBs for $h_0$ and functionals thereof via undersmoothing includes \cite{horowitz2012uniform}, \cite{chen2018optimal} and \cite{babii2020}. Our UCBs build on prior work on honest, adaptive UCBs for nonparametric density estimation \citep{gine2010confidence,chernozhukov2014anti} and Gaussian white noise models \citep{bull2012honest,gine2016mathematical}. But none of these works allows for nonparametric models with endogeneity, and our procedures do not follow easily from these existing methods due to the above-mentioned challenges present in model (\ref{eq:npiv}). Our UCBs for $h_0$ and its derivatives apply to nonparametric regression with non-Gaussian, heteroskedastic errors as a special case, which appears to be a new contribution.

Finally, our work also compliments several recent papers on (non data-driven) estimation and inference for nonparametric IV models with shape constraints; see for example \cite{BHP}, \cite{chetverikov2017}, \cite{FreybergerReeves} and \cite{CNS}.
These works all assume a deterministic sequence of tuning parameters satisfying regularity conditions that depend on unknown model features such as the smoothness of $h_0$ and instrument strength.
An exception is \cite{breunig2020adaptive} who study $L^2$ rate-adaptive testing of a \emph{specific} null hypothesis (e.g., monotone increasing, or a parametric functional form). Our approach is conceptually different from theirs: our UCBs graphically convey sampling uncertainty about an estimate of $h_0$ and its derivatives.
Hence, our UCBs are very useful for inferring the true shape of $h_0$ in situations---such as our trade application---where there are no specific prior shape 
restrictions suggested by economic theory.

\medskip

\paragraph{Outline.}
Section~\ref{s:procedure} introduces our methods. Section~\ref{s:application} presents the application to international trade. Section~\ref{s:theory} contains the main theoretical results. Section~\ref{s:mcs} provides additional simulation results for difficult designs. Section~\ref{sec:extensions} presents extensions to additive and partially linear models, and Section~\ref{s:conclusion} concludes.  Appendix~\ref{sec:regression} presents a simplified version of our procedures for nonparametric regression. Appendix~\ref{ax:application} provides additional details for the trade application and simulations. In the online supplement, Appendix~\ref{appsec:engel} presents additional simulations to an empirically calibrated Engel curve design, Appendix~\ref{ax:basis} gives details on basis functions and nonparametric function classes, and Appendix~\ref{ax:proofs} contains technical results and proofs.

\medskip

\paragraph{Notation.} Let $\mc X$ be the support of $X$, $d$ the dimension of $X$, and $L^2_X$ and $L^2_W$ the space of functions of $X$ and $W$ with finite second moments. Let $\|h\|_{\infty} :=\sup_{x\in \mc X} |h(x)|$ be the sup-norm of $h : \mc X \to \mb R$. Let $\mb N$ be the set of integers and $\mb N_0 := \mb N \cup \{0\}$ the non-negative integers. Let $\lceil a \rceil = \min\{n \in \mathbb N : n \geq a\}$ and $\lfloor a \rfloor = \max\{n \in \mathbb N_0 : n < a\}$. For a multi-index $a=(a_1,...,a_d) \in (\mb N_0)^d$ with order $|a| = \sum_{i=1}^d a_i$, the $a$-derivative of $h$ is defined as
\[
 \partial^a h(x) = \frac{\partial^{|a|} h(x)}{\partial^{a_1} x_1 \ldots \partial^{a_d} x_d} \,.
\]
Let $A^-$ denote the generalized (or Moore--Penrose) inverse of a matrix $A$ and $A^{-1/2}$ the inverse of the positive-definite square root of $A$.

\section{Procedures}\label{s:procedure}

We begin in Section~\ref{sec:npiv} by briefly reviewing sieve nonparametric IV estimation and UCBs with a deterministic sieve dimension. Section~\ref{sec:cv} explains why standard cross validation for regression fails in models with endogeneity. Section~\ref{sec:Jtilde} presents our data-driven choice of sieve dimension and Section~\ref{sec:ucbs} presents our data-driven UCBs. These methods extend naturally to partially linear and partially additive models (see Section~\ref{sec:extensions}). Both procedures apply to nonparametric regression as well (see Appendix~\ref{sec:regression}).

\subsection{Review: Estimators and UCBs with a Deterministic $J$}\label{sec:npiv}

\paragraph{Estimators.} Consider approximating $h_0$  by a linear combination of $J$ basis functions:
\begin{equation}\label{eq:h0_approx}
 h_0(x) \approx
 (\psi^J(x))'c_J\,,
\end{equation}
where $\psi^J(x) = (\psi_{J1}(x),\ldots,\psi_{JJ}(x))'$ is a vector of basis functions and $c_J = (c_{J1},\ldots,c_{JJ})'$ is a vector of coefficients. Combining (\ref{eq:npiv}) and (\ref{eq:h0_approx}), we obtain
\[
 Y = (\psi^J(X))'c_J + \mathrm{bias}_J + u\,,~~~\E[u|W] = 0\,,
\]
where $u = Y - h_0(X)$ and $\mathrm{bias}_J = h_0(X) - (\psi^J(X))'c_J$. Provided the bias term is ``small'' relative to $u$ in an appropriate sense, we have an approximate linear IV model where $\psi^J(X)$ is a $J\times 1$ vector of ``endogenous variables'' and $c_J$ is a vector of unknown ``parameters''. One can then estimate $c_J$ using TSLS or GMM using a $K\times 1$ vector of basis functions $b^K(W)= (b_{K1}(W),\ldots,b_{KK}(W))'$ of $W$ as instruments.
Evidently, $K\geq J$ is necessary to estimate $c_J$.

Given data $(X_i,Y_i,W_i)_{i=1}^n$, the TSLS estimator of $c_J$ is simply
\[
 \hat c_J = \left(\mf \Psi_J' \mf P_K^{\phantom \prime} \mf \Psi_J^{\phantom \prime} \right)^{-} \mf \Psi_J' \mf P_K^{\phantom \prime} \mf Y \,,
\]
where $\mf \Psi_J  = (\psi^J({X_1}),\ldots,\psi^J({X_n}))'$ and $\mf B_K  = (b^K({W_1}),\ldots,b^K({W_n}))'$ are $n \times J$ and $n \times K$ matrices,
$\mf P_K = \mf B_K^{\phantom \prime} (\mf B_K' \mf B_K^{\phantom \prime})^{-} \mf B_K^{\prime}$
is the projection matrix onto the instrument space, and $\mf Y = (Y_1,\ldots,Y_n)'$ is a $n \times 1$ vector.
Estimators of $h_0$ and its derivative $\partial^a h_0$ are given by
\[
 \hat h_J(x) = (\psi^J(x))' \hat c_J \,,~~~~\mbox{and}~~~~\partial^a \hat h_J(x) =(\partial^a \psi^J(x))' \hat c_J \,,
\]
where $\partial^a \psi^J(x) = (\partial^a \psi_{J1}(x),\ldots,\partial^a \psi_{JJ}(x))'$.

\underline{Sieve Bases}. Many linear sieves, such as polynomial splines, B-splines, wavelets, Fourier series, and various polynomials, can be used as the instrument basis $\{b_{Kk}\}_{k=1}^K$. However, only B-splines and Cohen--Daubechies--Vial (CDV) wavelet bases for $\{\psi_{Jj}\}_{j=1}^J$ have been shown to achieve the optimal minimax sup-norm rates under a suitable choice of $J$ \citep{chen2018optimal}.\footnote{Bases for $h_0$ must have bounded Lebesgue constant to attain the minimax sup-norm rate for nonparametric regression (see, e.g., \cite{belloni2015some} and \cite{chen2015optimal}). B-splines and CDV wavelets have this property. Bases without this property, such as polynomials and Fourier series, cannot attain the minimax sup-norm rate and hence cannot lead to sup-norm rate-adaptive estimators or UCBs.} As our objective is to have estimators that converge as fast as possible in sup-norm---which is essential for constructing UCBs that are as narrow and informative as possible---we restrict attention to B-splines and CDV wavelets for $\{\psi_{Jj}\}_{j=1}^J$ in our theory that follows. Moreover, since B-splines are easy to compute, much less collinear than polynomials and polynomial splines, and available in standard software packages, we  confine our presentation to B-spline bases for both $\{\psi_{Jj}\}_{j=1}^J$ and $\{b_{Kk}\}_{k=1}^K$ in the main text.

\underline{Key tuning parameter $J$}. Based on simulations and theoretical studies in  \cite{blundell2007semi}, \cite{chen2018optimal} and others, the performance of the sieve TSLS estimator for $h_0$ is sensitive to the choice of $J$ and not sensitive to $K$ as long as $K\geq J$. We introduce a data-driven method for choosing $J$ in Section~\ref{sec:Jtilde}. The choice of $K$ is pinned down by $J$ in our procedure, so we write $K(J)\geq J$, $b^{K(J)}(W)$, $\mf B_{K(J)}$ and $\mf P_{K(J)}$ in what follows. Let $\mf M_J = (\mf \Psi_J' \mf P_{K(J)}^{\phantom \prime} \mf \Psi_J^{\phantom \prime} )^{-} \mf \Psi_J' \mf P_{K(J)}^{\phantom \prime}$ be a $J \times n$ matrix. We can equivalently write
\begin{equation}\label{J-TSLS}
 \hat h_J(x) = (\psi^J(x))' \mf M_J \mf Y \,,~~~~\partial^a \hat h_J(x) =(\partial^a \psi^J(x))' \mf M_J \mf Y \,.
 \end{equation}

\medskip

\paragraph{``Undersmoothed'' UCBs.}
We now review the usual approach of constructing ``undersmoothed'' UCBs for $h_0$ and its derivatives based on a deterministic $J$.
Let $\hat{\mf u}_J = (\hat u_{1,J},\ldots,\hat u_{n,J})'$ denote the $n \times 1$ vector of residuals whose $i$th element is $\hat u_{i,J} = Y_i - \hat h_J(X_i)$. Then $\hat h_J (x) - h_0 (x)$ and $\partial^a \hat h_J (x) - \partial^a h_0 (x)$ can be estimated by
\begin{equation}\label{eq:DJ}
D_J (x)=(\psi^J(x))' \mf M_J \hat{\mf  u}_J~,~~~~D_J^{a}(x) = (\partial^a \psi^J(x))' \mf M_J^{\phantom \prime} \hat{\mf u}_J~,
\end{equation}
and their variances can be estimated by
\begin{equation}\label{eq:DJ-var}
\hat \sigma_{J}^{2}(x)  =  (\psi^J(x))' \mf M_J^{\phantom \prime} \widehat{\mf U}_{J,J}^{\phantom \prime} \mf M_J' \psi^J(x),~~~~\hat \sigma_{J}^{a2}(x) =(\partial^a \psi^J(x))' \mf M_J^{\phantom \prime} \widehat{\mf U}_{J,J}^{\phantom \prime} \mf M_J' (\partial^a  \psi^J(x))~
\end{equation}
where $\widehat{\mf U}_{J,J}$ is a $n\times n$ diagonal matrix whose $i$th diagonal entry is $\hat u_{i,J} \hat u_{i,J}$.

Let $\hat{\mf u}_J^* = (\hat u_{1,J}\varpi_1,\ldots,\hat u_{n,J}\varpi_n)'$ denote a multiplier bootstrap version of $\hat{\mf u}_J$, where $(\varpi_i)_{i=1}^n$ are IID  $N(0,1)$ draws independent of the data.
Then
\begin{equation}\label{eq:DJ*}
D_J^*(x) = (\psi^J(x))' \mf M_J^{\phantom \prime} \hat{\mf u}_J^*~,~~~~D_J^{a*}(x) = (\partial^a \psi^J(x))' \mf M_J^{\phantom \prime} \hat{\mf u}_J^*
\end{equation}
are bootstrap versions of $D_J (x)$ and $D_J^{a}(x)$.
For each independent draw of $(\varpi_i)_{i=1}^n$, compute the sup $t$-statistics:
\begin{equation} \label{eq:z_star-J-UCB}
 \sup_{x\in \mathcal{X}} \left| \frac{D_J^*(x)}{\hat \sigma_J(x)} \right|~,~~~~\sup_{x \in \mathcal{X}} \left| \frac{D_J^{a*}(x)}{\hat \sigma_J^a(x)} \right|.
\end{equation}
Let $z_{1-\alpha,J}^*$ and $z_{1-\alpha,J}^{a*}$ denote the $(1-\alpha )$ quantile of
these sup statistics across a large number (say $1000$) independent draws of $(\varpi_i)_{i=1}^n$.
\cite{chen2018optimal} construct 100$(1-\alpha)$\% UCBs for $h_0$ and $\partial^a h_0$ as follows:
\[
 C_{n,J}(x) = \bigg[ \hat{h}_{J}(x) -  z_{1-\alpha,J}^* \hat \sigma_{J}(x) , ~  \hat{h}_{J}(x) + z_{1-\alpha,J}^* \hat \sigma_{J}(x) \bigg] \,,
\]
\[
 C_{n,J}^a(x) = \bigg[ \partial^a \hat{h}_{J}(x) - z_{1-\alpha,J}^{a*} \hat \sigma_{J}^a(x) ,~ \partial^a \hat{h}_{J}(x) +  z_{1-\alpha,J}^{a*} \hat \sigma_{J}^a(x) \bigg] \,.
\]
The above UCBs are theoretically justified provided $J$ increases faster than the oracle $J_0$ (the optimal sieve dimension for estimating $h_0$ or its derivatives in sup-norm), so that the bias is of smaller order than sampling uncertainty. Unfortunately, $J_0$ is unknown in practice since it depends on the unknown smoothness of $h_0$ and other unknown model regularities of (\ref{eq:npiv}). This motivates us to propose the new data-driven UCBs in Section~\ref{sec:ucbs}.

\subsection{Problems with Standard Cross Validation}\label{sec:cv}

We briefly explain why the usual approach of cross validation (CV) for regression is not a valid method for choosing $J$ in models with endogeneity. Consider the standard CV criterion
\begin{equation}\label{eq:cv}
 \mathrm{CV}(J) = \frac{1}{n} \sum_{i=1}^n (Y_i - \hat h_{-i, J}(X_i))^2,
\end{equation}
where $n$ is the sample size and $\hat h_{-i,J}$ denotes version of $\hat h_J$ computed from a sub-sample that excludes the $i$th observation. Let $u_i = Y_i - h_0(X_i)$. We may then expand (\ref{eq:cv}) as
\[
 \mathrm{CV}(J) = \frac{1}{n} \sum_{i=1}^n (h_0(X_i) - \hat h_{-i, J}(X_i))^2 + \frac{1}{n} \sum_{i=1}^n u_i^2  +   \frac{2}{n} \sum_{i=1}^n u_i (h_0(X_i) - \hat h_{-i, J}(X_i)).
\]
The first term in the expansion is an estimate of the MSE $\E[ (h_0(X) - \hat h_J(X))^2]$ of $\hat h_J$ and the second term is independent of $J$. The third term is an estimate of $\E[u(h_0(X) - \hat h_J(X))]$. This term is asymptotically negligible without endogeneity (i.e., when $\E[u|X] = 0$) as is the case for nonparametric regression, making $\mathrm{CV}(J)$ a suitable sample analogue of the mean-square error of $\hat h_J$ in that case (see, e.g., \cite{Li1987}). But in models with endogeneity (i.e., when $\E[u|X] \neq 0$), there is no guarantee that $\E[u(h_0(X) - \hat h_J(X))] = 0$ and so this third term---which depends on $J$---may be non-negligible even asymptotically. If so, cross validation gives a biased estimate of the MSE of $\hat h_J$ and is therefore not a meaningful criterion by which to choose $J$ in models with endogeneity. Indeed, a cross-validated choice of $J$ may not even lead to a consistent estimator of $h_0$ in model (\ref{eq:npiv}).

In addition, even for nonparametric regression, the $J$ chosen by CV balances bias and sampling uncertainty in $L^2$ norm. Such as choice is not optimal for estimation of $h_0$ and its derivatives in sup-norm, nor is it sutiable for adaptive UCBs for $h_0$ and its derivatives.

\subsection{Procedure 1: Data-driven Choice of Sieve Dimension}\label{sec:Jtilde}

We now present our data-driven choice $\tilde J$ of sieve dimension using B-spline bases.
B-splines are characterized by their order $r$. In the simulations and empirical application, we use a cubic B-spline ($r = 4$) for $\{\psi_{Jj}\}_{j=1}^J$ and a quartic B-spline ($r = 5$) for $\{b_{Kk}\}_{k=1}^K$.\footnote{In the first submitted version we also used a quadratic B-spline ($r=3$) for $\{\psi_{Jj}\}_{j=1}^J$. In additional simulations we obtained very similar results with a Fourier basis for $\{b_{Kk}\}_{k=1}^K$.}

Let $\mc T = \{J=(2^l + r-1)^d : l \in \mb N_0 \}$ denote a dyadic grid of candidate values of $J$, where the integer $r$ is the order of the B-spline basis for $\{\psi_{Jj}\}_{j=1}^J$ (i.e., each $\psi_{Jj}$ is a piecewise polynomial of degree $r -1$). For example, $\mc T = \{J=2^l + 3: l \in \mb N_0 \} \equiv \{4,5,7,11,19,35,\ldots\}$ for a scalar $X$ ($d=1$) and cubic B-splines $(r=4)$.\footnote{
Letting $J$ vary over $\mc T$ ensures there is enough separation that we can accurately compare the bias and variance of estimators with different $J \in \mc T$. This helps improve the numerical stability of the method, coherent with implementations of Lepski's method in other nonparametric contexts.}
The index $l$ is the \emph{resolution level}. We construct $\{b_{Kk}\}_{k=1}^K$ similarly, using B-splines of order $(r + 1)$ because the reduced form is smoother than $h_0$. Given the resolution level $l$ for the basis for $X$, the resolution level for the basis for $W$ is $l_w = \lceil (l + q) d/d_w \rceil$ for some $q \in \mb N_0$ where $d_w$ is the dimension of $W$. Linking $l_w$ to $l$ in this manner defines a mapping $K(J)$ that satisfies $\lim_{J \to \infty} K(J)/J = c \in [1,\infty)$.
We recommend taking $q$ as the second- or third-smallest value for which $K(J) \geq J$ holds for all $J$ (i.e., $q = 1$ or $q = 2$ if both $X$ and $W$ are of the same dimension). We advise against choosing $q$ any larger, as the number of basis functions increases exponentially in the resolution level.
Let $J^+ = \min\{j \in \T : j > J\}$ be the smallest sieve dimension in $\mc T$ exceeding $J$.

For $J,J_2 \in \mc T$ with $J_2 > J$, the contrast $D_{J}(x)-D_{J_2}(x)$ is an estimate of $\hat h_J(x) - \hat h_{J_2}(x)$, whose
variance can be estimated by
\begin{equation}\label{eq:hat-var-J2}
 \hat \sigma_{J,J_2}^2(x) := \hat \sigma_{J}^2(x) + \hat \sigma_{J_2}^2(x) - 2 \tilde \sigma_{J,J_2}(x),~~\tilde \sigma_{J,J_2}(x)  = (\psi^J(x))' \mf M_J^{\phantom \prime} \widehat{\mf U}_{J,J_2}^{\phantom \prime} \mf M_{J_2}' \psi^{J_2}(x),
\end{equation}
where $\hat \sigma_{J}^2(x)$ is defined in (\ref{eq:DJ-var}) and
$\widehat{\mf U}_{J,J_2}$ is a $n\times n$ diagonal matrix whose $i$th diagonal entry is $\hat u_{i,J} \hat u_{i,J_2}$. Moreover, the multiplier bootstrap version of $D_{J}(x)-D_{J_2}(x)$ is
\[
D_{J}^*(x)-D_{J_2}^*(x) = (\psi^J(x))'\mf M_J^{\phantom \prime}  \hat{\mf u}_J^* - (\psi^{J_2}(x))' \mf M_{J_2}^{\phantom \prime} \hat{\mf u}_{J_2}^*.
\]
Finally let $\hat s_J$ be the smallest singular value of $(\mf B_{K(J)}'\mf B_{K(J)}^{\phantom \prime})^{-1/2} (\mf B_{K(J)}' \mf \Psi_J^{\phantom \prime}) (\mf \Psi_J'\mf \Psi_J^{\phantom \prime})^{-1/2}$.

\bigskip

 \centerline{\underline{\bf Procedure 1: Data-driven Choice of Sieve Dimension}}

\begin{enumerate}
\item Compute
\begin{align}
 \hat{J}_{\max} & = \min \bigg \{ J \in \T :   J \sqrt{\log J}  \hat{s}_J^{-1}     \leq 10 \sqrt{n}  <  J^{+} \sqrt{\log J^{+}}   \hat{s}_{J^{+}}^{-1}  \bigg \} \, \label{eq:J_hat_max} \\
 \hat{\mc J} & = \left\{ J \in \T : 0.1 ( \log \hat J_{\max})^2 \leq J \leq \hat J_{\max}\right\} \,. \label{eq:index_set}
\end{align}
\item Let $\hat \alpha = \min\{ 0.5 , (\log(\hat{J}_{\max})/\hat{J}_{\max})^{1/2}\}$.  For each draw of $(\varpi_i)_{i=1}^n$, compute
\begin{equation}\label{eq:sup-stat}
 \sup_{\{ (x,J,J_2) \in \mathcal{X} \times \hat{\mc J} \times \hat{\mc J} : J_2 > J \}} \left| \frac{D_{J}^*(x)-D_{J_2}^*(x)}{\hat \sigma_{J,J_2}(x)} \right|.
\end{equation}
Let $\theta^*_{1-\hat \alpha}$ denote the $(1- \hat \alpha )$ quantile of (\ref{eq:sup-stat}) across independent draws of $(\varpi_i)_{i=1}^n$.
\item Let $\hat J_n = \max\{J \in \hat{\mc J} : J < \hat J_{\max}\}$ and
\begin{equation}\label{eq:J_lepski}
 \hat{J} = \min \left \{ J \in \hat{\mc J} : \sup_{(x, J_2) \in \mathcal{X} \times \hat{\mathcal{J}} : J_2 > J } \left| \frac{\hat h_{J}(x)-\hat h_{J_2}(x)}{\hat \sigma_{J,J_2}(x)} \right| \leq 1.1 \theta^*_{1 - \hat \alpha} \right \} \,.
\end{equation}
The data-driven choice of sieve dimension is
\begin{equation} \label{eq:J-choice}
 \tilde{J} = \min\{\hat{J},\hat J_n\}\,.
\end{equation}
\end{enumerate}

\begin{remark}\normalfont
In practice, the supremums over $x$ in Steps 2 and 3 can be replaced by the maximum over a fine grid of $x$ values as the functions are continuous in $x$. We have used 1000 draws of $(\varpi_i)_{i=1}^n$ in our empirical and simulation studies. Note the $(\varpi_i)_{i=1}^n$ are held fixed when computing the supremum over $(x,J,J_2)$ for each draw. Our theory allows for constants other than 10 and 0.1 in Step~1 as long as they ensure $\hat{\mathcal J}$ contains several values of $J$ to search over. Our theory also allows for any constant larger than 1 in Step~3; the value 1.1 performed well in simulations and is used in other implementations of Lepski's method (see, e.g., \cite{chernozhukov2014anti}).
\end{remark}

We present the theoretical results on the adaptivity of $\tilde J$ in Section~\ref{sec:4estimation}.

\subsection{Procedure 2: Data-driven UCBs}\label{sec:ucbs}

Let $\ul p >d/2$ denote the minimal degree of smoothness assumed for $h_0$. For instance, if $X$ is scalar and $h_0$ is Lipschitz, then one could take $\ul p = 1$ even through the true smoothness of $h_0$ is unknown.
Let $\hat A = \log \log \tilde J$ and
\[
 \hat{\mc J}_{-} =
 \begin{cases}
 \{J \in \hat{\mc J} : J < \hat J_n\} & \mbox{~if $\tilde J = \hat J$}, \\
 \hat{\mc J} & \mbox{~if $\tilde J = \hat J_n$}.
 \end{cases}
\]

\bigskip

 \centerline{\underline{\bf Procedure 2: Data-driven UCBs for $h_0$}}
\begin{enumerate}
\item[4.] For each $(\varpi_i)_{i=1}^n$, compute
\begin{equation} \label{eq:z_star-UCB}
 \sup_{(x,J) \in \mathcal{X} \times \hat{\mc J}_{-}} \left| \frac{D_J^*(x)}{\hat \sigma_J(x)} \right|.
\end{equation}
Let $z_{1-\alpha}^*$ denote the $(1-\alpha )$ quantile of (\ref{eq:z_star-UCB}) across independent draws of $(\varpi_i)_{i=1}^n$.
\item[5.] Construct the 100$(1-\alpha)$\% UCB
\begin{equation} \label{band}
 C_n(x) = \bigg[ \hat{h}_{\tilde{J}}(x) - \mbox{cv}^*(x) \, \hat \sigma_{\tilde J}(x) , ~  \hat{h}_{\tilde{J}}(x) + \mbox{cv}^*(x) \, \hat \sigma_{\tilde J}(x) \bigg] \,,
\end{equation}
where
\begin{equation}\label{cv-h}
 \mbox{cv}^*(x) =
 \begin{cases}
 z_{1-\alpha}^* + \hat A \theta^*_{1-\hat \alpha}  & \mbox{~if $\tilde J = \hat J$}, \\
 z_{1-\alpha}^* + \hat A \max\{\theta^*_{1-\hat \alpha}\,,\, \tilde {J}^{-\ul p/d} /\hat \sigma_{\tilde J}(x) \} & \mbox{~if $\tilde J = \hat J_n$}.
 \end{cases}
\end{equation}
\end{enumerate}

\bigskip

 \centerline{\underline{\bf Procedure 2$^\prime$: Data-driven UCBs for $\partial^a h_0$ ($0<|a|<\ul p$)}}

\begin{enumerate}
\item[4$^\prime$.] For each $(\varpi_i)_{i=1}^n$, compute
\begin{equation} \label{eq:z_star-UCB-derivative}
 \sup_{(x,J) \in \mathcal{X} \times \hat{\mc J}_{-}} \left| \frac{D_J^{a*}(x)}{\hat \sigma_J^a(x)} \right|.
\end{equation}
Let $z_{1-\alpha}^{a*}$ denote the $(1-\alpha )$ quantile of (\ref{eq:z_star-UCB-derivative})  across independent draws of $(\varpi_i)_{i=1}^n$.
\item[5$^\prime$.] Construct the 100$(1-\alpha)$\% UCB
\begin{equation} \label{band-derivative}
 C_n^a(x) = \bigg[ \partial^a \hat{h}_{\tilde{J}}(x) - \mbox{cv}^{a*}(x) \, \hat \sigma_{\tilde J}^a(x) ,~ \partial^a \hat{h}_{\tilde{J}}(x) + \mbox{cv}^{a*}(x) \, \hat \sigma_{\tilde J}^a(x) \bigg] ,
\end{equation}
where
\begin{equation}\label{cv-dev}
 \mbox{cv}^{a*}(x) =
 \begin{cases}
 z_{1-\alpha}^{a*} + \hat A \theta^*_{1-\hat \alpha} & \mbox{~if $\tilde J = \hat J$}, \\
 z_{1-\alpha}^{a*} + \hat A \max\{\theta^*_{1-\hat \alpha}, \tilde {J}^{(|a|-\ul p)/d} /\hat \sigma_{\tilde J}^a(x) \} & \mbox{~if $\tilde J = \hat J_n$}.
 \end{cases}
\end{equation}
\end{enumerate}

\begin{remark}\normalfont
Procedures~1 and~2 require choosing the B-spline order $r$ and Procedure~2 requires specifying the minimal degree of smoothness $\ul p$. For sup-norm estimation and UCBs for first derivatives one can take $r \geq 3$ and $\ul p \geq 1$; for second derivatives and cross elasticities one can take $r \geq 4$ and $\ul p \geq 2$.
\end{remark}

\begin{remark}\normalfont
We establish that $\tilde J =\hat J$ with probability approaching one (wpa1) in the mild regime; and that $\tilde J , \hat J \in [c \hat{J}_n,\hat{J}_n]$ wpa1 in the severe regime (for a constant $c\in (0,1)$). Nevertheless, we find $\tilde J = \hat J$ in the empirical application and in the vast majority (between 99.6\% and 100\% depending on the design and sample size) of all simulations. In particular, $\tilde J = \hat J$ across all simulations in the Engel curve design which is in the severe regime (see Appendix~\ref{appsec:engel}).
\end{remark}

\noindent
Theoretical properties of these UCBs are presented in Sections~\ref{sec:4UCB} and~\ref{sec:4UCBde}. We show that the Procedures 2 and~2$^\prime$ UCBs are honest and adaptive for models in the mild regime (including nonparametric regression as a special case).
For models in the severe regime, we show that the Procedures 2 and~2$^\prime$ UCBs with critical values corresponding to $\tilde J =\hat{J}_n$ have valid (actually conservative) coverage.
Nevertheless, the Engel curve simulation in Appendix~\ref{appsec:engel} shows that the Procedure~2 UCBs still have valid (actually conservative) coverage for a severe regime design.

\section{International Trade: Simulations and Application}\label{s:application}

\citeauthor*{AAG} (\citeyear{AAG}, hereafter AAG) derive semiparametric gravity equations for the extensive and intensive margins of firm exports in a monopolistic competition model of international trade. Importantly, and in sharp contrast with the existing literature \citep{Melitz2003,Chaney,EKK,HMT,MelitzRedding}, AAG do not impose any parametric assumptions on the distribution of unobserved firm heterogeneity. The gravity equations identify functions which characterize the elasticities of the extensive and intensive margins of firm-level exports to changes in bilateral trade costs. AAG emphasize the importance of these elasticities for counterfactuals.

In this section, we apply our procedures to estimate and construct UCBs for the intensive margin and its elasticity using AAG's baseline model and data. We also present simulation studies based on empirical calibrations of two workhorse trade models to illustrate the sound performance of our procedures.

\subsection{Model and Data}

We begin by briefly summarizing the empirical framework of AAG. They use a monopolistic competition model of international trade---see \cite{MelitzRedding2014} for a review. There are a continuum of firms in each country. Firm $\omega$ in country $i$ is characterized by an entry potential $e_{ij}(\omega)$ and a revenue potential $r_{ij}(\omega)$ for selling in country $j$. Firms draw $e_{ij}(\omega)$
from a distribution $H_{ij}^e(e)$ then $r_{ij}(\omega)$
from a (possibly degenerate) distribution $H_{ij}^r(r|e)$. Firm $\omega$ in country $i$ exports to country $j$ if and only if $e_{ij}(\omega)$ exceeds a threshold. The proportion of firms in country $i$ that export to country $j$ is denoted $\pi_{ij}$.

The extensive margin is characterized by the inverse distribution of entry potential, i.e., $\epsilon_{ij}(\pi_{ij}) = (H_{ij}^e)^{-1}(1-\pi_{ij})$. Assuming homogeneity (so $H_{ij}^e = H^e$ and $\epsilon_{ij} = \epsilon$), AAG's gravity equation for the extensive margin is
\[
 \log \epsilon(\pi_{ij}) = \log (\bar f_{ij} \bar \tau_{ij}^{\tilde \sigma}) +  \delta_i^\epsilon +  \zeta_j^\epsilon,
\]
where $\bar \tau_{ij}$ and $\bar f_{ij}$ are variable and fixed trade costs from $i$ to $j$ and $ \delta_i^\epsilon$ and $ \zeta_j^\epsilon$ are exporter and importer fixed effects (FEs). Costs depend linearly on a cost shifter $z_{ij}$:
\[
 \begin{aligned}
 \log \bar \tau_{ij} & = \kappa^\tau z_{ij} + \delta_i^\tau + \zeta_j^\tau + \eta_{ij}^\tau, \\
 \log \bar f_{ij} & = \kappa^f z_{ij} + \delta_i^f + \zeta_j^f + \eta_{ij}^f,
 \end{aligned}
\]
where the idiosyncratic error terms $\eta_{ij}^\tau$ and $\eta_{ij}^f$ are conditionally mean-zero and independent of $z_{ij}$ and the FEs. This yields the estimating equation
\begin{equation}
 \log \epsilon(\pi_{ij})
 = (\kappa^f + \tilde \sigma \kappa^\tau) z_{ij} + (\delta_i^f + \tilde \sigma \delta_i^\tau +  \delta_i^\epsilon) + (\zeta_j^f + \tilde \sigma \zeta_j^\tau +  \zeta_j^\epsilon) + \eta_{ij}^f + \tilde \sigma \eta_{ij}^\tau. \label{eq:aag.log_eps}
\end{equation}
Note that $\pi_{ij}$ depends (possibly nonlinearly) on  $z_{ij}$ and the error terms $\eta_{ij}^f$ and $\eta_{ij}^\tau$.

The intensive margin is characterized by the average revenue potential of exporting firms:
\[
 \rho_{ij}(\pi) = \frac{1}{\pi} \int_0^\pi \E[r|e = \epsilon_{ij}(v)] \, \mathrm d v\,,
\]
where the expectation is taken under $H_{ij}^r(r|e)$. Assuming homogeneity (so $H^r_{ij} = H^r$ and $\rho_{ij} = \rho$), AAG's gravity equation for the intensive margin is
\[
 \log \bar x_{ij} - \log \rho (\pi_{ij}) = \log (\bar \tau_{ij}^{\tilde \sigma})  +  \delta_i^\rho +  \zeta_j^\rho,
\]
where $\bar x_{ij}$ are average firm exports and $ \delta_i^\rho$ and $ \zeta_j^\rho$ are FEs. With $\bar \tau_{ij}$ as above, AAG obtain
\begin{equation}\label{eq:aag.log_rho}
 \log \bar x_{ij} + \tilde \sigma \kappa^\tau z_{ij} = \log \rho(\pi_{ij}) + ( \delta_i^\rho - \tilde \sigma \delta_i^\tau) + ( \zeta_j^\rho - \tilde \sigma \zeta_j^\tau) - \tilde \sigma \eta_{ij}^\tau.
\end{equation}
More concisely,
\begin{equation}\label{eq:aag.npiv}
 y_{ij} = \log \tilde \rho(\tilde \pi_{ij}) + \delta_i + \zeta_j + u_{ij} \,,
\end{equation}
where $y_{ij} := \log \bar x_{ij} + \tilde \sigma \kappa^\tau z_{ij}$ is the dependent variable,\footnote{AAG construct $y_{ij}$ from data on $\bar x_{ij}$ and $z_{ij}$ based on external estimates of $\tilde \sigma$ and $\kappa^\tau$.} $\tilde \pi_{ij} := \log \pi_{ij}$ is the endogenous regressor, $\log \tilde \rho(\tilde \pi) := \log \rho(e^{\tilde \pi})$ is the unknown structural function, $\delta_i :=  \delta_i^\rho - \tilde \sigma \delta_i^\tau$ and $\zeta_j :=  \zeta_j^\rho - \tilde \sigma \zeta_j^\tau$ are exporter and importer FEs, and the idiosyncratic error term $u_{ij} :=  - \tilde \sigma \eta_{ij}^\tau$ is conditionally mean-zero and independent of the instrumental variable $z_{ij}$.

Our goal is to use (\ref{eq:aag.npiv}) to estimate $\log \tilde \rho$ and its derivative, as $\frac{\partial \log \tilde \rho(\tilde \pi)}{\partial \tilde \pi}\equiv \frac{\partial \log \rho(\pi)}{\partial \log \pi}$ characterizes the elasticity of the intensive margin of firm-level exports to changes in bilateral trade costs. We use the same data that AAG use for their baseline estimates, which consists of $\bar x_{ij}$, $z_{ij}$, and $\tilde \pi_{ij}$ for a sample of 1522 country pairs for the year 2012. We refer the reader to AAG for a detailed description of the data and its construction.

\subsection{Implementation}

Model (\ref{eq:aag.npiv}) differs from model (\ref{eq:npiv}) due to the presence of FEs.
AAG estimate $\log \tilde \rho$ and FEs jointly, using both $z_{ij}$ and exporter and importer country dummies as instruments. As such, they estimate a partially linear  model with a large number of linear regressors  (due to the country dummies) and, similarly, a large number of instrumental variables.\footnote{These comments are based on the November 2020 version of AAG, which is currently under revision. Some of their implementation and findings may differ in future versions.} Our methods and theoretical results are not formally developed for such a setting.\footnote{Our approach extends to partially linear models---see Section~\ref{sec:extensions}. But with bilateral trade data the number of dummy variables representing origin and destination FEs is increasing with the sample size $n$. This ``many regressors/many instruments'' asymptotic framework falls outside the scope of our analysis.} Therefore, we maintain their assumption that $z_{ij}$ and origin and destination FEs are exogenous, but we further assume that $\E[\log \tilde \rho(\tilde \pi_{ij})|z_{ij}, \delta_i, \zeta_j] = \E[\log \tilde \rho(\tilde \pi_{ij})|z_{ij}]$ (a.s.). That is, the intensive margin is conditional mean independent of exporter- and importer-specific factors given cost shifters. Note, however, that we are not imposing that average firm exports are conditional mean independent of exporter- and importer-specific factors. The reduced form for $y_{ij}$ is
\begin{equation}\label{eq:aag_1}
 y_{ij} = g(z_{ij}) + \delta_i + \zeta_j + e_{ij} \,,
\end{equation}
where $g(z_{ij}) = \E[\log \tilde \rho(\tilde \pi_{ij})|z_{ij}]$ and $\E[e_{ij} | z_{ij}, \delta_i, \zeta_j] = 0$. We estimate $\delta_i$ and $\zeta_j$ from (\ref{eq:aag_1}) by partially linear series regression. That is, we regress $y_{ij}$ on origin and destination dummies and functions $b_{K1},\ldots,b_{KK}$ of $z_{ij}$ at dimension $K(\hat J_{\max})$. We then apply our procedures using $Y_{ij} = y_{ij} - \hat \delta_i - \hat \zeta_j$ as the dependent variable ($Y$), $\tilde \pi_{ij}$ as the endogenous regressor ($X$), and $z_{ij}$ as the instrumental variable ($W$). We present simulations below for models with and without FEs and show that this first-stage estimation of $\delta_i$ and $\zeta_j$ does not affect the performance of our procedures.
Appendix~\ref{ax:application} provides further details on implementation.

\subsection{Empirical Results}

We implement our procedures using AAG's data. Our data-driven choice of sieve dimension is $\tilde J = 4$ for this sample. Figure~\ref{fig:aag} plots our estimate of $\log \rho$ and the elasticity of the intensive margin, together with their 95\% UCBs that are constructed as in displays (\ref{band}) and (\ref{band-derivative}), respectively. We report results over the interval $[0.1\%,50\%]$, as in AAG.

UCBs for $\log \rho$ and the elasticity of $\rho$ are both narrow and informative. Figure~\ref{fig:aag} also plots a linear IV estimate of $\log \rho$ and the corresponding (constant) elasticity estimate.\footnote{For the linear IV estimates, we estimate $\log \rho$ jointly with the FEs as in AAG.} These both lie outside the UCBs for much of the support of $\pi_{ij}$. As such, our UCBs for the elasticity provide evidence against the Pareto specification for unobserved firm productivity used, e.g., by \cite{Chaney}, under which the elasticity of $\rho$ is constant. Whereas Figure 1 of AAG shows that several conventional parameterizations of the distribution of unobserved firm heterogeneity used by \cite{EKK}, \cite{HMT}, and \cite{MelitzRedding} all imply a \emph{decreasing} elasticity over $[0.1\%,50\%]$. By contrast, 
decreasing elasticities necessarily fall outside our 95\% UCBs over $[0.1\%,50\%]$,
as the right-most point of the lower UCB lies above the upper UCB for smaller values of $\tilde \pi_{ij}$.

\begin{figure}[t]
\begin{center}

\begin{subfigure}[t]{0.49\textwidth}
\begin{center}
\includegraphics[width=\linewidth]{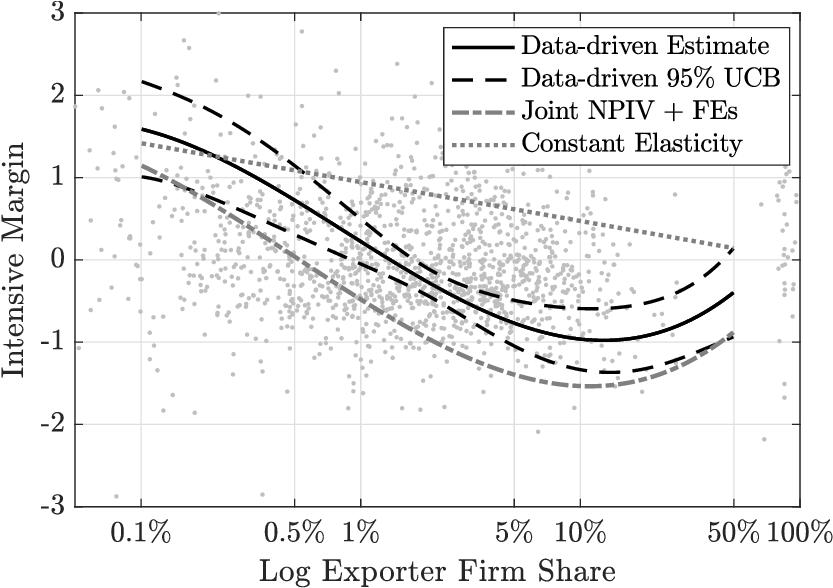}
\end{center}
\end{subfigure}
\begin{subfigure}[t]{0.49\textwidth}
\begin{center}
\includegraphics[width=\linewidth]{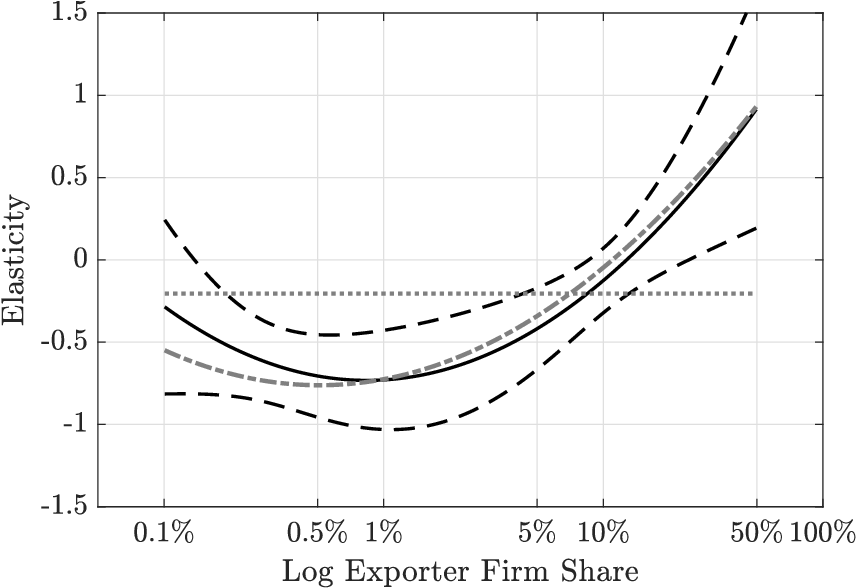}
\end{center}
\end{subfigure}

\caption{\label{fig:aag} Estimates of the intensive margin $\log \rho$ (left panel) and its elasticity (right panel) using AAG's data set ($1522$ observations). \emph{Note}: Solid black lines are estimates; dashed black lines are 95\% UCBs; dot-dash grey lines are nonparametric estimates with FEs estimated jointly with $\log \rho$ as in AAG; dotted grey lines are linear IV estimates.}

\end{center}
\vskip -20pt
\end{figure}

To show that our results are not sensitive to first-stage elimination of fixed effects, we also estimate $\log \rho$ and the FEs jointly, using our data-driven choice $\tilde J = 4$ and instrumenting with $b_{K(\tilde J)1}(z_{ij}),\ldots,b_{K(\tilde J)K(\tilde J)}(z_{ij})$ and the origin and destination dummies, and using $y_{ij}$ as the dependent variable. Estimates using this approach are also shown in Figure~\ref{fig:aag} (labeled Joint NPIV + FEs).
There is a vertical shift in the estimate of $\log \rho$ between the two approaches due to the different treatment of FEs, but the estimated elasticity---which is the focus of AAG---lies entirely within our 95\% UCB for the elasticity and is very close to our data-driven elasticity estimate over the whole range $[0.1\%,50\%]$.

\subsection{Simulation Results}

We now present simulation studies based on empirical calibrations of two workhorse trade models. The first design is based on \cite{HMT} who assume a log-normal distribution for latent firm productivity. The second design is based on \cite{Chaney} who assumes a Pareto distribution. In the first design the elasticity of $\rho$ is decreasing whereas in the second design $\log \rho(\pi) = \rho \log \pi$ and hence the elasticity is constant. For brevity we only present results for elasticity estimates in the log-normal design here. Additional results for the Pareto design and estimation of $\log \rho$ are deferred to Appendix~\ref{ax:simulations}.

We generate data by first sampling $z_{ij}$ independently with replacement from its empirical distribution. We then generate data on $\tilde \pi_{ij}$ and $\bar x_{ij}$ by simulating from equations (\ref{eq:aag.log_eps}) and (\ref{eq:aag.log_rho}), using the expressions for $\log \epsilon(\pi)$ and $\log \rho(\pi)$ implied by the log-normal assumption---see Appendix~\ref{ax:simulations}. As the empirical application has $n=1522$, we investigate the performance of our procedures across 1000 samples of size $761$, $1522$, $3044$, and $6088$.

Plots for a representative sample of size 1522 are presented in Figure~\ref{fig:trade.lognormal.deriv}(a). We generate the results in Table~\ref{tab:trade.lognormal.deriv} and Figure~\ref{fig:trade.lognormal.deriv} by implementing our procedures as in the empirical application. That is, the dependent variable is $Y_{ij} = y_{ij} - \hat \delta_i - \hat \zeta_j$, where $\hat \delta_i$ and $\hat \zeta_j$ are first-stage estimates of the exporter and importer fixed effects. We construct basis functions as in the  application; see Appendix~\ref{ax:application} for details. We also compute estimates and confidence bands over the range 0.1\% to 50\% for $\pi_{ij}$ as reported in the  application.

The first panel in Table~\ref{tab:trade.lognormal.deriv} presents the average and median (across simulations) of
\[
 \sup_{\pi \in [0.001, 0.5]} \left| \frac{d\, \widehat{\log \rho}(\pi)}{d \log \pi} - \frac{d \log \rho(\pi)}{d \log \pi} \right|,
\]
which is the maximal error of estimates of the elasticity of $\rho$ for $\pi_{ij}$ over $[0.1\%,50\%]$. We compare estimates using $\tilde J$ to estimates that use a deterministic choice of sieve dimension, namely $J = 4$, $5$, $7$, and $11$ (these are the first few values of $J$ over which our procedure searches). In each simulation, the maximal error is generally smallest with $J = 4$ or $J = 5$. The average $\tilde J$ is between 4.1 and 4.2 depending on the sample size. The maximal error of $\tilde J$ is at least half that with $J = 7$, and ten times smaller than with $J = 11$.

Turning to the coverage properties of UCBs for the elasticity, the second panel of Table~\ref{tab:trade.lognormal.deriv} shows our data-driven UCBs have correct but somewhat conservative coverage. Some conservativeness is to be expected, as our UCBs have uniform coverage guarantees over a class of DGPs. We also present coverage of UCBs based on the usual approach of ``undersmoothing'' from Section~\ref{sec:npiv}. These UCBs use a deterministic $J$ and have valid coverage provided $J$ is chosen sufficiently large that bias is negligible relative to sampling uncertainty. Of course, in any empirical application a researcher does not know the true function, and therefore doesn't know which values of $J$ are sufficiently large that sampling uncertainty dominates bias. As can be seen from Table~\ref{tab:trade.lognormal.deriv}, $J = 4$ or $J = 5$ seems too small, and consequently these bands under-cover. Bands with $J = 7$ have coverage closer to nominal coverage, but these bands are more than 70\% wider 
than the data-driven bands. Comparing the UCBs in Figures~\ref{fig:trade.lognormal.deriv}(a) and~\ref{fig:trade.lognormal.deriv}(c), we see the efficiency improvement of our bands relative to undersmoothed bands with $J = 7$, for estimating both $\rho$ and its elasticity.

\begin{table}[t]
\begin{center}
\caption{\label{tab:trade.lognormal.deriv} Simulation Results for the Elasticity of $\rho$, Log-normal Design}
{\small
\begin{tabular}{ccccccccccccc} \hline \hline \\[-10pt]
  & & \multicolumn{2}{c}{Data-driven}  & & \multicolumn{8}{c}{Deterministic} \\\cline{3-4} \cline{6-13} \\[-10pt]
  & & & & & \multicolumn{2}{c}{$J = 4$} & \multicolumn{2}{c}{$J = 5$} & \multicolumn{2}{c}{$J = 7$} & \multicolumn{2}{c}{$J = 11$} \\ \\[-10pt] \hline \\[-10pt]
  \multicolumn{13}{c}{Sup-norm Loss} \\ \\[-10pt]
$n$ & & mean & med. & & mean & med. & mean & med. & mean & med. & mean & med. \\ \\[-10pt]
 \phantom{0}761 & & 0.268 & 0.187 & & 0.207 & 0.178 & 0.314 & 0.281 & 0.579 & 0.472 & 2.063 & 1.902 \\
           1522 & & 0.184 & 0.129 & & 0.144 & 0.125 & 0.216 & 0.191 & 0.382 & 0.339 & 1.823 & 1.650 \\
           3044 & & 0.143 & 0.099 & & 0.106 & 0.095 & 0.149 & 0.139 & 0.283 & 0.254 & 1.562 & 1.385 \\
           6088 & & 0.111 & 0.071 & & 0.076 & 0.068 & 0.105 & 0.096 & 0.202 & 0.185 & 1.367 & 1.218 \\  \\[-10pt] \hline \\[-10pt]
  \multicolumn{13}{c}{UCB Coverage} \\ \\[-10pt]
  & & 90\% & 95\% & & 90\% & 95\% & 90\% & 95\% & 90\% & 95\% & 90\% & 95\% \\ \\[-10pt]
 \phantom{0}761 & & 0.989 & 0.997 & & 0.861 & 0.921 & 0.841 & 0.911 & 0.871 & 0.930 & 0.906 & 0.965 \\
           1522 & & 0.994 & 0.997 & & 0.872 & 0.924 & 0.857 & 0.921 & 0.889 & 0.936 & 0.940 & 0.976 \\
           3044 & & 0.993 & 0.998 & & 0.833 & 0.899 & 0.869 & 0.929 & 0.899 & 0.943 & 0.947 & 0.979 \\
           6088 & & 0.993 & 0.994 & & 0.800 & 0.890 & 0.868 & 0.936 & 0.899 & 0.952 & 0.949 & 0.982 \\ \\[-10pt] \hline \\[-10pt]
  & & \multicolumn{2}{c}{Frequency} & & \multicolumn{8}{c}{95\% UCB Relative Width (Deterministic/Data-driven)} \\\cline{3-4} \cline{6-13} \\[-10pt]
  & &  \multicolumn{2}{c}{reject} & & mean & med. & mean & med. & mean & med. & mean & med. \\ \\[-10pt]
 \phantom{0}761 & & \multicolumn{2}{c}{0.088} & & 0.624 & 0.651 & 0.922 & 0.932 & 1.750 & 1.568 & \phantom{1}6.398 & \phantom{1}6.140 \\
           1522 & & \multicolumn{2}{c}{0.344} & & 0.632 & 0.657 & 0.906 & 0.911 & 1.739 & 1.599 & \phantom{1}8.295 & \phantom{1}8.098 \\
           3044 & & \multicolumn{2}{c}{0.822} & & 0.638 & 0.657 & 0.888 & 0.902 & 1.746 & 1.625 & 10.340 & 10.071 \\
           6088 & & \multicolumn{2}{c}{0.959} & & 0.634 & 0.660 & 0.865 & 0.893 & 1.722 & 1.690 & 12.783 & 12.665 \\ \\[-10pt] \hline
\end{tabular}
}
\parbox{\textwidth}{\small \emph{Note:}
Column ``reject'' reports the proportion of simulations in which constant functions are excluded from data-driven 95\% UCBs for the  elasticity.}
\end{center}
\vskip -20pt
\end{table}

\begin{figure}[h!]
\begin{center}

\begin{subfigure}[t]{\textwidth}
\caption{Data-driven Estimates and UCBs}
\begin{center}
\vskip -10pt
\begin{subfigure}[t]{0.49\textwidth}
\begin{center}
\includegraphics[width=\linewidth]{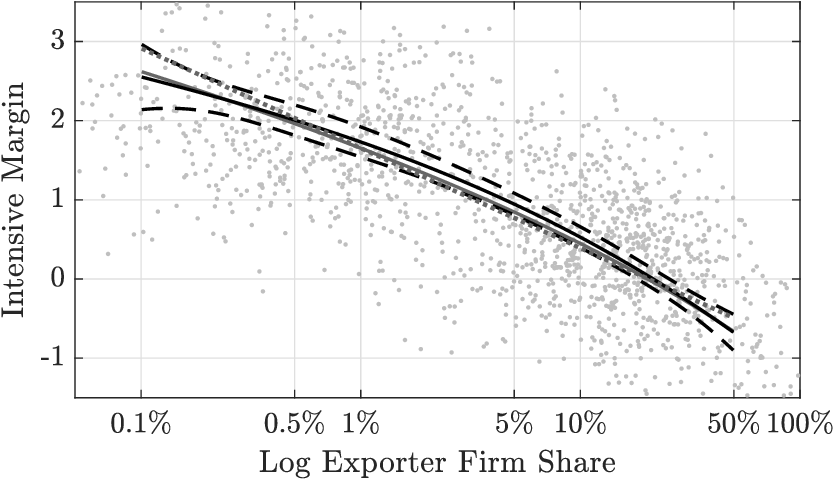}
\end{center}
\end{subfigure}
\begin{subfigure}[t]{0.49\textwidth}
\begin{center}
\includegraphics[width=\linewidth]{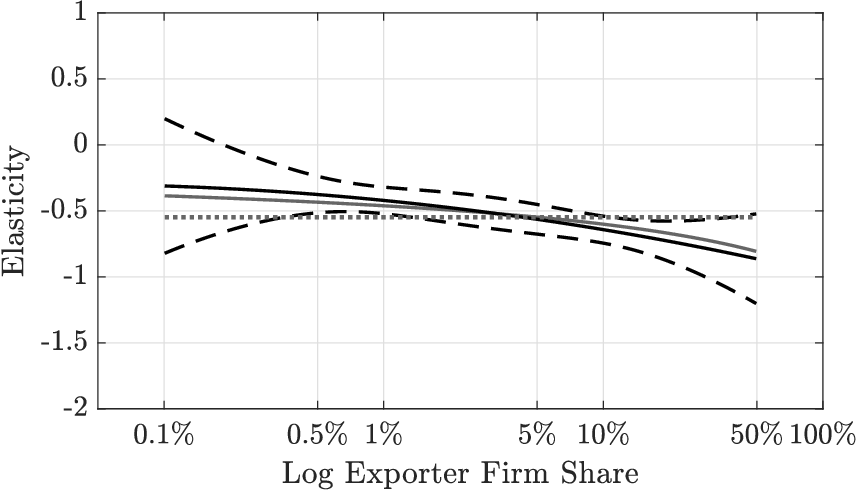}
\end{center}
\end{subfigure}

\end{center}
\end{subfigure}

\begin{subfigure}[t]{\textwidth}
\caption{Estimates and UCBs with $J = 5$}
\begin{center}
\vskip -10pt
\begin{subfigure}[t]{0.49\textwidth}
\begin{center}
\includegraphics[width=\linewidth]{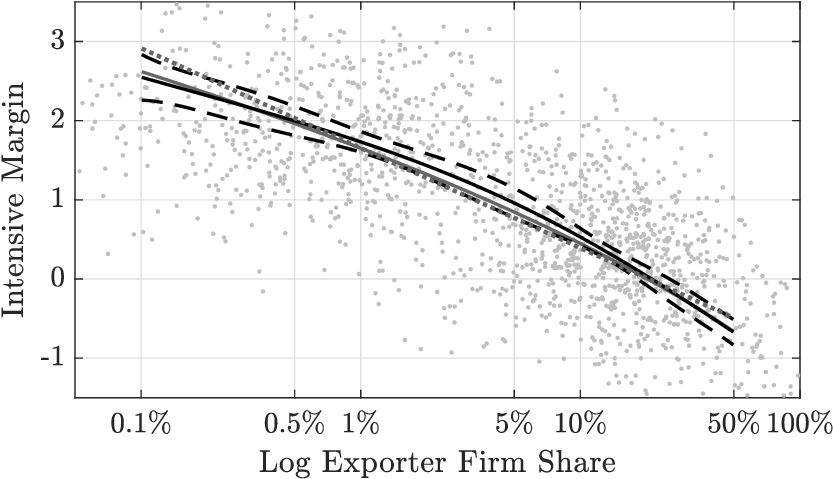}
\end{center}
\end{subfigure}
\begin{subfigure}[t]{0.49\textwidth}
\begin{center}
\includegraphics[width=\linewidth]{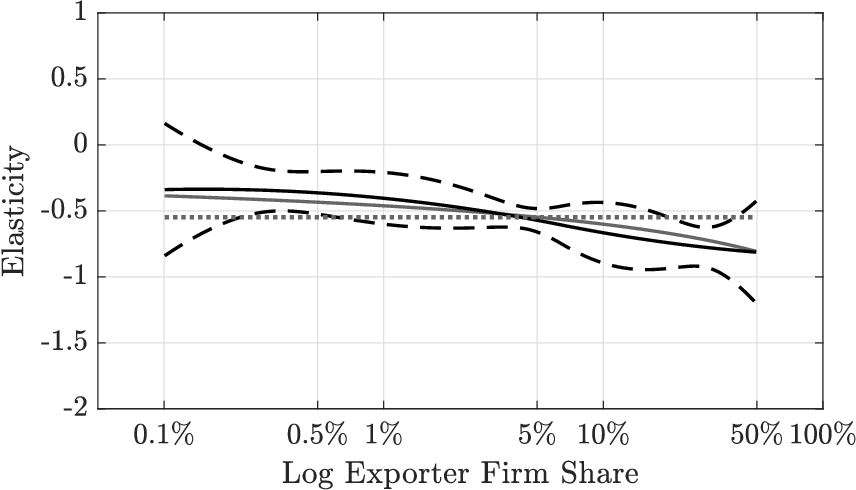}
\end{center}
\end{subfigure}

\end{center}
\end{subfigure}

\begin{subfigure}[t]{\textwidth}
\caption{Estimates and UCBs with $J = 7$}
\begin{center}
\vskip -10pt
\begin{subfigure}[t]{0.49\textwidth}
\begin{center}
\includegraphics[width=\linewidth]{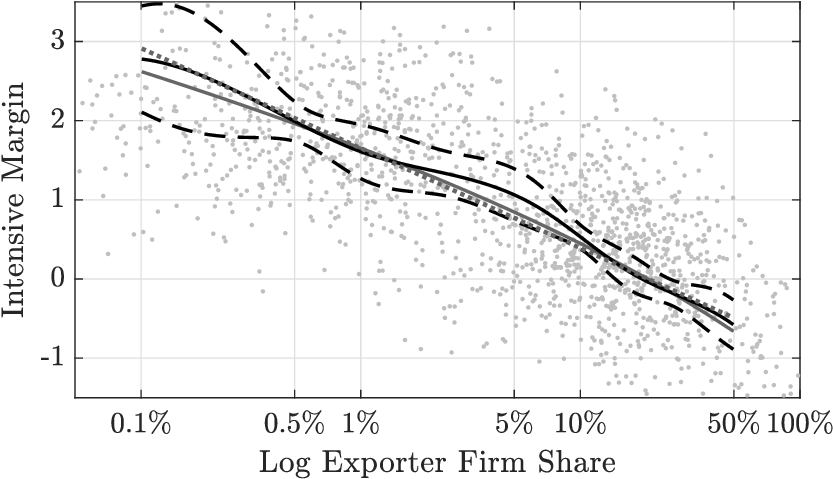}
\end{center}
\end{subfigure}
\begin{subfigure}[t]{0.49\textwidth}
\begin{center}
\includegraphics[width=\linewidth]{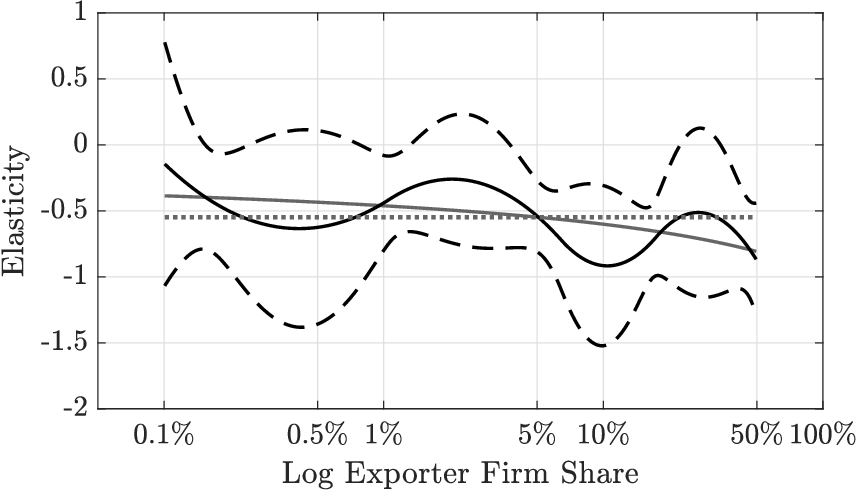}
\end{center}
\end{subfigure}

\end{center}
\end{subfigure}

\vskip -6pt

\caption{\label{fig:trade.lognormal.deriv} {Log-normal design: Plots for a representative sample of size $1522$. Left panels correspond to the intensive margin, right panels correspond to its elasticity. \emph{Note:} Solid grey lines are the true curves; solid black lines are estimates; dashed black lines are 95\% UCBs; dotted grey lines are linear IV estimates.}}
\end{center}
\vskip -20pt
\end{figure}

The fact that our UCBs are based on an optimal choice of $J$, and therefore contract faster than bands based on undersmoothing, has important practical consequences. Consider the data-driven UCBs for the elasticity of $\rho$ reported in Figure~\ref{fig:trade.lognormal.deriv}(a). These bands do not contain any constant function because the upper limit of the lower band exceeds the lower limit of the upper band. This provides evidence against the Pareto specification of productivity used by \cite{Chaney}, for which the elasticity of $\rho$ is constant.\footnote{Table~\ref{tab:trade.lognormal.deriv} presents the frequency that such a test rejects the constant elasticity specification.} Note this is despite the fact that our bands tend to be a bit conservative. The undersmoothed bands with $J = 7$ have coverage closer to nominal coverage. But for the sample shown in Figure~\ref{fig:trade.lognormal.deriv}, the undersmoothed bands with $J = 7$ are sufficiently wide that constant functions lie entirely within the bands. Hence, the researcher could not reject a constant elasticity specification on the basis of the undersmoothed bands in this sample. In fact, the undersmoothed bands with $J = 7$ only reject the constant elasticity specification in 15.8\% of simulations with $1522$ observations whereas the rejection rate for the data-driven bands is 34.4\%. This difference in rejection rates illustrates the general phenomenon that undersmoothed bands sacrifice efficiency for coverage. The undersmoothed bands are also quite wiggly, making it difficult to infer the shape of the true elasticity.

We note in closing that our procedures can equally be applied to other IV-based nonparametric analyses in international trade; see, e.g., \cite{ACD}.

\section{Theory}\label{s:theory}

We first outline the main regularity conditions in Section~\ref{sec:4assumptions}. Section~\ref{sec:4estimation} shows that $\tilde J$ leads to minimax convergence rates for estimators of both $h_0$ and its derivatives. We then present the main results for UCBs in Sections~\ref{sec:4UCB} and \ref{sec:4UCBde}.

\subsection{Assumptions}\label{sec:4assumptions}

We first state and then discuss the assumptions that we impose on the model and sieve space. We require these to hold for some constants $a_f, \ul c, \ol C, C_T, C_Q, \underline{\sigma}, \overline{\sigma} > 0$ and $\gamma \in (0,1)$.
Let $T : L^2_X \to L^2_W$ denote the operator $Th(w) = \E[h(X)|W = w]$. For nonparametric regression  we have $W \equiv X$ and so $T$ reduces to the identity.

\begin{assumption} \label{a-data}
(i)~$X$ has  support $ \mathcal{X} = [0,1]^d$ and its distribution has Lebesgue density $f_X$ which satisfies $a_{f}^{-1} < f_X(x) < a_{f}$ on $\mathcal X$; (ii)~$W$ has support $\mathcal{W} = [0,1]^{d_w} $ and its distribution has Lebesgue density $f_W$ which satisfies $a_{f}^{-1} < f_W(w) < a_f$ on $\mathcal W$; (iii)~$T$ is injective.
\end{assumption}

\begin{assumption} \label{a-residuals}
(i)~$\mathbb{P} \big( \E[u^4|W] \leq \overline \sigma^2  \big)  = 1$;
(ii)~$\mathbb{P} \big( \E[u^2 |W] \geq \underline \sigma^2 \big) = 1$.
\end{assumption}

Let $ \Psi_J$ and $ B_K$ be the closed linear subspaces of $L^2_X$ and $L^2_W$ spanned by $\psi_{J1},\ldots,\psi_{JJ}$ and $b_{K1},\ldots,b_{KK}$, respectively. Define
\[
 \tau_J = \sup_{h \in \Psi_J : \| h \|_{L^2_X} \neq 0} \frac{\|h \|_{L^2_X}}{\|  Th \|_{L^2_W}}\,,
\]
where $\|\cdot \|_{L^2_X}$ and $\|\cdot\|_{L^2_W}$ denote the $L^2_X$ and $L^2_W$ norms. The \emph{sieve measure of ill-posedness} $\tau_J$ quantifies the degree of difficulty of inverting $Th_0$ to recover $h_0$. As conditional expectations are (weakly) contractive, we have $\tau_J \geq 1$. Large 
$\tau_J$ indicate a more difficult inversion problem. The model (\ref{eq:npiv}) is said to be \emph{mildly ill-posed} (or in the \emph{mild} regime) if $\tau_J \asymp J^{\varsigma/d}$ for some $\varsigma \geq 0$  and \emph{severely ill-posed} (or in the \emph{severe} regime) if $\tau_J \asymp \exp(C J^{\varsigma/d}) $ for some $C,\varsigma > 0$, where $d = \dim(X)$. For nonparametric regression models we have $\tau_J = 1$ for all $J$. Hence, nonparametric regression is a special case of the mild regime with $\varsigma = 0$.

Let $\Pi_J : L^2_X \to \Psi_J$ and $\Pi_{K(J)} : L^2_W \to B_{K(J)}$ denote LS projections onto $\Psi_J$ and $B_{K(J)}$:
\[
\begin{aligned}
 \Pi_J f & = \mathrm{arg}\min_{g \in \Psi_J} \|f - g\|_{L^2_X} \,, &
 \Pi_{K(J)} f & = \mathrm{arg} \min_{g \in B_{K(J)}} \|f - g\|_{L^2_W} \,.
\end{aligned}
\]
Also let $Q_J : L^2_X \to \Psi_J$ denote the TSLS projection onto $\Psi_J$:
\begin{align*}
  Q_J f & = \mathrm{arg}\min_{h \in \Psi_J} \|\Pi_{K(J)}T(f - h)\|_{L^2_W}.
\end{align*}

\begin{assumption}\label{a-approx} (i)~$ \sup_{h \in \Psi_{J}, \|h \|_{L^2_X}=1} \tau_J \|\Pi_{K(J)} Th - T h\|_{L^2_W} \leq v_J$ where $v_J < 1$ for all $J \in \T$ and $v_J \to 0$ as $J \to \infty$; \\
(ii)~$\tau_J   \|T(h_0 - \Pi_J h_0)\|_{L^2_W} \leq C_T \|h_0 - \Pi_J h_0\|_{L^2_X}$ for all $ J \in \T$; \\
(iii)~$\|Q_J (h_0 - \Pi_J h_0) \|_\infty \leq C_Q \|h_0 - \Pi_J h_0 \|_{\infty}$ for all $ J \in \T$.
\end{assumption}

Denote the ``population'' sieve variance of $\hat h_J(x)$ as
 $\|{\sigma}_{x,J}\|^2_{sd} = L_{J,x}^{\phantom \prime} \Omega_J^{\phantom \prime} L_{J,x}'$ where $L_{J,x} =  (\psi^J(x))' [S_J' G_{b,J}^{-1} S_J^{\phantom \prime} ]^{-1} S_J' G_{b,J}^{-1}$ and $\Omega_J =  \E [ u^2 b^{K(J)}(W) (b^{K(J)}(W))' ]$ with $u = Y - h_0(X)$, $G_{b,J} = \E [ b^{K(J)}(W) (b^{K(J)}(W))' ]$, and $S_J = \E [ b^{K(J)}(W) (\psi^J(X))' ]$. Also let $\|\sigma_{x,J}\|^2 =  (\psi^J(x))'[S_J' G_{b,J}^{-1} S_J^{\phantom \prime}]^{-1} (\psi^J(x))$, which satisfies $\|\sigma_{x,J} \| \asymp \|{\sigma}_{x,J}\|_{sd}$ uniformly in $x$ by Assumption~\ref{a-residuals}.

\begin{assumption}\label{a-var}
(i)~$\underline{c} \tau_J^2 J \leq  \inf_{x \in \mathcal{X}} \|  \sigma_{x,J}  \|^2 \leq  \sup_{x \in \mathcal{X}} \| \sigma_{x,J} \|^2 \leq \overline{C} \tau_J^2 J$ for all $J \in \T$; \\
(ii) $\limsup_{J \to \infty} \sup_{x \in \mc X,J_2 \in \T : J_2 > J } ( \|{\sigma}_{x,J}\|_{sd}  / \|{\sigma}_{x,J_2}\|_{sd} ) < \gamma$.
\end{assumption}

Assumptions~\ref{a-data}(i)(ii) and~\ref{a-residuals} are standard conditions on the support of $X$ and $W$ and the conditional variance of the errors (see, e.g.,  \cite{chen2018optimal}) that can be relaxed. Assumption~\ref{a-data}(iii) is an identification condition that is generically satisfied under endogeneity (see \cite{andrews2017examples}) and is trivially satisfied for nonparametric regression because $T$ reduces to the identity in that case. Assumption~\ref{a-approx} is also trivially satisfied for nonparametric regression with $C_T,C_Q = 1$. Assumption~\ref{a-approx}(i) is imposed to ensure that $\hat s_J^{-1}$ is a suitable sample analog of $\tau_J$.
Assumption~\ref{a-approx}(ii) is the usual $L^2$ ``stability condition'' imposed in the NPIV literature to derive $L^2$-norm rates. Assumption~\ref{a-approx}(iii) is a  $L^{\infty}$-norm analogue used to control the bias in sup-norm. \cite{chen2018optimal} provide a thorough discussion of Assumption~\ref{a-var}(i) and derive primitive sufficient conditions for it in the context of nonparametric demand estimation.
Assumption \ref{a-var}(ii) says that $\|{\sigma}_{x,J}\|^2_{sd}$ is increasing in $J \in \T$, uniformly in $x$. We view this as mild because $J$ increases exponentially over $\T$.
Indeed, by Assumption~\ref{a-residuals} and \ref{a-var}(i) and the fact that $J \asymp 2^{Ld}$ for some $L \in \mathbb{N}$, for any $J,J_2 \in \T$ with $J_2 > J$ we have
\[
 \sup_{x \in \mc X} \frac{\|{\sigma}_{x,J}\|_{sd}}{\|{\sigma}_{x,J_2}\|_{sd}} \asymp \frac{ \tau_J \sqrt{J}}{\tau_{J_2} \sqrt{J_2}} \leq \frac{\tau_{2^{Ld}}}{\tau_{2^{(L+1)d}}}  2^{-d/2}\leq 2^{-d/2}<1\,.
\]

\subsection{Main Results: Adaptive Estimation in Sup-norm}\label{sec:4estimation}

We now show
$\tilde J$ leads to minimax rate-adaptive estimators of both the structural function $h_0$ and its derivatives. Our results encompass nonparametric regression as a special case.

We first define the parameter space for $h_0$. Let $B_{\infty,\infty}^p(M)$ denote the H\"older ball of smoothness $p$ and radius $M$ (see Appendix \ref{ax:besov} for a formal definition). For given constants $C_T,C_Q,M > 0 $ and $\ol p > \ul p > \frac d2$ with $r \geq \lfloor \ol p \rfloor + 1$, let  $ \mathcal{H}^p = \mathcal{H}^p(M,C_T,C_Q)$ denote the subset of $B_{\infty,\infty}^p (M)$ that satisfies Assumption \ref{a-approx}(ii)(iii) for any distribution of $(X,W,u)$ satisfying Assumptions \ref{a-data}-\ref{a-var}, and let $\mc H = \bigcup_{p \in [\ul p,\ol p]} \mc H^p$.
For each $h_0 \in \mc H$, we let $\mb P_{h_0}$ denote the distribution of $(X_i,Y_i,W_i)_{i=1}^\infty$ where each observation is generated by an IID draw from a distribution of $(X,W,u)$ satisfying Assumptions \ref{a-data}-\ref{a-var} with $Y = h_0(X) + u$.

\begin{theorem}\label{lepski2}
Let Assumptions \ref{a-data}-\ref{a-var} hold.
\begin{enumerate}[nosep]
\item[(i)] Suppose the model is mildly ill-posed. Then: there is a universal constant $C_{\ref{lepski2}}$ for which
\[
 \sup_{p \in [\ul p,\ol p]} \sup_{h_0 \in \mc H^p} \mathbb{P}_{h_0} \bigg( \| \hat{h}_{\tilde{J}} - h_0   \|_{\infty} > C_{\ref{lepski2}} \bigg( \frac{\log n}{n} \bigg)^{\frac{p}{2(p + \varsigma) + d}} \bigg) \rightarrow 0.
\]
\item[(ii)] Suppose the model is severely ill-posed.  Then: there is a universal constant $C_{\ref{lepski2}}$ for which
\[
 \sup_{p \in [\ul p,\ol p]} \sup_{h_0 \in \mc H^p} \mathbb{P}_{h_0} \big(  \| \hat{h}_{\tilde{J}} - h_0    \|_{\infty} > C_{\ref{lepski2}} (\log n)^{-p / \varsigma}   \big) \rightarrow 0.
\]
\end{enumerate}
\end{theorem}

We now show $\tilde J$ also leads to adaptive estimation of derivatives of $h_0$. Intuitively, estimating the derivative of $h_0$ inflates convergence rate of the (squared) bias and variance terms by the same factor (a power of $J$). Therefore, a rate-optimal choice of $J$ for estimating $h_0$ is also rate-optimal for estimating derivatives of $h_0$.

\begin{corollary}\label{cor:lepski2}
Let Assumptions \ref{a-data}-\ref{a-var} hold and let $a \in (\mb N_0)^d$ with $0 < |a| < \ul p$.
\begin{enumerate}[nosep]
\item[(i)] Suppose the model is mildly ill-posed. Then: there is a universal constant $C_{\ref{lepski2}}'$ for which
\[
 \sup_{p \in [\ul p,\ol p]} \sup_{h_0 \in \mc H^p} \mathbb{P}_{h_0} \bigg( \| \partial^a \hat{h}_{\tilde{J}} - \partial^a h_0   \|_{\infty} > C_{\ref{lepski2}}' \bigg( \frac{\log n}{n} \bigg)^{\frac{p-|a|}{2(p + \varsigma) + d}} \bigg) \rightarrow 0.
\]
\item[(ii)] Suppose the model is severely ill-posed.  Then: there is a universal constant $C_{\ref{lepski2}}'$ for which
\[
 \sup_{p \in [\ul p,\ol p]} \sup_{h_0 \in \mc H^p} \mathbb{P}_{h_0} \big(  \| \partial^a \hat{h}_{\tilde{J}} - \partial^a h_0  \|_{\infty} > C_{\ref{lepski2}}' (\log n)^{-(p-|a|) / \varsigma}   \big) \rightarrow 0.
\]
\end{enumerate}
\end{corollary}

\begin{remark} \normalfont
The convergence rates in Theorem~\ref{lepski2} and Corollary~\ref{cor:lepski2} are the minimax rates for estimating $h_0$ and $\partial^a h_0$ under sup-norm loss; see \cite{chen2018optimal}. Hence, $\hat h_{\tilde J}$ and $\partial^a \hat h_{\tilde J}$ converge at the minimax rate in both the mildly and severely ill-posed cases. Case (i) encompasses nonparametric regression as a special case with $\varsigma = 0$. To the best of our knowledge, Theorem \ref{lepski2} and Corollary \ref{cor:lepski2} are the first results on adaptive estimation in sup-norm for NPIV and, more generally, ill-posed inverse problems with unknown operator.
\end{remark}

\begin{remark} \normalfont
Our procedure requires the B-spline order $r$ to satisfy $r \geq \lfloor p \rfloor + 1$ for exact minimax rate adaptivity.
If the true $p$ is larger so that $r < \lfloor p \rfloor + 1$, then our method is still ``adaptive'' in the sense that it yields consistent estimates of $h_0$ and its derivatives without requiring prior knowledge of the true smoothness of $h_0$ or the strength of the instruments. In this case the data-driven estimators $\hat h_{\tilde J}$ and $\partial^a \hat h_{\tilde J}$ will converge at the rates presented in Theorem~\ref{lepski2} and Corollary~\ref{cor:lepski2} with $p = r$. 
Thus, our procedure should be attractive to applied researchers who often use a relatively low choice of $r$ in applications. For instance, \cite{ABB} use linear splines ($r = 2$). While in principle our method could be extended to let $r$ become large, known results from approximation theory imply that the basis becomes ill-conditioned (i.e., collinear) as $r$ increases (see, e.g., \cite{Lyche1978} and \cite{Scherer1999}). As a consequence, the resulting procedure would be less numerically stable than with smaller $r$.
\end{remark}

\subsection{Main Results: UCBs for $h_0$}\label{sec:4UCB}

It is known since \cite{low1997} that it is impossible to construct confidence bands that are simultaneously honest and adaptive over H\"older classes of different smoothness. As is standard following \cite{picard2000}, \cite{gine2010confidence}, \cite{bull2012honest}, \cite{chernozhukov2014anti}, and many others, we establish coverage guarantees over a ``generic'' subclass $\mc G$ of $\mc H$.
To describe $\mc G$, first note by the discussion in Appendix \ref{ax:besov} that there exists a constant $\ol B < \infty$ for which $\sup_{h \in \mc H^p} \|h - \Pi_J h\|_\infty \leq \ol B J^{- \frac pd}$ holds for all $J \in \mc T$ and all $p \in [\ul p, \ol p]$. For any small fixed $\ul B \in (0,\ol B)$ and any $\ul J \in \mc T$, we  define
\[
 \mathcal{G}^p = \left\{ h \in \mc H^p : \ul B J^{- \frac pd} \leq \|h - \Pi_J h\|_\infty  \mbox{ for all $J \in \T$ with $J \geq \ul J$}  \right\},~~~\mc G = \bigcup_{p \in [\ul p, \ol p]} \mc G^p\,.
\]
The class $\mc G$ is sometimes called a class of ``self-similar'' functions.
\cite{gine2010confidence,gine2016mathematical} present several results establishing the genericity of $\mc G$ in $\mc H$. Loosely speaking, their results say $\mc H^p \setminus (\cup_{\ul B > 0,\ul J \in \mc T} \mc G^p)$ is nowhere dense in $\mc H^p$ under the norm topology of $\mc H^p$. Thus, the set of functions in $\mc H^p$ but not in $\mc G^p$ for some $\ul B$ and $\ul J$ is topologically meagre. 

We say that a UCB $\{C_n(x) : x \in \mc X\}$ is \emph{honest} over $\mc G$ with level $\alpha$ if
\begin{equation} \label{eq:honest}
 \liminf_{n \to \infty}  \inf_{h_0 \in \mathcal{G}} \mathbb{P}_{h_0} \big(  h_0(x) \in C_n(x) \; \; \; \; \forall \; \; x \in \mathcal{X} \big) \geq 1 - \alpha \,,
\end{equation}
and \emph{adaptive} if for every $\epsilon > 0$ there exists a constant $D$ for which
\[
 \liminf_{n \to \infty} \inf_{p \in [\ul p,\ol p]}  \inf_{h_0 \in \mc G^p} \mathbb{P}_{h_0} \bigg( \sup_{x \in \mathcal{X}} |C_n(x)| \leq D r_n(p) \bigg) \geq 1 - \epsilon \,,
\]
where $|\,\cdot\,|$ is Lebesgue measure and $r_n(p)$ is the minimax sup-norm rate of estimation over $\mc H^p$. Let $C_n(x,A)$ denote the UCB from (\ref{band}) replacing $\hat A$ with a constant $A > 0$. Our first main result is that $C_n(x,A)$ is honest and adaptive in the mildly ill-posed case:

\begin{theorem}\label{confmild}
Let Assumptions \ref{a-data}-\ref{a-var} hold and suppose the model is mildly ill-posed. Then: there is a constant $A^* > 0$ (independent of $\alpha$) such that for all $A \geq A^*$, 
\begin{flalign*}
 \mbox{\textit{(i)}} & ~~~~\liminf_{n \rightarrow \infty} \inf_{h_0 \in \mathcal{G}} \mathbb{P}_{h_0} \big(  h_0(x) \in C_n(x,A) \; \; \; \; \forall \; \; x \in \mathcal{X} \big) \geq 1 - \alpha \,; && \\
 \mbox{\textit{(ii)}} & ~~~~ \inf_{p \in [\ul p,\ol p]}  \inf_{h_0 \in \mc G^p} \mathbb{P}_{h_0} \bigg( \sup_{x \in \mathcal{X}} |C_n(x,A)| \leq C_{\ref{confmild}} (1+A) \bigg(  \frac{\log n}{n} \bigg)^{\frac{p}{2(p + \varsigma) + d}} \bigg) \rightarrow 1 \,, &&
\end{flalign*}
where  $C_{\ref{confmild}}>0$ is a universal constant.
\end{theorem}

\begin{remark}\normalfont
Theorem \ref{confmild} shows that our UCBs are honest and adaptive in mildly ill-posed models (where $\tau_J \asymp J^{\varsigma/d}$) for all $\varsigma \geq 0$. Importantly, the researcher doesn't need to know the true instrument strength as measured by $\varsigma$ to implement our procedures.
\end{remark}

\begin{remark} \normalfont
As the mildly ill-posed case nests nonparametric regression as a special case with $\varsigma = 0$, Theorem \ref{confmild} shows that our UCBs are honest and adaptive for general nonparametric regression models with non-Gaussian, heteroskedastic errors.
\end{remark}

\begin{remark} \normalfont
The constant $A^*$ in Theorem \ref{confmild} depends implicitly on $\underline B$ and becomes large as $\underline B \downarrow 0$, coherent with the findings of \cite{armstrong2021adaptation} for Gaussian white noise models. This constant cannot be chosen in a data-dependent way (i.e., one cannot adapt to unknown $\underline B$). In practice, $A$ can actually be quite small to guarantee coverage for a fixed DGP---see the simulations in Section \ref{s:mcs}. The UCBs in Section~\ref{s:procedure} replace a fixed constant $A$ by $\hat A = \log \log \tilde J$, which increases no faster than $\log \log n$. These UCBs therefore have coverage guarantees over $\mc G$ defined for any small $\underline B > 0$ and contract  within a $\log \log n$ factor of the minimax rate.
\end{remark}

Theorem~\ref{confmild} establishes that the UCBs for $h_0$ in Procedure~2 is honest and adaptive in the mildly ill-posed case. We have found that the UCBs in Procedure~2 perform well in terms of coverage across many simulation designs including the severely ill-posed design in Appendix~\ref{appsec:engel}.
Nevertheless, for the severely ill-posed case, we can only establish valid coverage of the UCBs in Procedure~2 using the critical value $\mbox{cv}^*(x)$ corresponding to $\tilde J = \hat J_n$ case, i.e.,
\begin{equation} \label{eq:band-cv-2}
 \mbox{cv}^*(x) =
 z_{1-\alpha}^* + \hat A \max\{\theta^*_{1-\hat \alpha}\,,\, \tilde {J}^{-\ul p/d} /\hat \sigma_{\tilde J}(x) \}\,.
\end{equation}
The term $\tilde {J}^{-\ul p/d}$ bounds
the order of the bias term $\|\Pi_{\tilde J} h_0 - h_0  \|_{\infty}$, which accounts for the fact that the optimal choice of $J$ in severely ill-posed models is bias-dominating. This band reduces to the Procedure~2 UCB when $\theta^*_{1-\hat \alpha} \geq \tilde {J}^{-\ul p/d} /\hat \sigma_{\tilde J}(x)$ for all $x$.

\begin{remark}\normalfont
In our empirical application to estimating the intensive margin and its elasticity, the UCB (\ref{band}) using critical value (\ref{eq:band-cv-2}) reduces to the Procedure~2 band provided $\underline p \geq 0.7$. The condition $\underline p \geq 0.7$ is naturally satisfied as $h_0$ is assumed to be differentiable in order to estimate the elasticity.
\end{remark}

Let $C_n(x,A)$ denote the UCB (\ref{band}) with the critical value (\ref{eq:band-cv-2}), except replacing $\hat A$ with a constant $A > 0$.

\begin{theorem}\label{confsevere}
Let Assumptions~\ref{a-data}-\ref{a-var} hold and suppose the model is severely ill-posed. Then: there is a constant $A^* > 0$ (independent of $\alpha$) such that for all $A \geq A^*$, 
\begin{flalign*}
 \mbox{\textit{(i)}} & ~~~~\liminf_{n \rightarrow \infty} \inf_{h_0 \in \mathcal{G}} \mathbb{P}_{h_0} \big(  h_0(x) \in C_n(x,A) \; \; \; \; \forall \; \; x \in \mathcal{X} \big) \geq 1 - \alpha \,; && \\
 \mbox{\textit{(ii)}} & ~~~~ \inf_{p \in [\ul p,\ol p]}  \inf_{h_0 \in \mc G^p} \mathbb{P}_{h_0} \bigg( \sup_{x \in \mathcal{X}} |C_n(x,A)| \leq  C_{\ref{confsevere}} (1+A) (\log n)^{-\ul p/\varsigma} \bigg) \rightarrow 1\,, &&
\end{flalign*}
where $C_{\ref{confsevere}}>0$ is a universal constant.
\end{theorem}

Our recommended choice $\hat A =  \log \log \tilde J$ ensures that the UCBs are asymptotically valid over $\mc G$ for any $\underline B > 0$ and contract within a $\log \log n$ factor of the minimax rate if the true smoothness is $p = \ul p$, and within a $\log n$ factor of the minimax rate otherwise.

\begin{remark}\label{rmk:impossibility} \normalfont
If the true $p > \ul p $, then the factor $\tilde{J}^{- \ul p /d}$ is conservative and the UCB does not contract at the minimax rate. This raises the question as to whether it is possible to construct UCBs that are adaptive in severely ill-posed settings. As stated in Chapter $8.3$ of \cite{gine2016mathematical}, the existence of rate-adaptive UCBs implicitly requires the estimation of certain aspects of the unknown function, e.g. smoothness, to be feasible. In mildly ill-posed settings, the condition $h_0 \in \mathcal{G}^p$ is sufficient to ensure that $\tilde{J}$ diverges at the oracle rate $J_{0} \asymp (n/ \log n)^{d/(2(p + \varsigma) + d)}$. As it turns out, $\tilde J$ is sufficiently informative about the unknown smoothness $p$ to facilitate the construction of adaptive UCBs. In severely ill-posed models the oracle choice is $J_{0} =  (a \log n)^{d/\varsigma}$ for $0 < a < (2C)^{-1}$, which is independent of $p$. Therefore, the adaptivity of $\tilde J$ cannot be used to ascertain information about $p$. We conjecture that any UCB that is centered around an adaptive estimator that aims to mimic the oracle $ \hat{h}_{J_{0}} $ 
will likely face the same ``identifiability'' problem of recovering information about $p$ from $J_{0}$.
\end{remark}

\subsection{Main Results: UCBs for Derivatives}\label{sec:4UCBde}

We now present an analogous set of results for data-driven UCBs for derivatives of $h_0$. Here we require an additional regularity condition similar to Assumption \ref{a-var}(i), which is only needed for the results in this subsection. Let  $\|\sigma^a_{x,J}\|^2 =  (\partial^a \psi^J(x))'[S_J' G_{b,J}^{-1} S_J^{\phantom \prime}]^{-1} (\partial^a \psi^J(x))$.

\setcounter{assumption}{3}

\begin{assumption}[continued]
(iii)~There exist constants $\ul c, \ol C > 0$ for which $\underline{c} \tau_J^2 J^{1+2|a|/d} \leq  \inf_{x \in \mathcal{X}} \|  \sigma^a_{x,J}  \|^2 \leq  \sup_{x \in \mathcal{X}} \| \sigma^a_{x,J} \|^2 \leq \overline{C} \tau_J^2 J^{1+2|a|/d}$ for all $J \in \T$.
\end{assumption}

We first present results for the mildly ill-posed case. Let $C_n^a(x,A)$ denote the UCB $C_n^a(x)$ from (\ref{band-derivative}) when $\hat A$ is replaced by a constant $A > 0$.

\begin{theorem}\label{confmild-derivative}
Let Assumptions \ref{a-data}-\ref{a-var} hold, $|a| < \ul p$, and suppose the model is mildly ill-posed. Then: there is a constant $A^* > 0$ (independent of $\alpha$) such that for all $A \geq A^*$,
\begin{flalign*}
 \mbox{\textit{(i)}} & ~~~~\liminf_{n \rightarrow \infty} \inf_{h_0 \in \mathcal{G}} \mathbb{P}_{h_0} \big( \partial^a h_0(x) \in C_n^a(x,A) \; \; \; \; \forall \; \; x \in \mathcal{X} \big) \geq 1 - \alpha \,; && \\
 \mbox{\textit{(ii)}} & ~~~~ \inf_{p \in [\ul p,\ol p]}  \inf_{h_0 \in \mc G^p} \mathbb{P}_{h_0} \bigg( \sup_{x \in \mathcal{X}} |C_n^a(x,A)| \leq C_{\ref{confmild-derivative}}(1+A) \bigg(  \frac{\log n}{n} \bigg)^{\frac{p-|a|}{2(p + \varsigma) + d}} \bigg) \rightarrow 1 \,, &&
\end{flalign*}
where  $C_{\ref{confmild-derivative}}>0$ is a universal constant.
\end{theorem}

\begin{remark} \normalfont
As the mildly ill-posed case nests nonparametric regression as a special case, our UCBs are honest and adaptive for derivatives of $h_0$ in general nonparametric regression models with non-Gaussian, heteroskedastic errors.
\end{remark}

As in the previous subsection,
for the severely ill-posed case, we can only establish valid coverage of the UCB (\ref{band-derivative}) using the critical value $\mbox{cv}^{a*}(x)$ corresponding to $\tilde J = \hat J_n$, i.e.,
\begin{equation} \label{cv-derivative-2}
 \mbox{cv}^{a*}(x) =
 z_{1-\alpha}^{a*} + \hat A \max\{\theta^*_{1-\hat \alpha}\,,\, \tilde {J}^{(|a|-\ul p)/d} /\hat \sigma_{\tilde J}^a(x) \} \,.
\end{equation}
This band reduces to the Procedure~2$^\prime$ UCB when $\theta^*_{1-\hat \alpha} \geq \tilde {J}^{|a|-\ul p/d} /\hat \sigma_{\tilde J}^a(x)$ for all $x$, which is the case in our empirical application.

Let $C_n^a(x,A)$ denote the band (\ref{band-derivative}) with critical value (\ref{cv-derivative-2}) when $\hat A$ is replaced by a constant $A>0$.

\begin{theorem}\label{confsevere-derivative}
Let Assumptions \ref{a-data}-\ref{a-var} hold, $|a| < \ul p$, and suppose the model is severely ill-posed. Then: there is a constant $A^* > 0$ (independent of $\alpha$) such that for all $A \geq A^*$, 
\begin{flalign*}
 \mbox{\textit{(i)}} & ~~~~\liminf_{n \to \infty} \inf_{h_0 \in \mathcal{G}} \mathbb{P}_{h_0} \big( \partial^a h_0(x) \in C_n^a(x,A) \; \; \; \; \forall \; \; x \in \mathcal{X} \big) \geq 1 - \alpha \,; && \\
 \mbox{\textit{(ii)}} & ~~~~ \inf_{p \in [\ul p,\ol p]}  \inf_{h_0 \in \mc G^p} \mathbb{P}_{h_0} \bigg( \sup_{x \in \mathcal{X}} |C_n^a (x,A)| \leq  C_{\ref{confsevere-derivative}} (1+A) (\log n)^{(|a|-\ul p)/\varsigma} \bigg) \rightarrow 1\,, &&
\end{flalign*}
where  $C_{\ref{confsevere-derivative}}>0$  is a universal constant.
\end{theorem}

\section{Additional Simulations}\label{s:mcs}

In this section we present two additional simulation studies. The first is a nonparametric IV design with a non-monotonic, non-Lipschitz structural function. The second is a very wiggly nonparametric regression design, which shows that $\tilde J$ can choose a relatively high-dimensional model when needed. Finally, Appendix~\ref{appsec:engel} presents a third set of simulations in an empirically calibrated Engel curve design which is severely ill-posed.

\subsection{Nonparametric IV Design}

This design features a non-monotonic, non-Lipschitz structural function. We first draw $(U,V)$ from a bivariate normal distribution with mean zero, unit variance, and correlation $0.75$, and draw $Z \sim N(0,1)$ independent of $(U,V)$. We then set $W = \Phi(Z)$ where $\Phi(\cdot)$ denotes the standard normal CDF, $X = \Phi( D(Z + V) + (1-D) V)$ where $D$ is an independent Bernoulli random variable taking the values $0$ and $1$ each with probability $0.5$, and
\begin{equation} \label{eq:npiv_mc}
 Y = \sin(4 X) \log (X) + U\,.
\end{equation}
The structural function $h_0(x) = \sin (4x) \log(x)$ is plotted in Figure~\ref{fig:npiv_mc}. Note that the derivative of $h_0$ diverges to $-\infty$ as $x \downarrow 0$. Therefore, $h_0$ is H\"older continuous with exponent $p$ for any $p < 1$, but not Lipschitz continuous.

For each simulated data set we compute our data-driven estimator $\hat h_{\tilde J}$ and UCBs from (\ref{band}). We compare these with estimators and UCBs using deterministic choices of sieve dimensions for $J = 4$, $5$, $7$, and $11$ (the first few dimensions over which our procedure searches). We again use a cubic B-spline basis to approximate $h_0$ and a quartic B-spline for the reduced form.

\begin{table}[t]
\begin{center}
\caption{\label{tab:npiv_mc} Simulation Results for the Nonparametric IV Design (\ref{eq:npiv_mc}).}
{\small
\begin{tabular}{ccccccccccccc} \hline \hline \\[-10pt]
  & & \multicolumn{2}{c}{Data-driven}  & & \multicolumn{8}{c}{Deterministic} \\\cline{3-4} \cline{6-13} \\[-10pt]
  & & & & & \multicolumn{2}{c}{$J = 4$} & \multicolumn{2}{c}{$J = 5$} & \multicolumn{2}{c}{$J = 7$} & \multicolumn{2}{c}{$J = 11$} \\ \\[-10pt] \hline \\[-10pt]
  \multicolumn{13}{c}{Sup-norm Loss} \\ \\[-10pt]
$n$ & & mean & med. & & mean & med. & mean & med. & mean & med. & mean & med. \\ \\[-10pt]
 \phantom{1}1250 & & 0.541 & 0.491 & & 0.539 & 0.489 & 0.678 & 0.630 & 1.087 & 1.000 & 1.524 & 1.422 \\
 \phantom{1}2500 & & 0.395 & 0.360 & & 0.393 & 0.359 & 0.486 & 0.451 & 0.890 & 0.835 & 1.342 & 1.283 \\
 \phantom{1}5000 & & 0.323 & 0.292 & & 0.319 & 0.291 & 0.367 & 0.345 & 0.761 & 0.696 & 1.231 & 1.169 \\
10000 & & 0.262 & 0.241 & & 0.256 & 0.239 & 0.270 & 0.255 & 0.623 & 0.556 & 1.186 & 1.136 \\ \\[-10pt] \hline \\[-10pt]
  \multicolumn{13}{c}{UCB Coverage} \\ \\[-10pt]
 & & 90\% & 95\% & & 90\% & 95\% & 90\% & 95\% & 90\% & 95\% & 90\% & 95\% \\ \\[-10pt]
 \phantom{1}1250 & & 0.997 & 0.999 & & 0.816 & 0.892 & 0.930 & 0.974 & 0.951 & 0.978 & 0.967 & 0.984 \\
 \phantom{1}2500 & & 0.995 & 0.997 & & 0.744 & 0.859 & 0.910 & 0.950 & 0.956 & 0.983 & 0.978 & 0.991 \\
 \phantom{1}5000 & & 0.978 & 0.992 & & 0.566 & 0.724 & 0.881 & 0.947 & 0.937 & 0.976 & 0.975 & 0.989 \\
10000 & & 0.908 & 0.949 & & 0.324 & 0.470 & 0.847 & 0.921 & 0.935 & 0.986 & 0.967 & 0.989 \\  \\[-10pt] \hline \\[-10pt]
  & & & & & \multicolumn{8}{c}{95\% UCB Relative Width (Deterministic/Data-driven)} \\ \cline{6-13} \\[-10pt]
  & &  &  & & mean & med. & mean & med. & mean & med. & mean & med. \\ \\[-10pt]
 \phantom{1}1250 & & & & & 0.658 & 0.663 & 0.925 & 0.897 & 1.502 & 1.451 & 2.122 & 2.046 \\
 \phantom{1}2500 & & & & & 0.661 & 0.665 & 0.923 & 0.908 & 1.790 & 1.731 & 2.554 & 2.502 \\
 \phantom{1}5000 & & & & & 0.663 & 0.668 & 0.917 & 0.914 & 2.255 & 2.158 & 3.286 & 3.228 \\
10000 & & & & & 0.661 & 0.668 & 0.913 & 0.914 & 2.830 & 2.757 & 4.515 & 4.445 \\ \\[-10pt] \hline
\end{tabular}
}
\end{center}
\vskip -14pt
\end{table}

\begin{table}[t]
\begin{center}
\caption{\label{tab:npiv_2_A} Coverage of 95\% UCBs $C_n(x, A)$, Nonparametric IV Design (\ref{eq:npiv_mc}).}
{\small
\begin{tabular}{ccccccccccccccc}
\hline \hline \\[-10pt]
 &  \multicolumn{13}{c}{$A$} \\  \cline{3-15} \\[-10pt]
 $n$  & & $0.00$ & $0.01$ & $0.05$ & $0.10$ & $0.20$ & $0.30$ & $0.40$ & $0.50$ & $0.60$ & $0.70$ & $0.80$ & $0.90$ & $1.00$ \\ \\[-10pt]
\phantom{1}1250 & & 0.96 & 0.97 & 0.98 & 0.99 & 1.00 & 1.00 & 1.00 & 1.00 & 1.00 & 1.00 & 1.00 & 1.00 & 1.00 \\
\phantom{1}2500 & & 0.94 & 0.94 & 0.95 & 0.97 & 0.98 & 0.99 & 1.00 & 1.00 & 1.00 & 1.00 & 1.00 & 1.00 & 1.00 \\
\phantom{1}5000 & & 0.87 & 0.88 & 0.91 & 0.94 & 0.97 & 0.99 & 0.99 & 1.00 & 1.00 & 1.00 & 1.00 & 1.00 & 1.00 \\
10000 & & 0.72 & 0.73 & 0.78 & 0.83 & 0.91 & 0.95 & 0.97 & 0.99 & 0.99 & 1.00 & 1.00 & 1.00 & 1.00 \\ \\[-10pt] \hline
\end{tabular}
}
\end{center}
\vskip -14pt
\end{table}

The first panel of Table~\ref{tab:npiv_mc} presents the average sup-norm loss of $\hat h_{\tilde J}$ across simulations. These are of similar magnitude to the loss for deterministic-$J$ estimates with $J = 4$ and $5$ and are much smaller than the loss with $J = 7$ and $11$. Our data-driven UCBs demonstrate valid but slightly conservative coverage for smaller $n$ and coverage close to nominal coverage for $n = 10000$. Bands with $J = 4$ have poor coverage while bands with $J = 5$ have valid coverage for the smaller sample sizes but under-cover for $n = 10000$. It seems $J = 7$ or $J = 11$ is required to have valid coverage for $n = 10000$ in this design. Note that while our bands are slightly conservative for smaller $J$, they are only about 10\% wider than the $J = 5$ bands, and less than half the width of the $J = 7$ bands.

In Figure~\ref{fig:npiv_mc} we plot data-driven estimates and UCBs for $h_0$ and its derivative over $[0.01, 0.99]$ for a sample of size 2500, alongside deterministic-$J$ estimates and UCBs. In this sample, $\tilde J = 4$ and our data-driven UCBs contain the true structural function. The data-driven bands are narrower and more accurately convey the shape of $h_0$ than the $J = 7$ bands, which are much more wiggly. Our bands are also of a similar width to (but are less wiggly than) the $J = 5$ bands. Panel~(d) of Figure~\ref{fig:npiv_mc} also presents data-driven estimates and UCBs for the conditional mean of $Y$ given $X$. Here the data-driven choice is again $\tilde J =4$. The true structural function falls outside the UCBs for the conditional mean function over almost all of the support of $X$, highlighting the importance of estimating $h_0$ using IV methods in this design.

Finally, in Table~\ref{tab:npiv_2_A} we present the coverage of our data-driven UCBs $C_n(x,A)$ where we replace $\hat A = \log \log \tilde J$ with a deterministic choice $A$ ranging over $[0,1]$. For this design, $A \geq 0.3$ suffices for correct coverage. In particular,  $\hat A = \log \log \tilde J$ yields correct coverage.

\begin{figure}[p]
\begin{center}

\begin{subfigure}[t]{\textwidth}
\caption{Data-driven Estimates and UCBs}
\begin{center}
\vskip -10pt
\begin{subfigure}[t]{0.45\textwidth}
\begin{center}
\includegraphics[width=\linewidth]{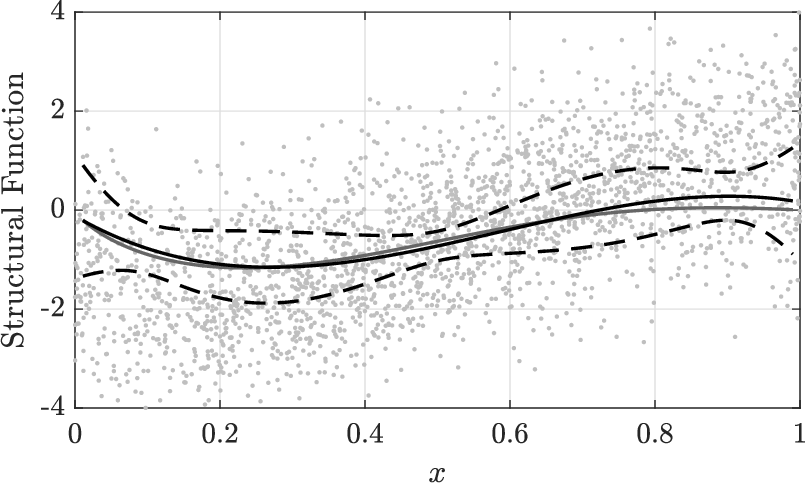}
\end{center}
\end{subfigure}
\begin{subfigure}[t]{0.45\textwidth}
\begin{center}
\includegraphics[width=\linewidth]{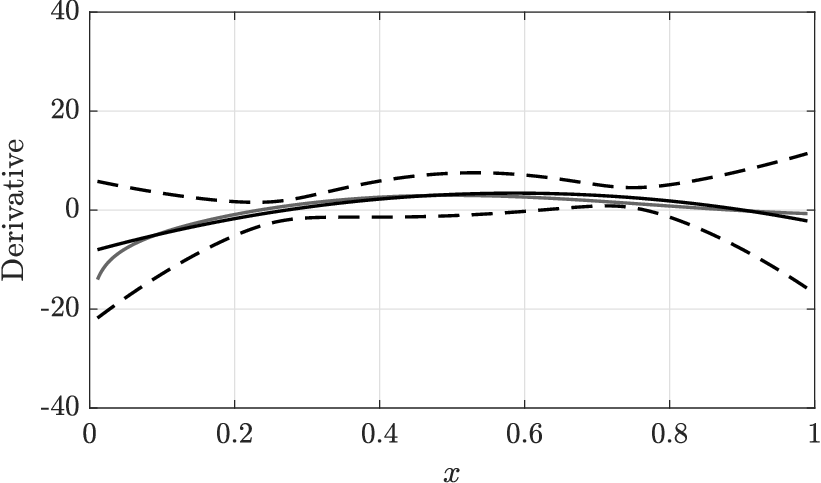}
\end{center}
\end{subfigure}
\end{center}
\end{subfigure}
\vskip -4pt

\begin{subfigure}[t]{\textwidth}
\caption{Estimates and UCBs with $J = 5$}
\begin{center}
\vskip -10pt
\begin{subfigure}[t]{0.45\textwidth}
\begin{center}
\includegraphics[width=\linewidth]{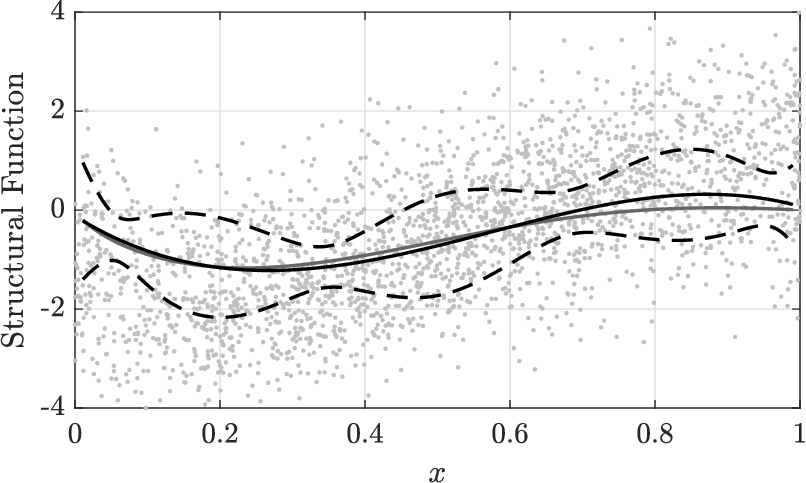}
\end{center}
\end{subfigure}
\begin{subfigure}[t]{0.45\textwidth}
\begin{center}
\includegraphics[width=\linewidth]{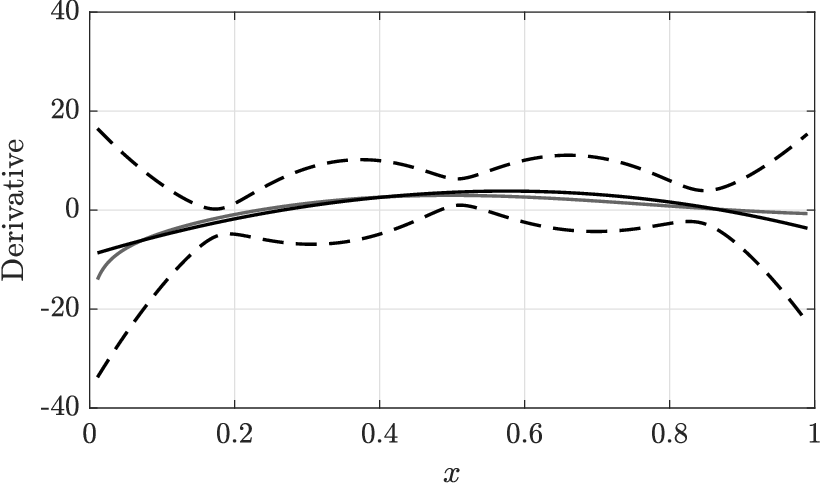}
\end{center}
\end{subfigure}
\end{center}
\end{subfigure}
\vskip -4pt

\begin{subfigure}[t]{\textwidth}
\caption{Estimates and UCBs with $J = 7$}
\begin{center}
\vskip -10pt
\begin{subfigure}[t]{0.45\textwidth}
\begin{center}
\includegraphics[width=\linewidth]{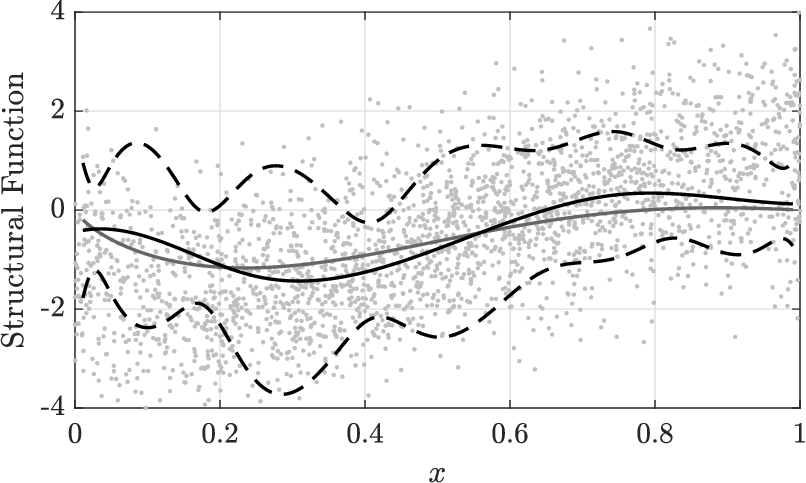}
\end{center}
\end{subfigure}
\begin{subfigure}[t]{0.45\textwidth}
\begin{center}
\includegraphics[width=\linewidth]{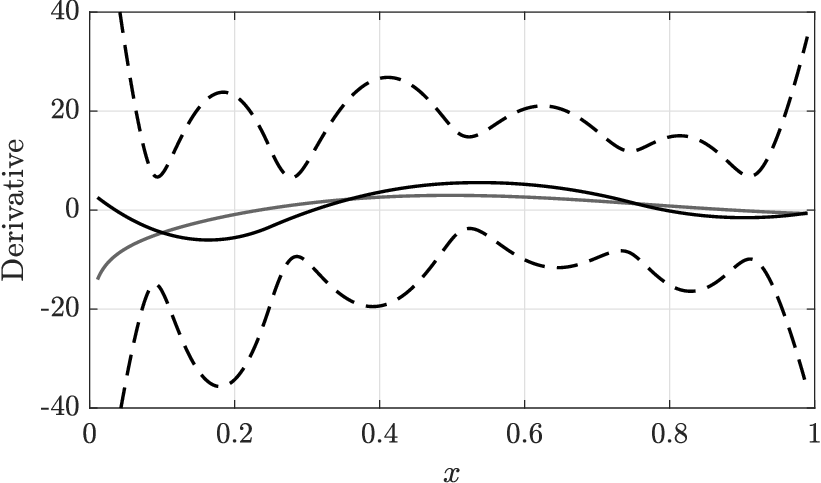}
\end{center}
\end{subfigure}
\end{center}
\end{subfigure}
\vskip -4pt

\begin{subfigure}[t]{\textwidth}
\caption{Data-driven Estimates and UCBs for the Conditional Mean of $Y$ given $X$}
\begin{center}
\vskip -10pt
\begin{subfigure}[t]{0.45\textwidth}
\begin{center}
\includegraphics[width=\linewidth]{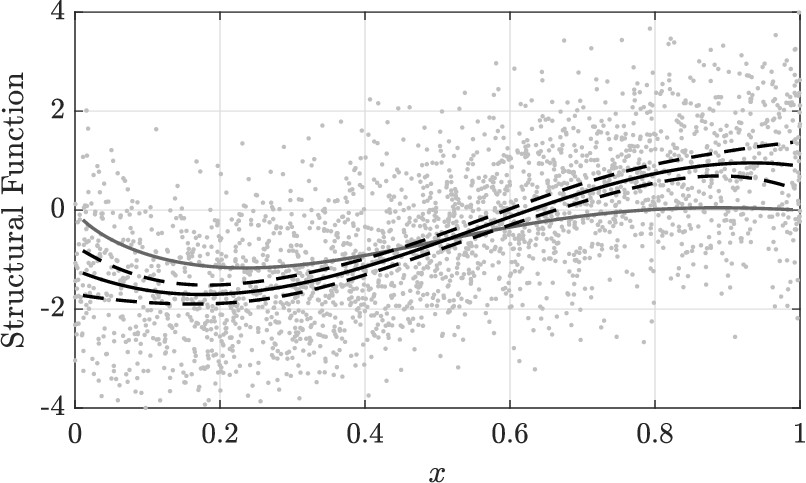}
\end{center}
\end{subfigure}
\begin{subfigure}[t]{0.45\textwidth}
\begin{center}
\includegraphics[width=\linewidth]{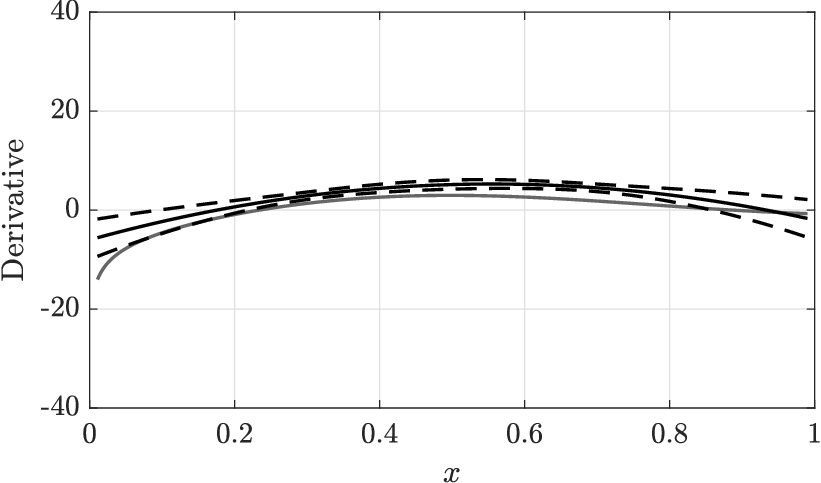}
\end{center}
\end{subfigure}
\end{center}
\end{subfigure}

\vskip -6pt

\caption{\label{fig:npiv_mc} Nonparametric IV design (\ref{eq:npiv_mc}): Plots for a sample of size $n = 2500$. Left panels correspond to the structural function, right panels correspond to its derivative. \emph{Note:} Solid grey lines are the true structural function and derivative; solid black lines are estimates, dashed black lines are 95\% UCBs. Supports are truncated to $[0.01, 0.99]$ as the derivative is unbounded as $x \downarrow 0$.}
\vskip -15pt
\end{center}
\end{figure}

\subsection{Nonparametric Regression Design}\label{s:mc_npr}

For this design we simulate $X \sim U[0,1]$ and $U \sim N(0,1)$ independently, then set
\begin{equation} \label{eq:reg_mc}
 Y = \sin(15 \pi X) \cos (X) + U \,.
\end{equation}
Here $h_0(x) = \sin(15 \pi x) \cos (x)$ is very wiggly over $[0,1]$ and requires a high value of $J$ to be selected in order to well approximate $h_0$ (see Figure~\ref{fig:regression}). While $h_0$ is infinitely differentiable, its Lipschitz constant is at least $47.1$, the Lipschitz constant of its derivative is at least $2220$, and Lipschitz constants grow rapidly for higher derivatives.

We again compare our data-driven estimator and UCBs using the procedures described in Appendix \ref{sec:regression} with estimators and UCBs that use deterministic choices of $J$ for $J = 11$, $19$, $35$, and $67$ (these are a subset of values over which our procedure searches). We again use cubic B-splines to approximate $h_0$.

\begin{table}[t]
\begin{center}
\caption{\label{tab:regression} Simulation Results for the Nonparametric Regression Design (\ref{eq:reg_mc}).}
{\small
\begin{tabular}{ccccccccccccc} \hline \hline \\[-10pt]
  & & \multicolumn{2}{c}{Data-driven}  & & \multicolumn{8}{c}{Deterministic} \\\cline{3-4} \cline{6-13} \\[-10pt]
  & & & & & \multicolumn{2}{c}{$J = 11$} & \multicolumn{2}{c}{$J = 19$} & \multicolumn{2}{c}{$J = 35$} & \multicolumn{2}{c}{$J = 67$} \\ \\[-10pt] \hline \\[-10pt]
  \multicolumn{13}{c}{Sup-norm Loss} \\ \\[-10pt]
$n$ & & mean & med. & & mean & med. & mean & med. & mean & med. & mean & med. \\ \\[-10pt]
 \phantom{1}1250 & & 0.778 & 0.650 & & 1.242 & 1.175 & 0.808 & 0.732 & 0.671 & 0.591 & 1.111 & 0.898 \\
 \phantom{1}2500 & & 0.490 & 0.423 & & 1.182 & 1.133 & 0.705 & 0.650 & 0.483 & 0.415 & 0.698 & 0.603 \\
 \phantom{1}5000 & & 0.347 & 0.303 & & 1.140 & 1.109 & 0.641 & 0.608 & 0.332 & 0.294 & 0.486 & 0.426 \\
10000 & & 0.236 & 0.209 & & 1.113 & 1.095 & 0.606 & 0.585 & 0.233 & 0.206 & 0.330 & 0.291 \\  \\[-10pt] \hline \\[-10pt]
  \multicolumn{13}{c}{UCB Coverage} \\ \\[-10pt]
 & & 90\% & 95\% & & 90\% & 95\% & 90\% & 95\% & 90\% & 95\% & 90\% & 95\% \\ \\[-10pt]
 \phantom{1}1250 & & 0.999 & 0.999 & & 0.000 & 0.000 & 0.000 & 0.000 & 0.790 & 0.864 & 0.627 & 0.713 \\
 \phantom{1}2500 & & 1.000 & 1.000 & & 0.000 & 0.000 & 0.000 & 0.000 & 0.847 & 0.899 & 0.776 & 0.857 \\
 \phantom{1}5000 & & 1.000 & 1.000 & & 0.000 & 0.000 & 0.000 & 0.000 & 0.857 & 0.909 & 0.845 & 0.910 \\
10000 & & 1.000 & 1.000 & & 0.000 & 0.000 & 0.000 & 0.000 & 0.889 & 0.936 & 0.867 & 0.934 \\ \\[-10pt] \hline \\[-10pt]
  & & & & & \multicolumn{8}{c}{95\% UCB Relative Width (Deterministic/Data-driven)} \\ \cline{6-13} \\[-10pt]
  & & &  & & mean & med. & mean & med. & mean & med. & mean & med. \\ \\[-10pt]
 \phantom{1}1250 & &  &  & & 0.217 & 0.209 & 0.287 & 0.279 & 0.410 & 0.405 & 0.616 & 0.582 \\
 \phantom{1}2500 & &  &  & & 0.206 & 0.206 & 0.279 & 0.279 & 0.401 & 0.405 & 0.603 & 0.599 \\
 \phantom{1}5000 & &  &  & & 0.190 & 0.191 & 0.256 & 0.260 & 0.374 & 0.382 & 0.565 & 0.568 \\
10000 & &  &  & & 0.195 & 0.196 & 0.261 & 0.262 & 0.380 & 0.383 & 0.573 & 0.572 \\  \\[-10pt] \hline
\end{tabular}
}
\end{center}
\vskip -20pt
\end{table}

It is clear from the simulation results presented in Table~\ref{tab:regression} that $J > 19$ is required to well approximate the true $h_0$. The average sup-norm loss of $\hat h_{\tilde J}$ is similar to that of the deterministic-$J$ estimator for $J = 35$, and is smaller than the average loss for all other $J$ presented in the table. Our data-driven UCBs also deliver valid, but conservative, coverage for the true conditional mean function. UCBs based on a deterministic choice of $J$ have zero coverage for $J = 11$ and $J = 19$ as these dimensions are too small to adequately approximate $h_0$, and tend to under-cover for the remaining $J$, except perhaps for $J = 35$ when $n = 10000$.

In this design a much smaller value of $A$ suffices to deliver valid coverage, as seen in Table~\ref{tab:regression_A}. The reason is that the set $\hat{\mc J}_-$ is large and  $\hat h_J$ varies a lot across different $J$ due to the wiggliness of $h_0$. Therefore  $z_{1-\alpha}^*$, which is the quantile of a sup-statistic over $\mc X \times \hat{\mc J}_-$, is relatively more conservative than for the other designs. This extra conservativeness suffices to deliver valid coverage in this design with smaller $A$.

\begin{table}[t]
\begin{center}
\caption{\label{tab:regression_A} Coverage of 95\% UCBs $C_n(x, A)$, Nonparametric Regression Design (\ref{eq:reg_mc})}
{\small
\begin{tabular}{ccccccccccccccc}
\hline \hline \\[-10pt]
 &  \multicolumn{13}{c}{$A$} \\  \cline{3-15} \\[-10pt]
 $n$  & & $0.00$ & $0.01$ & $0.05$ & $0.10$ & $0.20$ & $0.30$ & $0.40$ & $0.50$ & $0.60$ & $0.70$ & $0.80$ & $0.90$ & $1.00$ \\ \\[-10pt]
\phantom{1}1250 & & 0.88 & 0.88 & 0.90 & 0.93 & 0.96 & 0.98 & 0.98 & 0.99 & 0.99 & 0.99 & 1.00 & 1.00 & 1.00 \\
\phantom{1}2500 & & 0.93 & 0.94 & 0.96 & 0.98 & 0.99 & 1.00 & 1.00 & 1.00 & 1.00 & 1.00 & 1.00 & 1.00 & 1.00 \\
\phantom{1}5000 & & 0.96 & 0.96 & 0.98 & 0.99 & 1.00 & 1.00 & 1.00 & 1.00 & 1.00 & 1.00 & 1.00 & 1.00 & 1.00 \\
10000 & & 0.97 & 0.97 & 0.98 & 0.99 & 1.00 & 1.00 & 1.00 & 1.00 & 1.00 & 1.00 & 1.00 & 1.00 & 1.00 \\  \\[-10pt] \hline
\end{tabular}
}
\end{center}
\vskip -20pt
\end{table}

\begin{figure}[t]
\begin{center}

\begin{subfigure}[t]{\textwidth}
\caption{Data-driven Estimates and UCBs}
\begin{center}
\vskip -8pt
\begin{subfigure}[t]{0.475\textwidth}
\begin{center}
\includegraphics[width=\linewidth]{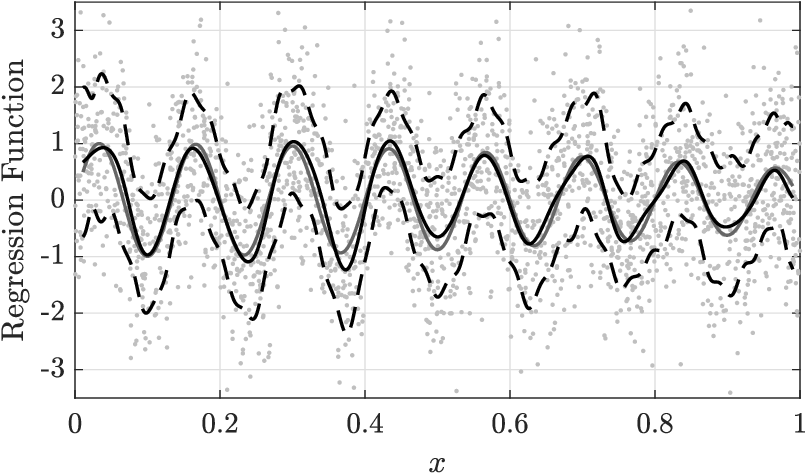}
\end{center}
\end{subfigure}
\begin{subfigure}[t]{0.475\textwidth}
\begin{center}
\includegraphics[width=\linewidth]{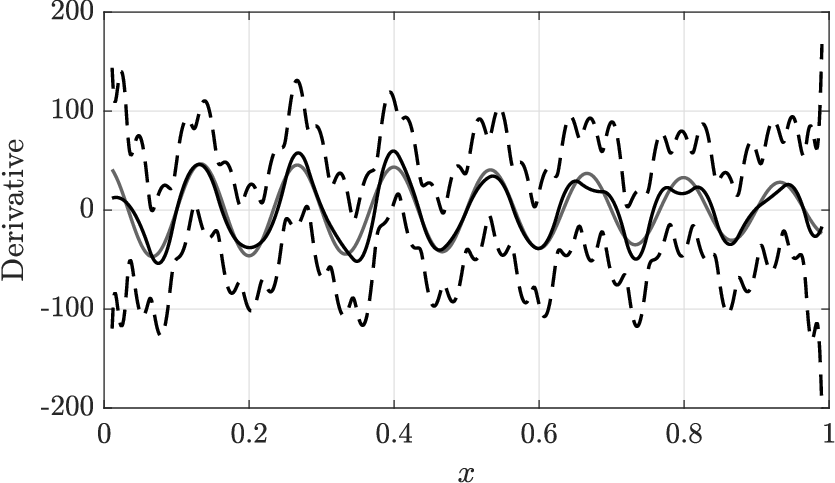}
\end{center}
\end{subfigure}

\end{center}
\end{subfigure}

\begin{subfigure}[t]{\textwidth}
\caption{Estimates and UCBs with $J = 67$}
\begin{center}
\vskip -8pt
\begin{subfigure}[t]{0.475\textwidth}
\begin{center}
\includegraphics[width=\linewidth]{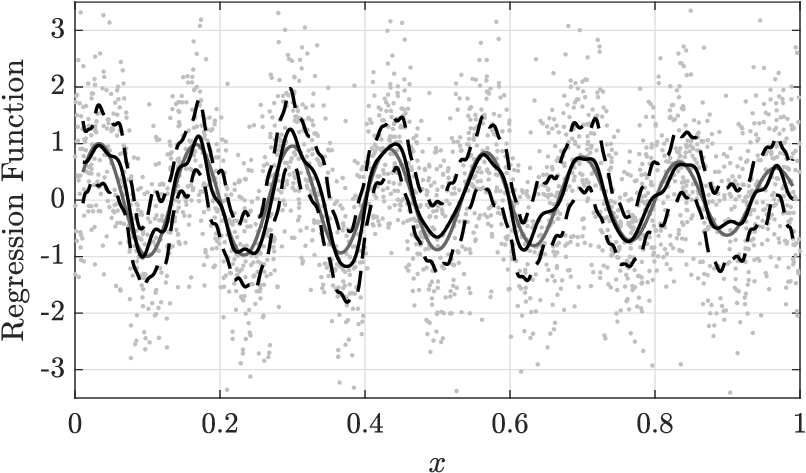}
\end{center}
\end{subfigure}
\begin{subfigure}[t]{0.475\textwidth}
\begin{center}
\includegraphics[width=\linewidth]{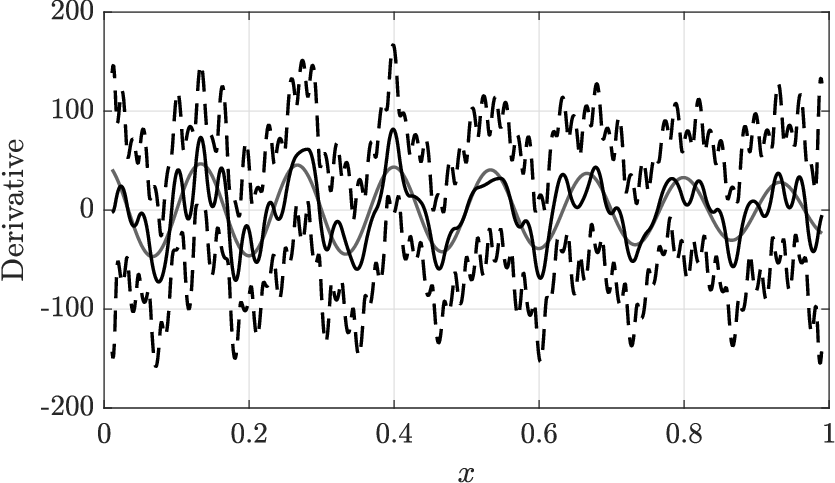}
\end{center}
\end{subfigure}

\end{center}
\end{subfigure}

\vskip -4pt

\caption{\label{fig:regression} Nonparametric regression design (\ref{eq:reg_mc}): Plots for a sample of size $n = 2500$. Left panels correspond to the conditional mean function, right panels correspond to its derivative. \emph{Note:} Solid grey lines are the true conditional mean function and its derivative; solid black lines are estimates, dashed black lines are 95\% UCBs.}
\end{center}
\vskip -20pt
\end{figure}

Figure \ref{fig:regression} plots our data-driven estimator $\hat h_{\tilde J}$ and 95\% UCBs for the conditional mean function for a sample of size 2500. In this sample, $\tilde J = 35$. The data-driven estimator well approximates the true conditional mean function $h_0$, which lies entirely within the 95\% UCBs, and the same is true for estimates and UCBs for the derivative of $h_0$. Deterministic-$J$ bands with $J = 67$ are of a similar width to our data-driven bands for this sample, even though they use a less conservative critical value which only accounts for sampling uncertainty. The estimator is also much wigglier with $J = 67$ than our data-driven estimator and does not approximate $h_0$ as well.

\section{Extensions}\label{sec:extensions}

So far we have assumed the structural function $h_0$ is a general $d$-variate function. As with many other nonparametric estimation problems, minimax rates deteriorate as $d$ increases. This so-called curse of dimensionality applies to any estimator of $h_0$. However, it can be circumvented by imposing additional structure on $h_0$ (when appropriate), such as additivity or partial linearity. In this section, we show how our data-driven procedures extend to additive and partially linear models.

\medskip

\paragraph{Additive Structural Functions.}

Consider first the additive structural function:
\[
 h_0(x) = c_0 + h_{10}(x_1) + \ldots + h_{d0}(x_d)
\]
where $x = (x_1,\ldots,x_d)'$. Here $c_0$ is a constant representing an ``intercept'' term and the $h_{i0}$ are suitably normalized for identifiability. In the context of nonparametric regression, \cite{Stone1985} showed that imposing additivity can yield estimators of $h_0$ that achieve the same (optimal) rate for general $d$ as for $d = 1$.

Our methods may be easily adapted to additive models as follows. We assume for sake of exposition that $X$ is bivariate $(d = 2)$. Let $\psi^J(x) = (1,\tilde \psi_1^J(x_1)',\tilde \psi_2^J(x_2)')'$ where for $i = 1,2$ we have $\tilde \psi_i^J(x_i) = (\tilde \psi_{J1}(x_i),\ldots,\tilde \psi_{JJ}(x_i))'$. Here $J$ represents the dimensions of sieves used to approximate both $h_{10}$ and $h_{20}$. The basis functions $\tilde \psi_{J1},\ldots,\tilde \psi_{JJ}$ are formed by setting $\tilde \psi_{Jj}(x_i) = \psi_{Jj}(x_i) - \int_0^1 \psi_{Jj}(v) \mr d v$ with $\psi_{J1}(x_1),\ldots,\psi_{JJ}(x_1)$ a univariate B-spline basis. We estimate $c_0$ and $c_J^i$, $i = 1,2$, by TSLS regression of $Y$ on $\psi^J(X)$ using $b^{K(J)}(W)$ as instruments:
\[
 \left( \begin{array}{c} \hat c_h \\ \hat c_J^1 \\ \hat c_J^2 \end{array} \right) = \left(\mf \Psi_J' \mf P_{K(J)}^{\phantom \prime} \mf \Psi_J^{\phantom \prime} \right)^{-} \mf \Psi_J' \mf P_{K(J)}^{\phantom \prime} \mf Y =\mf M_J \mf Y\,,
\]
where the notation is as in Section~\ref{s:procedure} but  with $\psi^J(x) = (1,\tilde \psi_1^J(x_1)',\tilde \psi_2^J(x_2)')'$. The estimator of $h_{i0}$ is $\hat h_{iJ}(x_i) =  (\psi^J_i(x_i))' \hat c_J^i$. Derivatives of $h_{i0}$ are 
estimated by differentiating $\hat h_{iJ}$.

Our data-driven choice of $J$ is implemented exactly as described in Section~\ref{sec:Jtilde} with $\psi^J(x) = (1,\tilde \psi_1^J(x_1)',\tilde \psi_2^J(x_2)')'$. Data-driven UCBs for $h_{10}$ are formed analogously to Section~\ref{sec:ucbs} with two small modifications. First, when computing the critical value $z_{1 - \alpha}^*$ in Step 4 of Procedure 2 we now use the sup-statistic
\[
 \sup_{(x_1,J) \in [0,1] \times \hat{\mc J}_{-}} \left| \frac{D_{1J}^*(x_1)}{\hat \sigma_{1J}(x_1)} \right|
\]
where $D_{1J}^*(x_1) = (0,\tilde \psi_1^J(x_1)',0_J')' \mf M_J \hat{\mf u}_J^*$ with $0_J$ a $J$-vector of zeros, and
\[
 \hat \sigma_{1J}^2(x)  =  (0,\tilde \psi_1^J(x_1)',0_J') \mf M_J^{\phantom \prime} \widehat{\mf U}_{J,J}^{\phantom \prime} \mf M_J' (0,\tilde \psi_1^J(x_1)',0_J')'\,.
\]
The 100$(1-\alpha)$\% UCB for $h_{10}$ is
\begin{equation*}
 C_n(x_1) = \bigg[ \hat{h}_{1\tilde{J}}(x_1) -  \mbox{cv}^*(x_1) \hat \sigma_{1\tilde J}(x_1) , ~  \hat{h}_{1\tilde{J}}(x_1) +  \mbox{cv}^*(x_1) \hat \sigma_{1\tilde J}(x_1) \bigg] \,
\end{equation*}
with
\[
 \mbox{cv}^*(x_1) =
 \begin{cases}
 z_{1-\alpha}^* + \hat A \theta^*_{1-\hat \alpha}  & \mbox{~if $\tilde J = \hat J$}, \\
 z_{1-\alpha}^* + \hat A \max\{\theta^*_{1-\hat \alpha}\,,\, \tilde {J}^{-\ul p} /\hat \sigma_{1\tilde J}(x_1) \} & \mbox{~if $\tilde J = \hat J_n$}
 \end{cases}
\]
where $\ul p$ is the minimal smoothness assumed for $h_{10}$ and $h_{20}$. UCBs for derivatives of $h_{10}$ are constructed analogously.

\medskip

\paragraph{Partially Linear Structural Functions.}

An alternative to additivity is the partially linear specification \citep{ai2003efficient}
\[
 h_0(x) = h_{10}(x_1) + x_2'\beta_0
\]
where $x$ is partitioned as $x = (x_1',x_2')'$ with $x_1$ of dimension $d_1 < d$, $h_{10}$ is an unknown function, and $\beta_0$ is an unknown vector of parameters. When $X$ is exogenous (so $W \equiv X$) this is the important partially linear regression model of \cite{Robinson1988}.

Our methods may be adapted to estimate and construct UCBs for $h_{10}$ as follows. First, we let $\psi^J(x) = (\psi^J_1(x_1)',x_2')'$ where $\psi^J_1(x_1) = (\psi_{J1}(x_1),\ldots,\psi_{JJ}(x_1))'$.\footnote{We assume without loss of generality that the $X_2$ variables have mean zero, which permits identification of $h_0$ and $\beta$. In practice these variables can be de-meaned.} We estimate $c_J$ and $\beta$ by TSLS regression of $Y$ on $\psi^J(X)$ using $b^{K(J)}(W)$ as instruments:
\[
 \left( \begin{array}{c} \hat c_J \\ \hat \beta \end{array} \right) = \left(\mf \Psi_J' \mf P_{K(J)}^{\phantom \prime} \mf \Psi_J^{\phantom \prime} \right)^{-} \mf \Psi_J' \mf P_{K(J)}^{\phantom \prime} \mf Y =\mf M_J \mf Y\,,
\]
where the notation is as in Section~\ref{s:procedure} but with $\psi^J(x) = (\psi^J_1(x_1)',x_2')'$. The estimator of $h_{10}$ is $\hat h_{1J}(x_1) =  (\psi^J_1(x_1))' \hat c_J$. Derivatives of $h_{10}$ are again estimated by differentiating $\hat h_{1J}$. When $X$ is exogenous, we simply take $w = x$ and $b^K(w) = \psi^J(x)$.

Our data-driven choice of $J$ is implemented analogously to Section~\ref{sec:Jtilde}, except we form the contrasts $D_J$, $D_{J}(x)-D_{J_2}(x)$, and $D_{J}^*(x)-D_{J_2}^*(x)$ and the variance terms $\hat \sigma_J^2(x)$ and $\hat \sigma_{J,J_2}^2(x)$ 
using $\psi^J_0(x_1) := (\psi^J_1(x_1)',0_{d_2})'$ in place of $\psi^J(x)$. As such, the $t$-statistics are functions of $x_1$ only and the supremums in the sup-statistics in Steps 2 and 3 of Procedure 1 only need to be computed over the support $\mc X_1$ of $x_1$. UCBs for $h_{10}$ are constructed analogously to Section~\ref{sec:ucbs}, where the contrast $D_J^*(x)$ and the variance term $\hat \sigma_{J}^2(x)$ are again formed using $\psi^J_0(x_1)$ in place of $\psi^J(x)$.
The 100$(1-\alpha)$\% UCB for $h_{10}$ is
\begin{equation*}
 C_n(x_1) = \bigg[ \hat{h}_{1\tilde{J}}(x_1) - \mbox{cv}^*(x_1) \hat \sigma_{1\tilde J}(x_1) , ~  \hat{h}_{1\tilde{J}}(x_1) + \mbox{cv}^*(x_1) \hat \sigma_{1\tilde J}(x_1) \bigg] \,
\end{equation*}
with
\[
 \mbox{cv}^*(x_1) =
 \begin{cases}
 z_{1-\alpha}^* + \hat A \theta^*_{1-\hat \alpha}  & \mbox{~if $\tilde J = \hat J$}, \\
 z_{1-\alpha}^* + \hat A \max\{\theta^*_{1-\hat \alpha}\,,\, \tilde {J}^{-\ul p/d_1} /\hat \sigma_{1\tilde J}(x_1) \} & \mbox{~if $\tilde J = \hat J_n$}
 \end{cases}
\]
where $\ul p$ is the minimal degree of smoothness assumed for $h_{10}$. UCBs for derivatives of $h_{10}$ are constructed analogously.

\section{Conclusion}\label{s:conclusion}

We have introduced data-driven procedures for estimation and inference on a nonparametric structural function $h_0$ and its derivatives using instrumental variables. Our data-driven choice of sieve dimension leads to estimators of $h_0$ and its derivatives that converge at the fastest possible (i.e., minimax) rate in sup-norm. Our data-driven uniform confidence bands (UCBs) for $h_0$ and its derivatives are shown to have coverage guarantees and contract at, or within a logarithmic factor of, the minimax rate. Both procedures have good finite sample performance in various simulation designs, including empirically-calibrated trade and Engel curve designs.
Our methods are simple to compute, and are applied to estimate and construct UCBs for the elasticity of the intensive margin of firm exports in a monopolistic competition model of international trade.

Aside from the extensions in Section~\ref{sec:extensions}, it would be straightforward to extend our methods to weakly dependent data, which is relevant for dynamic causal inference and reinforcement learning.
It would also be interesting to consider sup-norm rate-minimaxity jointly with respect to both $p$ and the degree of ill-posedness.

\appendix

\section{Nonparametric Regression}\label{sec:regression}

Here we specialize our data-driven procedures to nonparametric regression. The conditional mean function $h_0(x) = \E[Y|X = x]$ is estimated by
\[
\hat h_J(x) = (\psi^J(x))'\hat c_J~,~~ \hat c_J = \left(\mf \Psi_J' \mf \Psi_J^{\phantom \prime}\right)^- \mf \Psi_J' \mf Y\,.
\]
Notation is as in Section \ref{sec:Jtilde}, except now we set $\mf M_J = (\mf \Psi_J'  \mf \Psi_J^{\phantom \prime} )^{-}\mf \Psi_J'$.

\medskip

\paragraph{1.} Compute an upper truncation point $\hat J_{\max}$ of the index set as
\begin{equation} \label{eq:J_hat_max_regression}
 \hat{J}_{\max} = \min \bigg \{ J \in \T :   J \sqrt{\log J} \upsilon_n \leq 10 \sqrt n <  J^{+} \sqrt{\log J^{+}} \upsilon_n  \bigg \}
\end{equation}
with $\upsilon_n = \max\{1, (0.1 \log n)^4\}$, then compute $\hat{\mc J}$ as in (\ref{eq:index_set}) with this choice of $\hat J_{\max}$.

\medskip

\paragraph{2.} Let $\theta^*_{1-\hat \alpha}$ denote the $(1- \hat \alpha )$ quantile of (\ref{eq:sup-stat}) across independent draws of $(\varpi_i)_{i=1}^n$.
\medskip

\paragraph{3.} Take $\tilde J = \hat J$ for $\hat J$ defined in (\ref{eq:J_lepski}).

\medskip

Data-driven UCBs are also constructed analogously.

\medskip

\paragraph{4.} For UCBs for $h_0$, compute the critical value $z^*_{1-\alpha}$ as in (\ref{eq:z_star-UCB}). For UCBs for $\partial^a h_0$, compute the critical value $z^{a*}_{1-\alpha}$ as in (\ref{eq:z_star-UCB-derivative}).

\medskip

\paragraph{5.} The UCB for $h_0$ is
\[
 C_n(x) = \bigg[ \hat{h}_{\tilde{J}}(x) - \left( z_{1-\alpha}^* + \hat A \theta^*_{1-\hat \alpha} \right) \hat \sigma_{\tilde J}(x) , ~  \hat{h}_{\tilde{J}}(x) + \left( z_{1-\alpha}^* + \hat A \theta^*_{1-\hat \alpha} \right) \hat \sigma_{\tilde J}(x) \bigg] .
\]
The UCB for $\partial^a h_0$ is
\[
 C_n^a(x) = \bigg[ \partial^a \hat{h}_{\tilde{J}}(x) - \left( z^{a*}_{1-\alpha} + \hat A \theta^*_{1-\hat \alpha} \right) \hat \sigma_{\tilde J}^a(x) , ~  \partial^a \hat{h}_{\tilde{J}}(x) + \left( z^{a*}_{1-\alpha} + \hat A \theta^*_{1-\hat \alpha} \right) \hat \sigma^a_{\tilde J}(x) \bigg].
\]

\medskip

Theorem \ref{lepski2} and Corollary \ref{cor:lepski2} establish that $\tilde J$ leads to estimators of $h_0$ and its derivatives that attain the minimax sup-norm rates for nonparametric regression. Theoretical properties of the data-driven UCBs are established in Theorems~\ref{confmild} and \ref{confmild-derivative}.

\section{Additional Details for Section~\ref{s:application}}\label{ax:application}

\subsection{Basis Functions}

We construct basis functions the same way in both the simulations and empirical application. We use cubic B-splines ($r = 4$) to approximate $h_0$ and quartic B-splines ($r = 5$) to estimate the reduced-form. We also link the dimensions $J$ and $K(J)$ using $q = 2$.

As B-splines are supported on $[0,1]$ but $\tilde \pi_{ij}$ is negative, we transform $\tilde \pi_{ij}$ to $[0,1]$ using $\tilde \pi \mapsto \max\{0, \tilde \pi/10 + 1\}$. Under this transformation the very small fraction of observations for which $\tilde \pi_{ij} < -10$ or, equivalently, $\pi_{ij} < 0.005\%$, are truncated to zero (there were only four such observations in the empirical application). Similarly, we transform $z_{ij}$ to have support $[0,1]$ using its empirical CDF. The transformed $\tilde \pi_{ij}$ is not uniformly distributed on $[0,1]$ so we place interior knots at its empirical quantiles. The transformed $z_{ij}$ are uniformly distributed on $[0,1]$ so we place interior knots uniformly between $[0,1]$.

\subsection{Simulations}\label{ax:simulations}

\paragraph{DGP.}

Our first simulation design is based on \cite{HMT}. As in \cite{Melitz2003}, the only source of firm heterogeneity in their model is productivity. Hence, $r_{ij}(\omega) = e_{ij}(\omega)$, which is assumed to be lognormally distributed. The extensive margin is
\begin{equation}\label{eq:aag.log_eps.lognormal}
 \log \epsilon(\pi) = \mu + \sigma \sqrt 2 \mathrm{erf}^{-1}(1 - 2\pi),
\end{equation}
where $\mathrm{erf}(x) = \frac{2}{\sqrt \pi} \int_0^x e^{-\frac{1}{2}t^2} \, \mr d t$ is the error function and $\mathrm{erf}^{-1}$ is its inverse, and $\mu$ and $\sigma^2$ are the mean and variance of $\log e_{ij}$. The intensive margin function may be shown to be
\begin{equation}\label{eq:aag.log_rho.lognormal}
 \log \rho(\pi) = \mu + \frac{\sigma^2}{2} - \log(2 \pi) + \log \left( 1 + \mathrm{erf} \left( \frac{\sigma^2}{\sqrt 2} - \mathrm{erf}^{-1}(1 - 2\pi) \right) \right).
\end{equation}
Its elasticity is
\[
 \frac{d \log \rho(\pi)}{d \log \pi} = -1 + 2 \pi \frac{\exp \left(-\frac{\sigma^2}{\sqrt 2}\left(\frac{\sigma^2}{\sqrt 2} - 2 \mathrm{erf}^{-1}(1 - 2\pi) \right)\right)}{1 + \mathrm{erf} \left( \frac{\sigma^2}{\sqrt 2} - \mathrm{erf}^{-1}(1 - 2\pi) \right)}.
\]
Our second simulation design is based on the Pareto specification of \cite{Chaney}. In this design the intensive margin is $\log \rho(\pi) = \rho \log \pi$ and hence its elasticity is constant.

We generate data on $z_{ij}$ by sampling IID with replacement from the empirical distribution of $z_{ij}$. We then generate data on $\pi_{ij}$ and $\bar x_{ij}$ as follows. For the lognormal design, we estimate two partially linear IV models based on (\ref{eq:aag.log_eps}) and (\ref{eq:aag.log_rho}), namely
\[
 \begin{aligned}
 \log \epsilon(\pi_{ij})
 & = \beta_\epsilon z_{ij} + \delta_i^\epsilon + \zeta_j^\epsilon + e_{ij}^\epsilon , \\
 \log \bar x_{ij} - \log \rho(\pi_{ij}) & =  \beta_\rho z_{ij} + \delta_i^\rho + \zeta_j^\rho + e_{ij}^\rho.
 \end{aligned}
\]
In the first equation, we treat $\log \epsilon(\pi_{ij})$ as the dependent variable using the functional form (\ref{eq:aag.log_eps.lognormal}) with $\mu = -2$ and $\sigma = 1.2$. In the second, we treat $\log \bar x_{ij} - \log \rho(\pi_{ij})$ as the dependent variable using the functional form (\ref{eq:aag.log_rho.lognormal}). We compute the covariance matrix $\hat \Sigma$ of the residuals $(\hat e_{ij}^\epsilon, \hat e_{ij}^\rho)$.
We simulate $(e_{ij}^\epsilon,e_{ij}^\rho)$ as independent $N(0,\hat \Sigma)$ random vectors. Given $e_{ij}^\epsilon$ and $z_{ij}$, we set $\log \epsilon(\pi_{ij}) = 0.875 z_{ij} - 7 + e_{ij}^\epsilon$, then invert $\log \epsilon(\pi_{ij})$ using (\ref{eq:aag.log_eps.lognormal}) to obtain $\log \pi_{ij}$. This gives a distribution with support, mean, and variance roughly calibrated to the data used in the application. We then set $\bar x_{ij} = \log \rho(\pi_{ij}) - \tilde \sigma \kappa^\tau z_{ij} + \delta_i + \zeta_j + e_{ij}^\rho$ using the functional form (\ref{eq:aag.log_rho.lognormal}) for $\log \rho$, with $\delta_i = 0$, $\zeta_j = 0$, and with $\tilde \sigma = 2.9$ and $\kappa^\tau = 0.36$ as in AAG. 
We set exporter and importer FEs to zero for $\log \rho$ so that we can compare the effect of first-stage estimation of these FEs on the performance of our procedures.

We generate data for the Pareto design (for which the elasticity of $\rho$ is constant) as described above, except we use $\log \rho(\log \pi) = -0.23 \log \pi$ in place of (\ref{eq:aag.log_rho.lognormal}),  where the coefficient $-0.23$ matches AAG's estimate for the constant elasticity specification.\footnote{Note that we maintain the same DGP for $\pi$ as in the lognormal specification. While one could also generate $\pi$ using the Pareto assumption, this would change the joint distribution of $(\pi_{ij},z_{ij})$, and hence the instrument strength and degree of endogeneity. We keep the distribution fixed across designs so that any difference in results is attributable to the different structural functions $\log \rho$ only.}

\medskip

\paragraph{Simulation Results for the Log-normal Design without Fixed Effects.}

We first present in Tables~\ref{tab:trade.lognormal.h.nofe} and \ref{tab:trade.lognormal.d.nofe} results for estimating $\log \rho$ and the elasticity of $\rho$ in the log-normal design when we treat the FEs $\delta_i$ and $\zeta_j$ as zero. These results shut down any estimation error that may be introduced by first-stage estimation of the FEs. Overall, the results are very similar to those reported in Table~\ref{tab:trade.lognormal.deriv} with first-stage estimation of FEs: the sup-norm loss of the data-driven estimators of $\log \rho$ and the elasticity of $\rho$ are similar in magnitude to estimators with deterministic $J = 4$ or $J = 5$, and are several multiples smaller than those with larger $J$. Coverage of the fixed $J$ UCBs is generally too small when $J=4, 5$, whereas our data-driven UCBs deliver valid, albeit conservative, coverage. Our data-driven UCBs also demonstrate an improvement in terms of width relative to the deterministic-$J$ UCBs when $J$ is large enough (say $J=7, 11$) to ensure sufficient coverage. Rejection probabilities of a test of constant elasticity based on our data-driven UCBs for the elasticity of $\rho$ are also similar to those reported in Table~\ref{tab:trade.lognormal.deriv}. Figure~\ref{fig:trade.lognormal.d.nofe} presents plots of estimates and UCBs when we treat the FEs as zero using the same sample as Figure~\ref{fig:trade.lognormal.deriv} (where FEs were estimated in the first-stage). The estimates and UCBs reported in these figures are virtually identical, indicating  first-stage estimation of  FEs is innocuous.

\begin{table}[t]
\begin{center}
\caption{\label{tab:trade.lognormal.h.nofe} Simulation Results for Estimating $\log \rho$, Log-normal Design, no FEs}
{\small
\begin{tabular}{ccccccccccccc} \hline \hline \\[-10pt]
  & & \multicolumn{2}{c}{Data-driven}  & & \multicolumn{8}{c}{Deterministic} \\\cline{3-4} \cline{6-13} \\[-10pt]
  & & & & & \multicolumn{2}{c}{$J = 4$} & \multicolumn{2}{c}{$J = 5$} & \multicolumn{2}{c}{$J = 7$} & \multicolumn{2}{c}{$J = 11$} \\ \\[-10pt] \hline \\[-10pt]
  \multicolumn{13}{c}{Sup-norm Loss} \\ \\[-10pt]
$n$ & & mean & med. & & mean & med. & mean & med. & mean & med. & mean & med. \\ \\[-10pt]
 \phantom{0}761 & & 0.180 & 0.146 & & 0.166 & 0.142 & 0.184 & 0.159 & 0.361 & 0.302 & 0.670 & 0.609 \\
           1522 & & 0.120 & 0.099 & & 0.113 & 0.096 & 0.125 & 0.111 & 0.265 & 0.218 & 0.584 & 0.537 \\
           3044 & & 0.088 & 0.072 & & 0.080 & 0.069 & 0.087 & 0.080 & 0.196 & 0.163 & 0.510 & 0.468 \\
           6088 & & 0.068 & 0.053 & & 0.058 & 0.051 & 0.063 & 0.058 & 0.147 & 0.121 & 0.456 & 0.408 \\  \\[-10pt] \hline \\[-10pt]
  \multicolumn{13}{c}{UCB Coverage} \\ \\[-10pt]
 & & 90\% & 95\% & & 90\% & 95\% & 90\% & 95\% & 90\% & 95\% & 90\% & 95\% \\ \\[-10pt]
 \phantom{0}761 & & 0.992 & 0.996 & & 0.885 & 0.937 & 0.879 & 0.938 & 0.909 & 0.960 & 0.937 & 0.972 \\
           1522 & & 0.996 & 0.998 & & 0.898 & 0.940 & 0.893 & 0.944 & 0.915 & 0.955 & 0.964 & 0.985 \\
           3044 & & 0.998 & 1.000 & & 0.875 & 0.937 & 0.903 & 0.948 & 0.933 & 0.964 & 0.956 & 0.988 \\
           6088 & & 0.999 & 1.000 & & 0.864 & 0.936 & 0.880 & 0.949 & 0.913 & 0.958 & 0.951 & 0.984 \\  \\[-10pt] \hline \\[-10pt]
  & & & & & \multicolumn{8}{c}{95\% UCB Relative Width (Deterministic/Data-driven)} \\ \cline{6-13} \\[-10pt]
  & &  &  & & mean & med. & mean & med. & mean & med. & mean & med. \\ \\[-10pt]
 \phantom{0}761 & & & & & 0.660 & 0.675 & 0.693 & 0.695 & 1.576 & 1.422 & 2.606 & 2.457 \\
           1522 & & & & & 0.668 & 0.681 & 0.692 & 0.696 & 1.784 & 1.700 & 3.425 & 3.168 \\
           3044 & & & & & 0.667 & 0.682 & 0.683 & 0.692 & 1.904 & 1.855 & 4.310 & 4.029 \\
           6088 & & & & & 0.663 & 0.684 & 0.673 & 0.691 & 2.007 & 2.008 & 5.637 & 5.055 \\  \\[-10pt] \hline
\end{tabular}
}
\end{center}
\vskip -14pt
\end{table}

\begin{table}[t]
\begin{center}
\caption{\label{tab:trade.lognormal.d.nofe} Simulation Results for Estimating the Elasticity of $\rho$, Log-normal Design, no~FEs}
{\small
\begin{tabular}{ccccccccccccc} \hline \hline \\[-10pt]
  & & \multicolumn{2}{c}{Data-driven}  & & \multicolumn{8}{c}{Deterministic} \\\cline{3-4} \cline{6-13} \\[-10pt]
  & & & & & \multicolumn{2}{c}{$J = 4$} & \multicolumn{2}{c}{$J = 5$} & \multicolumn{2}{c}{$J = 7$} & \multicolumn{2}{c}{$J = 11$} \\ \\[-10pt] \hline \\[-10pt]
  \multicolumn{13}{c}{Sup-norm Loss} \\ \\[-10pt]
$n$ & & mean & med. & & mean & med. & mean & med. & mean & med. & mean & med. \\ \\[-10pt]
 \phantom{0}761 & & 0.238 & 0.179 & & 0.204 & 0.176 & 0.308 & 0.272 & 0.569 & 0.468 & 2.000 & 1.825 \\
           1522 & & 0.161 & 0.128 & & 0.141 & 0.124 & 0.213 & 0.187 & 0.374 & 0.334 & 1.785 & 1.614 \\
           3044 & & 0.131 & 0.094 & & 0.104 & 0.091 & 0.146 & 0.135 & 0.279 & 0.253 & 1.544 & 1.362 \\
           6088 & & 0.113 & 0.071 & & 0.075 & 0.068 & 0.104 & 0.096 & 0.200 & 0.179 & 1.361 & 1.236 \\ \\[-10pt] \hline \\[-10pt]
  \multicolumn{13}{c}{UCB Coverage} \\ \\[-10pt]
 & & 90\% & 95\% & & 90\% & 95\% & 90\% & 95\% & 90\% & 95\% & 90\% & 95\% \\ \\[-10pt]
 \phantom{0}761 & & 0.993 & 0.998 & & 0.895 & 0.939 & 0.880 & 0.931 & 0.894 & 0.941 & 0.939 & 0.974 \\
           1522 & & 0.995 & 0.998 & & 0.878 & 0.933 & 0.892 & 0.933 & 0.913 & 0.953 & 0.956 & 0.985 \\
           3044 & & 0.997 & 1.000 & & 0.846 & 0.918 & 0.891 & 0.941 & 0.918 & 0.958 & 0.957 & 0.984 \\
           6088 & & 0.993 & 0.994 & & 0.820 & 0.894 & 0.891 & 0.943 & 0.915 & 0.960 & 0.955 & 0.984 \\ \\[-10pt] \hline \\[-10pt]
  & & \multicolumn{2}{c}{Frequency} & & \multicolumn{8}{c}{95\% UCB Relative Width (Deterministic/Data-driven)} \\\cline{3-4} \cline{6-13} \\[-10pt]
  & & \multicolumn{2}{c}{reject} & & mean & med. & mean & med. & mean & med. & mean & med. \\ \\[-10pt]
 \phantom{0}761 & & \multicolumn{2}{c}{0.067} & & 0.633 & 0.651 & 0.936 & 0.932 & 1.760 & 1.570 & \phantom{1}6.464 & \phantom{1}6.140 \\
           1522 & & \multicolumn{2}{c}{0.304} & & 0.640 & 0.657 & 0.920 & 0.918 & 1.756 & 1.615 & \phantom{1}8.407 & \phantom{1}8.177 \\
           3044 & & \multicolumn{2}{c}{0.819} & & 0.642 & 0.659 & 0.895 & 0.903 & 1.752 & 1.623 & 10.406 & 10.141 \\
           6088 & & \multicolumn{2}{c}{0.966} & & 0.636 & 0.660 & 0.869 & 0.895 & 1.729 & 1.694 & 12.832 & 12.754 \\  \\[-10pt] \hline
\end{tabular}
}
\parbox{\textwidth}{\small \emph{Note:}
Column ``reject'' reports the proportion of simulations in which constant functions are excluded from our data-driven 95\% UCBs for the true elasticity.}
\end{center}
\vskip -14pt
\end{table}

\begin{figure}[h!]
\begin{center}

\begin{subfigure}[t]{\textwidth}
\caption{Data-driven Estimates and UCBs}
\begin{center}
\vskip -10pt
\begin{subfigure}[t]{0.49\textwidth}
\begin{center}
\includegraphics[width=\linewidth]{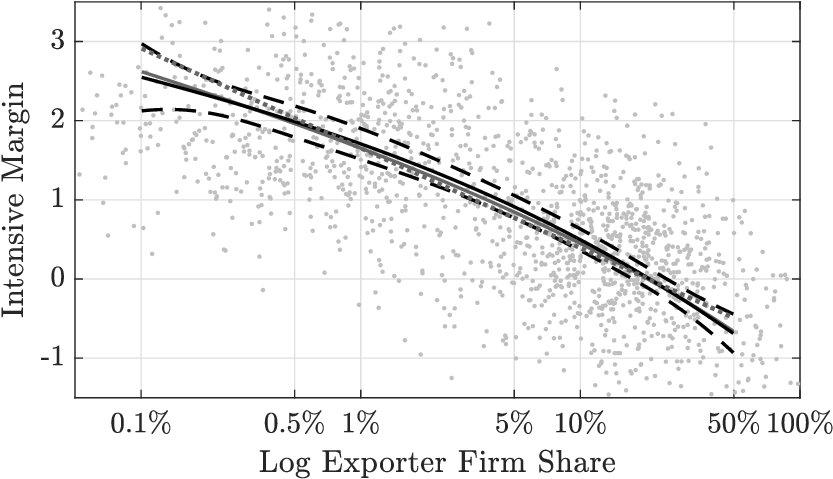}
\end{center}
\end{subfigure}
\begin{subfigure}[t]{0.49\textwidth}
\begin{center}
\includegraphics[width=\linewidth]{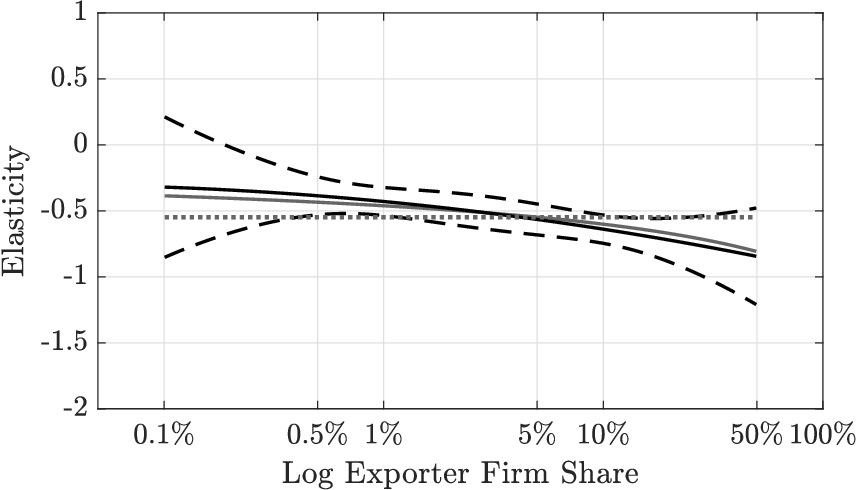}
\end{center}
\end{subfigure}

\end{center}
\end{subfigure}

\begin{subfigure}[t]{\textwidth}
\caption{Estimates and UCBs with $J=5$}
\begin{center}
\vskip -10pt
\begin{subfigure}[t]{0.49\textwidth}
\begin{center}
\includegraphics[width=\linewidth]{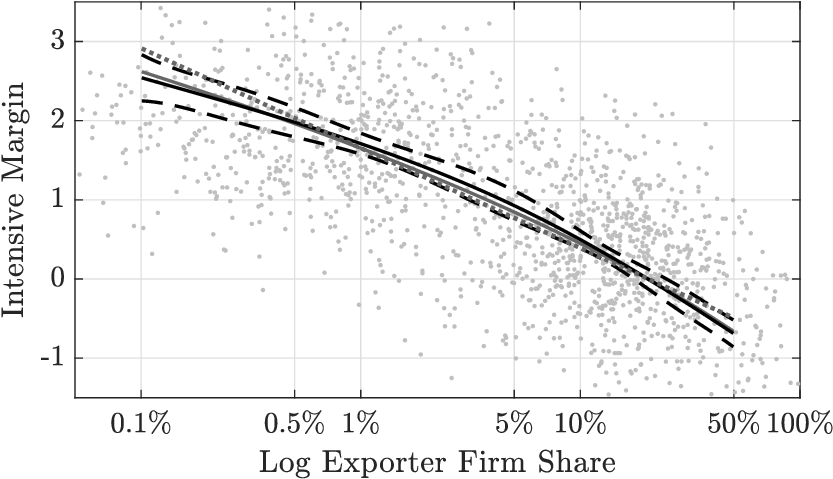}
\end{center}
\end{subfigure}
\begin{subfigure}[t]{0.49\textwidth}
\begin{center}
\includegraphics[width=\linewidth]{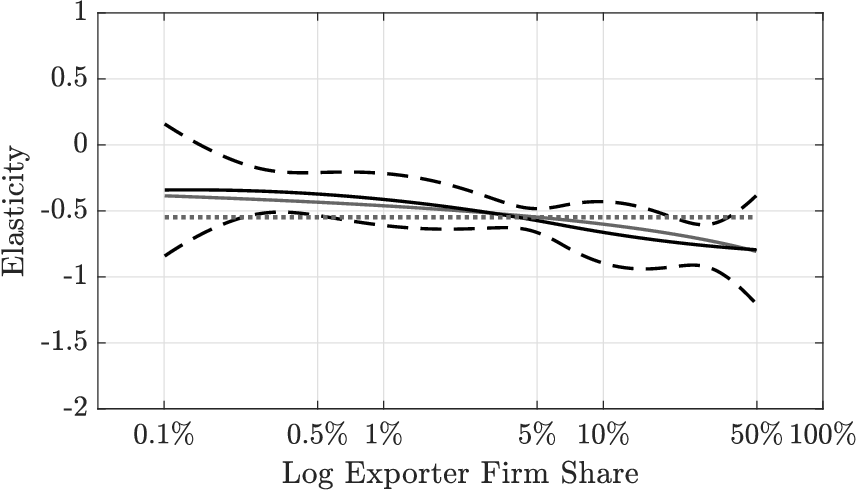}
\end{center}
\end{subfigure}

\end{center}
\end{subfigure}

\begin{subfigure}[t]{\textwidth}
\caption{Estimates and UCBs with $J = 7$}
\begin{center}
\vskip -10pt
\begin{subfigure}[t]{0.49\textwidth}
\begin{center}
\includegraphics[width=\linewidth]{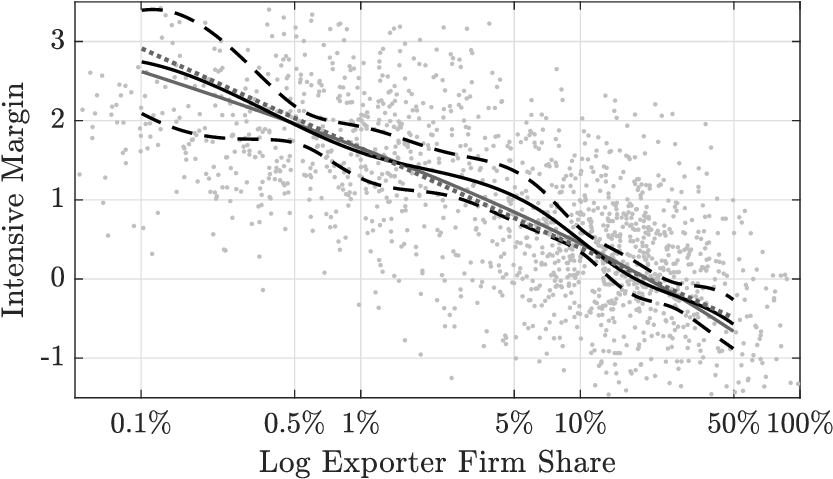}
\end{center}
\end{subfigure}
\begin{subfigure}[t]{0.49\textwidth}
\begin{center}
\includegraphics[width=\linewidth]{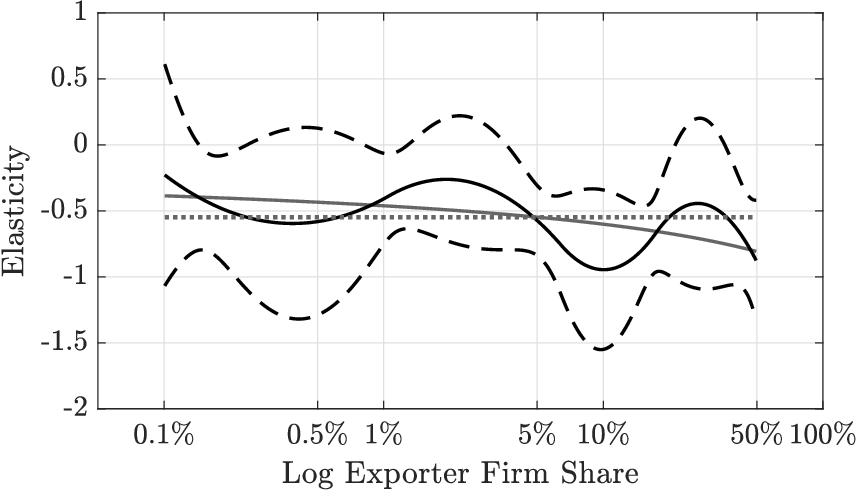}
\end{center}
\end{subfigure}

\end{center}
\end{subfigure}

\vskip -8pt

\caption{\label{fig:trade.lognormal.d.nofe} Log-normal design without fixed effects: Plots for a representative sample of size $1522$. Left panels correspond to the intensive margin, right panels correspond to its elasticity. \emph{Note:} Solid grey lines are the true curves; solid black lines are estimates; dashed black lines are 95\% UCBs; dotted grey lines are linear IV estimates.}
\end{center}
\vskip -20pt
\end{figure}

\medskip

\paragraph{Simulation Results for the Pareto Design.}

We now turn to the Pareto design in which  $\log \rho$ is linear and hence the elasticity of $\rho$ is constant. For brevity we just present in Table~\ref{tab:trade.pareto.deriv} the simulation results for estimating the elasticity of $\rho$. We adjust for first-stage estimation of exporter and importer FEs, as in the empirical application. The optimal choice is $J = 4$, which is the smallest dimension of a cubic B-spline basis. There is no bias with $J = 4$ because the basis functions span cubic functions. As can be seen, the maximal error in estimating the elasticity of $\rho$ using $\tilde J$ is very close to the estimator with fixed $J = 4$. Data-driven UCBs again demonstrate valid but conservative coverage for the elasticity. UCBs with fixed $J = 4$ demonstrate coverage close to (but still slightly under) nominal coverage with $n = 6088$. The UCBs with fixed $J=4$ are narrower (by about 36\%) than our data-driven UCBs as they do not account for potential approximation bias whereas our bands do. Of course, in a real data application the researcher doesn't know whether the true elasticity is constant, and therefore whether the UCBs with fixed $J = 4$ is sufficient to guarantee coverage. Figure~\ref{fig:trade.pareto.deriv} presents plots for a representative sample of size $1522$, again implementing our procedures as described in Section~\ref{s:application}. With our data-driven choice $\tilde{J}=4$, our nonparametric IV estimate of $\log \rho$ is very close to linear and our estimated elasticity is very close to the true, constant elasticity.

\begin{table}[t]
\begin{center}
\caption{\label{tab:trade.pareto.deriv} Simulation Results for Estimating the Elasticity of $\rho$, Pareto Design}
{\small
\begin{tabular}{ccccccccccccc} \hline \hline \\[-10pt]
  & & \multicolumn{2}{c}{Data-driven}  & & \multicolumn{8}{c}{Deterministic} \\\cline{3-4} \cline{6-13} \\[-10pt]
  & & & & & \multicolumn{2}{c}{$J = 4$} & \multicolumn{2}{c}{$J = 5$} & \multicolumn{2}{c}{$J = 7$} & \multicolumn{2}{c}{$J = 11$} \\ \\[-10pt] \hline \\[-10pt]
  \multicolumn{13}{c}{Sup-norm Loss} \\ \\[-10pt]
$n$ & & mean & med. & & mean & med. & mean & med. & mean & med. & mean & med. \\ \\[-10pt]
 \phantom{0}761 & & 0.228 & 0.164 & & 0.183 & 0.157 & 0.283 & 0.256 & 0.519 & 0.429 & 1.866 & 1.699 \\
           1522 & & 0.163 & 0.118 & & 0.125 & 0.115 & 0.193 & 0.172 & 0.343 & 0.305 & 1.620 & 1.469 \\
           3044 & & 0.125 & 0.085 & & 0.092 & 0.082 & 0.134 & 0.122 & 0.254 & 0.229 & 1.394 & 1.236 \\
           6088 & & 0.095 & 0.058 & & 0.061 & 0.056 & 0.093 & 0.085 & 0.180 & 0.161 & 1.212 & 1.092 \\ \\[-10pt] \hline \\[-10pt]
  \multicolumn{13}{c}{UCB Coverage} \\ \\[-10pt]
  & & 90\% & 95\% & & 90\% & 95\% & 90\% & 95\% & 90\% & 95\% & 90\% & 95\% \\ \\[-10pt]
 \phantom{0}761 & & 0.994 & 0.998 & & 0.884 & 0.934 & 0.854 & 0.916 & 0.877 & 0.926 & 0.906 & 0.958 \\
           1522 & & 0.993 & 0.997 & & 0.884 & 0.937 & 0.873 & 0.936 & 0.892 & 0.946 & 0.939 & 0.976 \\
           3044 & & 0.998 & 1.000 & & 0.878 & 0.937 & 0.872 & 0.932 & 0.891 & 0.949 & 0.940 & 0.975 \\
           6088 & & 0.994 & 0.998 & & 0.890 & 0.942 & 0.886 & 0.939 & 0.901 & 0.954 & 0.944 & 0.982 \\\\[-10pt] \hline \\[-10pt]
  & & & & & \multicolumn{8}{c}{95\% UCB Relative Width (Deterministic/Data-driven)} \\ \cline{6-13} \\[-10pt]
  & &  &  & & mean & med. & mean & med. & mean & med. & mean & med. \\ \\[-10pt]
 \phantom{0}761 & & & & & 0.626 & 0.650 & 0.928 & 0.934 & 1.760 & 1.553 & \phantom{1}6.439 & \phantom{1}6.149 \\
           1522 & & & & & 0.633 & 0.656 & 0.909 & 0.920 & 1.753 & 1.619 & \phantom{1}8.315 & \phantom{1}8.110 \\
           3044 & & & & & 0.637 & 0.658 & 0.890 & 0.906 & 1.754 & 1.642 & 10.373 & 10.128 \\
           6088 & & & & & 0.635 & 0.660 & 0.872 & 0.899 & 1.738 & 1.709 & 12.901 & 12.862 \\ \\[-10pt] \hline
\end{tabular}
}
\end{center}
\vskip -14pt
\end{table}

\begin{figure}[h!]
\begin{center}

\begin{subfigure}[t]{\textwidth}
\caption{Data-driven Estimates and UCBs}
\begin{center}
\vskip -10pt
\begin{subfigure}[t]{0.49\textwidth}
\begin{center}
\includegraphics[width=\linewidth]{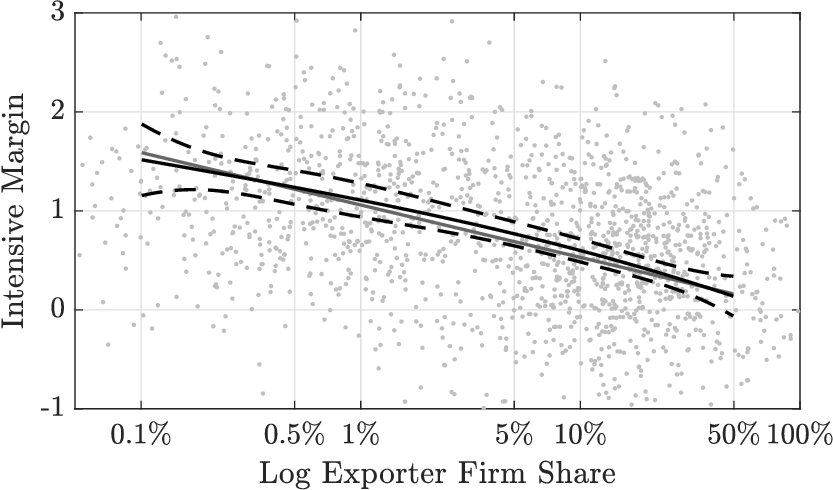}
\end{center}
\end{subfigure}
\begin{subfigure}[t]{0.49\textwidth}
\begin{center}
\includegraphics[width=\linewidth]{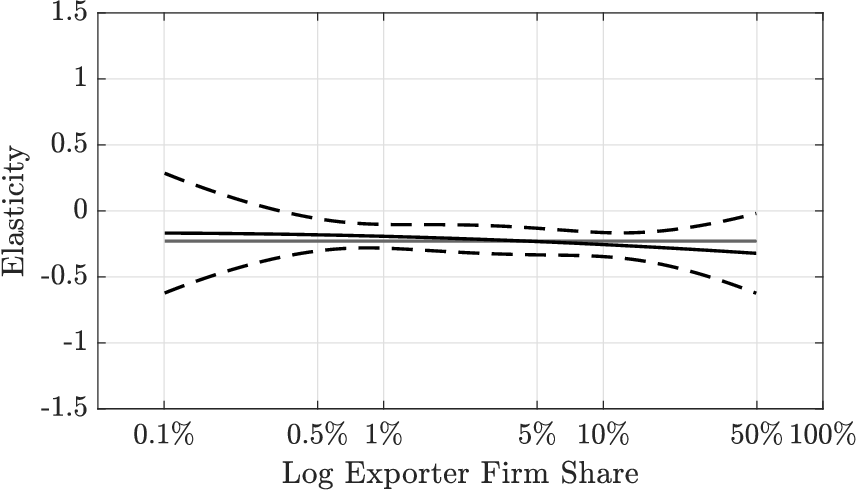}
\end{center}
\end{subfigure}

\end{center}
\end{subfigure}

\begin{subfigure}[t]{\textwidth}
\caption{Estimates and UCBs with $J = 5$}
\begin{center}
\vskip -10pt
\begin{subfigure}[t]{0.49\textwidth}
\begin{center}
\includegraphics[width=\linewidth]{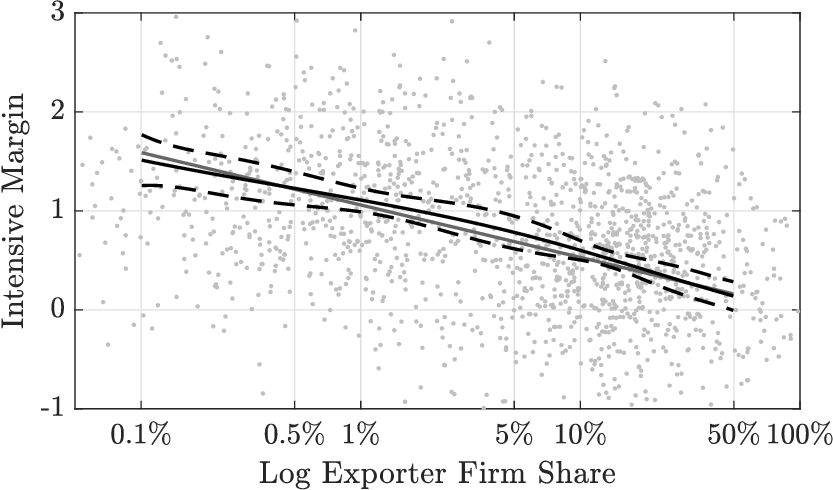}
\end{center}
\end{subfigure}
\begin{subfigure}[t]{0.49\textwidth}
\begin{center}
\includegraphics[width=\linewidth]{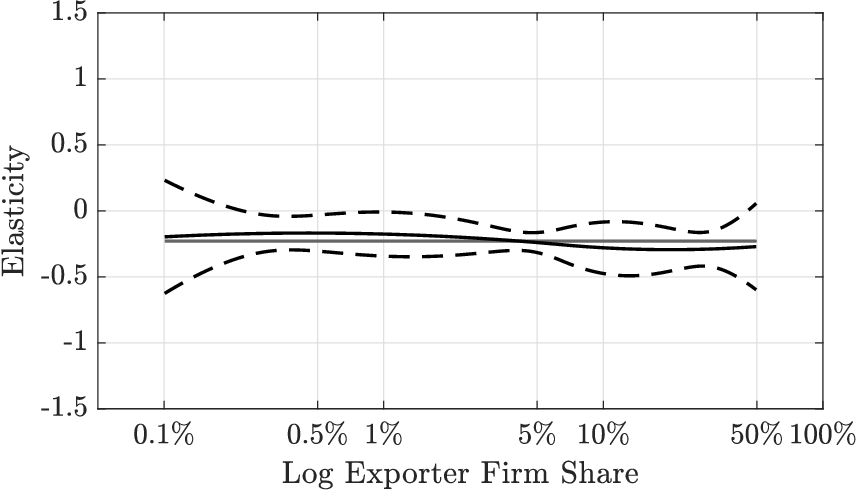}
\end{center}
\end{subfigure}

\end{center}
\end{subfigure}

\begin{subfigure}[t]{\textwidth}
\caption{Estimates and UCBs with $J = 7$}
\begin{center}
\vskip -10pt
\begin{subfigure}[t]{0.49\textwidth}
\begin{center}
\includegraphics[width=\linewidth]{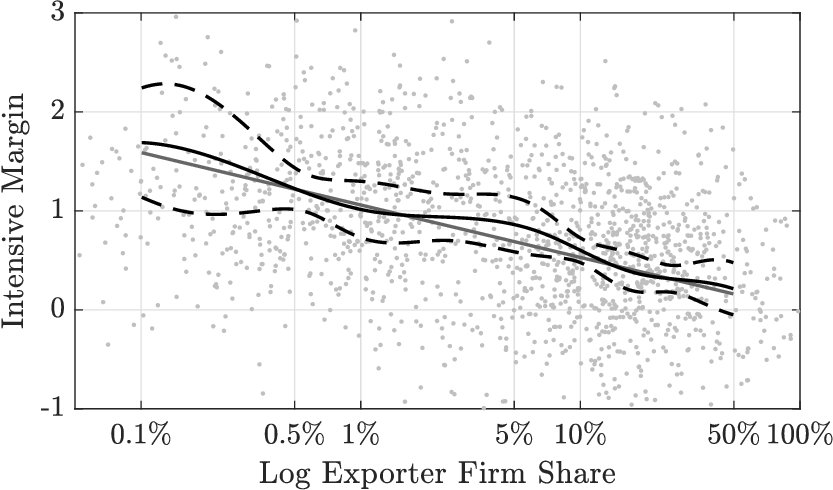}
\end{center}
\end{subfigure}
\begin{subfigure}[t]{0.49\textwidth}
\begin{center}
\includegraphics[width=\linewidth]{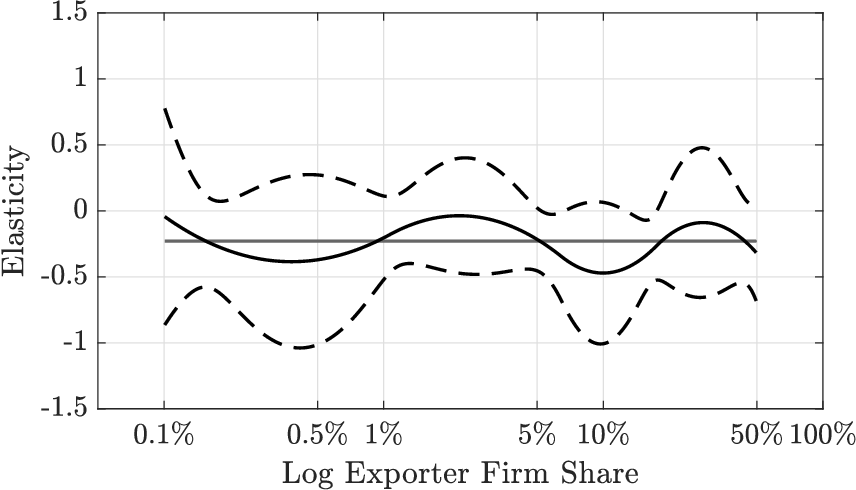}
\end{center}
\end{subfigure}

\end{center}
\end{subfigure}

\vskip -8pt

\caption{\label{fig:trade.pareto.deriv} Pareto design (with first-stage estimation of fixed effects): Plots for a representative sample of size $1522$. Left panels correspond to the intensive margin, right panels correspond to its elasticity. \emph{Note:} Solid grey lines are the true curves; solid black lines are estimates; dashed black lines are 95\% UCBs.}
\end{center}
\vskip -20pt
\end{figure}

\let\oldbibliography\thebibliography
\renewcommand{\thebibliography}[1]{\oldbibliography{#1}
\setlength{\itemsep}{0pt}}

{\singlespacing
\putbib
}

\end{bibunit}

\begin{bibunit}

\newpage
\clearpage
\pagenumbering{arabic}\renewcommand{\thepage}{\arabic{page}}

\begin{center}
{\Large Online Appendix to ``Adaptive Estimation and Uniform Confidence Bands for Nonparametric Structural Functions and Elasticities''}

\vskip 24pt

{\large Xiaohong Chen \quad \quad Timothy Christensen \quad \quad Sid Kankanala}

\end{center}

\vskip 8pt

\section{Additional Simulation: Engel Curves}\label{appsec:engel}

In this appendix we present additional simulation results for estimating a nonparametric structural function in an empirically calibrated Engel curve setting. The design is based on the British Family Expenditure Survey data used in \cite{blundell2007semi}. We draw household expenditure $X$ and household income $W$ from a bivariate normal density with correlation $\rho = 0.52$, which is the sample correlation of the expenditure and income data used in \cite{blundell2007semi}. We then transform $X$ and $W$ to have Uniform$[0,1]$ marginals using their respective inverse marginal CDFs. As a consequence, $X$ and $W$ are linked via a Gaussian copula and the design is severely ill-posed.\footnote{This follows from, e.g., \cite{Beare2010}, equation (3.3).} We then set $h_0(x) = \Phi(5 x - 2. 5)$ and set $u = h_0(X) - \E[h_0(X)|W] + v$ for $v \sim N(0,0.01)$. The implementation is the same as the other Monte Carlos from Section~\ref{s:mcs}. For each simulated data set we compute our data-driven estimator $\hat h_{\tilde J}$ and UCBs from (\ref{band}). We compare these with estimators and UCBs using deterministic choices of sieve dimensions for $J = 4$, $5$, $7$, and $11$ (the first few dimensions over which our procedure searches). We again use a cubic B-spline basis to approximate $h_0$ and a quartic B-spline basis for the reduced form.

Turning first to the simulation results presented in Table~\ref{tab:engel_mc}, we see that the average sup-norm loss of our data-driven estimator is similar to that of an estimator $\hat h_J$ for deterministic $J$ with $J = 4$ and several multiples smaller than that with $J = 5$, $7$, or $11$. This is to be expected, as the design is severely ill-posed and the true function is very smooth, so a very small choice of $J$ is appropriate. Of course, in practice the researcher does not know the degree of ill-posedness or the degree of smoothness of the structural function.

The second panel of Table~\ref{tab:engel_mc} shows our data-driven UCBs have valid, albeit conservative, coverage across all sample sizes. By contrast, undersmoothed UCBs with $J = 4$ and $J = 5$ under-cover for $n = 2500$, $5000$, and $10000$. Undersmoothed UCBs with $J = 7$ have valid but conservative coverage, but these are 40\% (with $n = 1250$) to 250\% (with $n = 10000$) wider than our data-driven UCBs. It is important to note that although the design is severely ill-posed, we are reporting coverage of our UCBs (\ref{band}). In each simulated data set we have $\hat J = \tilde J$ irrespective of the sample size $n$, so the critical value is effectively $z^*_{1-\alpha} + \hat A \theta^*_{1-\hat \alpha}$. While Theorem~\ref{confmild} does not formally establish coverage guarantees of this band in the severely ill-posed case, these simulation results show that the band nevertheless has good coverage in this empirically relevant design.

\begin{table}[t]
\begin{center}
\caption{\label{tab:engel_mc} Simulation Results for the Engel Curve Design.}
{\small
\begin{tabular}{ccccccccccccc} \hline \hline \\[-10pt]
  & & \multicolumn{2}{c}{Data-driven}  & & \multicolumn{8}{c}{Deterministic} \\\cline{3-4} \cline{6-13} \\[-10pt]
  & & & & & \multicolumn{2}{c}{$J = 4$} & \multicolumn{2}{c}{$J = 5$} & \multicolumn{2}{c}{$J = 7$} & \multicolumn{2}{c}{$J = 11$} \\ \\[-10pt] \hline \\[-10pt]
  \multicolumn{13}{c}{Sup-norm Loss} \\ \\[-10pt]
$n$ & & mean & med. & & mean & med. & mean & med. & mean & med. & mean & med. \\ \\[-10pt]
 \phantom{1}1250 & & 0.221 & 0.183 & & 0.218 & 0.180 & 0.337 & 0.287 & 0.450 & 0.400 & 0.558 & 0.488 \\ 
 \phantom{1}2500 & & 0.167 & 0.139 & & 0.164 & 0.138 & 0.285 & 0.240 & 0.417 & 0.369 & 0.526 & 0.468 \\ 
 \phantom{1}5000 & & 0.115 & 0.094 & & 0.113 & 0.094 & 0.233 & 0.197 & 0.361 & 0.318 & 0.484 & 0.432 \\ 
10000 & & 0.083 & 0.068 & & 0.080 & 0.068 & 0.173 & 0.148 & 0.322 & 0.299 & 0.448 & 0.414 \\ \\[-10pt] \hline \\[-10pt]
  \multicolumn{13}{c}{UCB Coverage} \\ \\[-10pt]
 & & 90\% & 95\% & & 90\% & 95\% & 90\% & 95\% & 90\% & 95\% & 90\% & 95\% \\ \\[-10pt]
 \phantom{1}1250 & & 0.998 & 0.999 & & 0.917 & 0.961 & 0.903 & 0.952 & 0.934 & 0.972 & 0.916 & 0.968 \\ 
\phantom{1}2500 & & 0.998 & 0.999 & & 0.868 & 0.931 & 0.867 & 0.943 & 0.950 & 0.980 & 0.941 & 0.982 \\ 
\phantom{1}5000 & & 0.998 & 0.999 & & 0.833 & 0.896 & 0.884 & 0.939 & 0.967 & 0.989 & 0.968 & 0.987 \\ 
10000 & & 0.991 & 0.994 & & 0.700 & 0.826 & 0.826 & 0.904 & 0.956 & 0.988 & 0.964 & 0.992 \\  \\[-10pt] \hline \\[-10pt]
  & & & & & \multicolumn{8}{c}{95\% UCB Relative Width (Deterministic/Data-driven)} \\ \cline{6-13} \\[-10pt]
  & &  &  & & mean & med. & mean & med. & mean & med. & mean & med. \\ \\[-10pt]
 \phantom{1}1250 & & & & & 0.653 & 0.658 & 1.056 & 0.978 & 1.431 & 1.351 & 1.798 & 1.718 \\ 
 \phantom{1}2500 & & & & & 0.655 & 0.661 & 1.213 & 1.120 & 1.797 & 1.692 & 2.372 & 2.270 \\ 
 \phantom{1}5000 & & & & & 0.656 & 0.661 & 1.458 & 1.331 & 2.219 & 2.107 & 3.082 & 2.991 \\ 
10000 & & & & & 0.658 & 0.664 & 1.577 & 1.478 & 2.712 & 2.574 & 3.962 & 3.792  \\\\[-10pt] \hline
\end{tabular}
}
\end{center}
\vskip -14pt
\end{table}

\begin{figure}[p]
\begin{center}

\begin{subfigure}[t]{\textwidth}
\caption{Data-driven Estimates and UCBs}
\begin{center}
\vskip -10pt
\begin{subfigure}[t]{0.45\textwidth}
\begin{center}
\includegraphics[width=\linewidth]{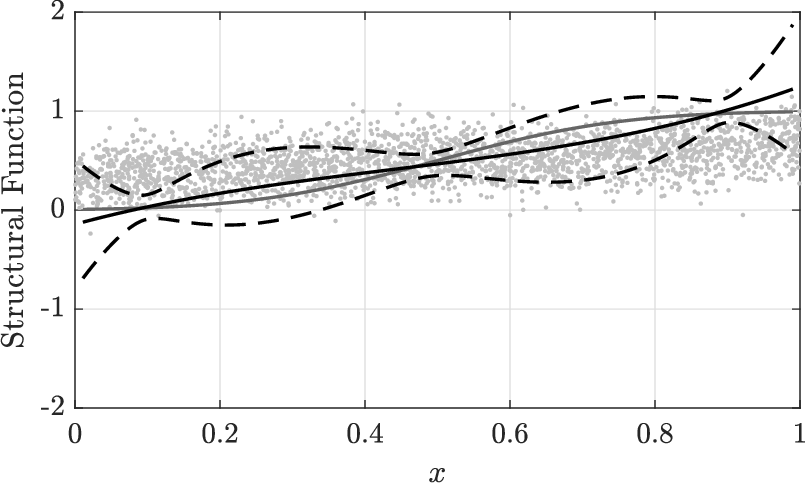}
\end{center}
\end{subfigure}
\begin{subfigure}[t]{0.45\textwidth}
\begin{center}
\includegraphics[width=\linewidth]{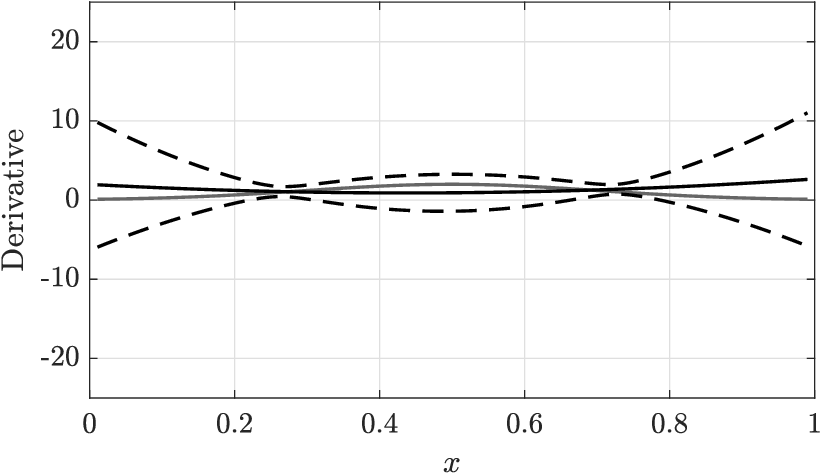}
\end{center}
\end{subfigure}
\end{center}
\end{subfigure}
\vskip -4pt

\begin{subfigure}[t]{\textwidth}
\caption{Estimates and UCBs with $J = 5$}
\begin{center}
\vskip -10pt
\begin{subfigure}[t]{0.45\textwidth}
\begin{center}
\includegraphics[width=\linewidth]{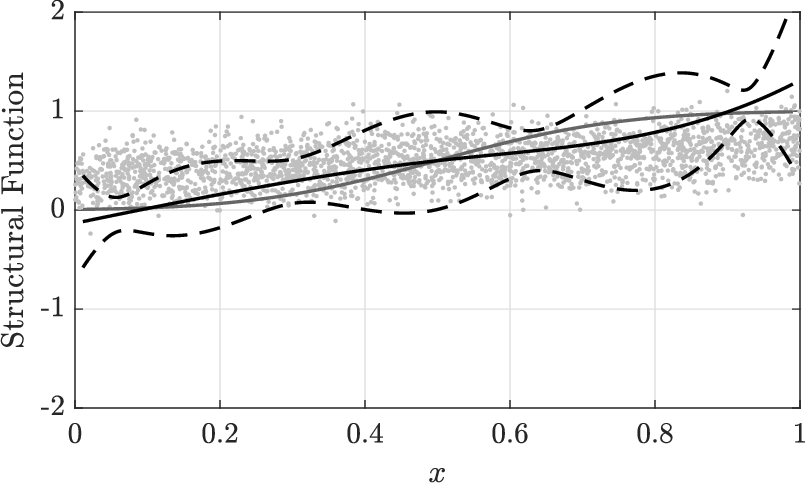}
\end{center}
\end{subfigure}
\begin{subfigure}[t]{0.45\textwidth}
\begin{center}
\includegraphics[width=\linewidth]{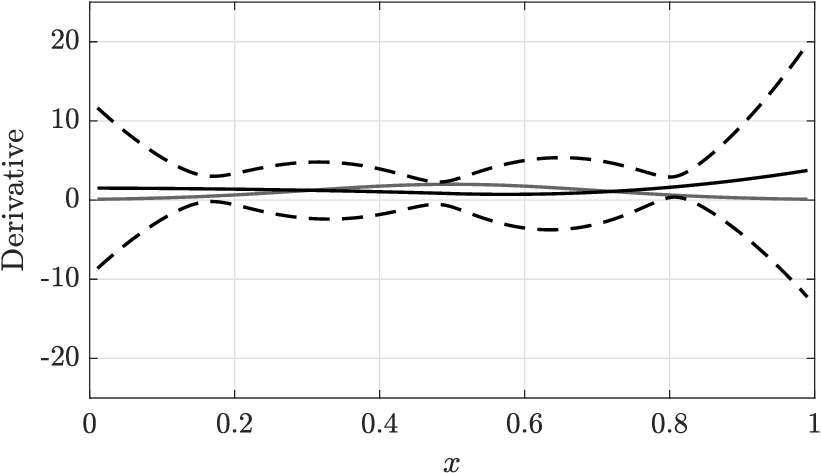}
\end{center}
\end{subfigure}
\end{center}
\end{subfigure}
\vskip -4pt

\begin{subfigure}[t]{\textwidth}
\caption{Estimates and UCBs with $J = 7$}
\begin{center}
\vskip -10pt
\begin{subfigure}[t]{0.45\textwidth}
\begin{center}
\includegraphics[width=\linewidth]{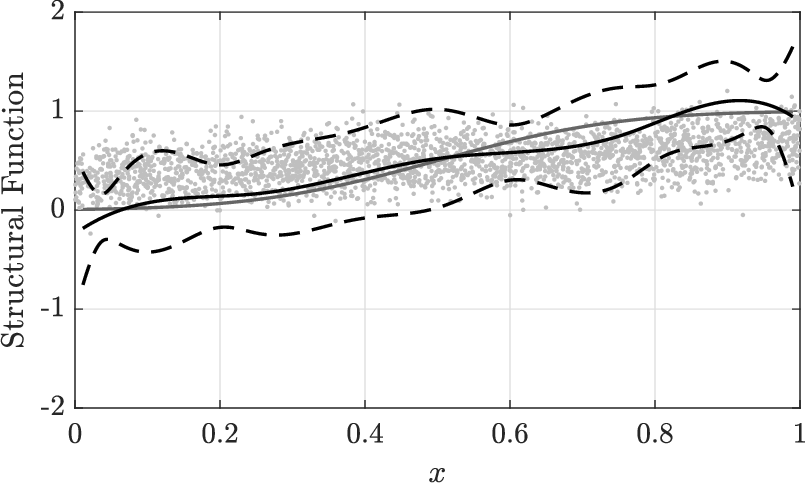}
\end{center}
\end{subfigure}
\begin{subfigure}[t]{0.45\textwidth}
\begin{center}
\includegraphics[width=\linewidth]{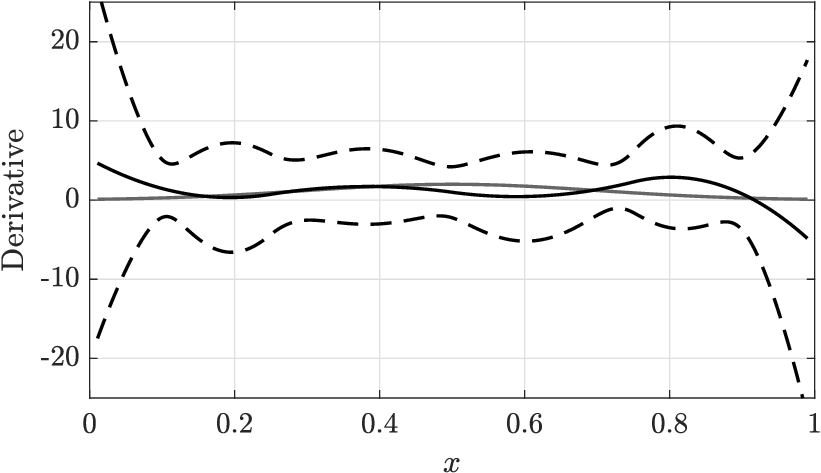}
\end{center}
\end{subfigure}
\end{center}
\end{subfigure}
\vskip -4pt

\begin{subfigure}[t]{\textwidth}
\caption{Data-driven Estimates and UCBs for the Conditional Mean of $Y$ given $X$}
\begin{center}
\vskip -10pt
\begin{subfigure}[t]{0.45\textwidth}
\begin{center}
\includegraphics[width=\linewidth]{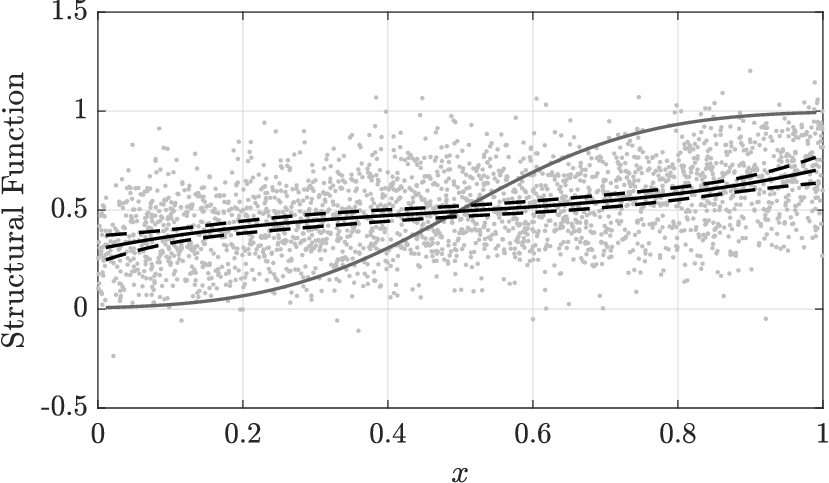}
\end{center}
\end{subfigure}
\begin{subfigure}[t]{0.45\textwidth}
\begin{center}
\includegraphics[width=\linewidth]{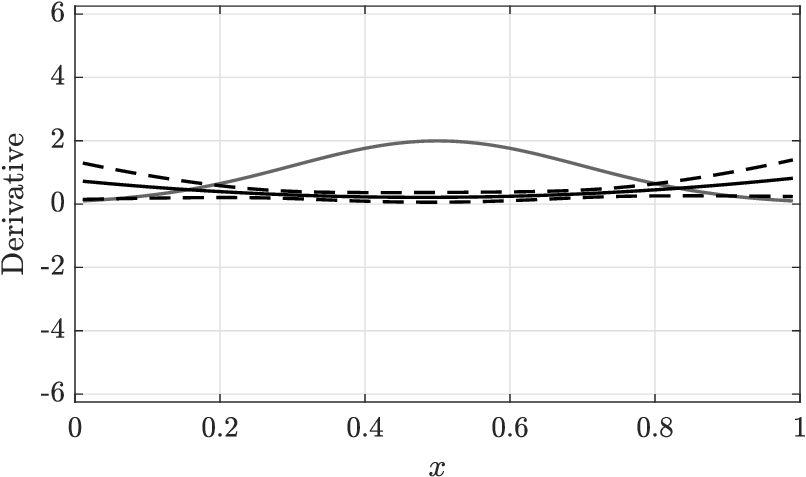}
\end{center}
\end{subfigure}
\end{center}
\end{subfigure}

\vskip -4pt

\caption{\label{fig:engel_mc} Engel curve design: Plots for a sample of size $n = 2500$. Left panels correspond to the structural function, right panels correspond to its derivative. \emph{Note:} Solid grey lines are the true structural function and derivative; solid black lines are estimates, dashed black lines are 95\% UCBs.}
\vskip -15pt
\end{center}
\end{figure}

Figure~\ref{fig:engel_mc} presents plots of data-driven estimates and UCBs for $h_0$ and its derivative for a sample of size 2500, alongside deterministic-$J$ estimates and UCBs. In this sample, $\tilde J = 4$ and our data-driven UCBs contain the true structural function. As with the other simulations, the data-driven bands are narrower and more accurately convey the shape of $h_0$ than the $J = 7$ bands, which are much more wiggly. Our bands are also slightly narrower than the $J = 5$ bands. Panel~(d) of Figure~\ref{fig:npiv_mc} also presents data-driven estimates and UCBs for the conditional mean of $Y$ given $X$. Evidently, the true structural function falls outside the UCBs for the conditional mean function over almost all of the support of $X$, again highlighting the importance of estimating $h_0$ using IV methods in this design.

\section{Basis Functions and H\"older Classes} \label{ax:basis}

Let $\Psi_J$ denote the closed linear subspace of $L^2_X$ spanned by a basis $\{\psi_{J1},\ldots,\psi_{JJ}\}$. We use the following notation for vectors and matrices formed from the basis functions
\begin{align*}
 \psi^J_x & = (\psi_{J1}(x),\ldots,\psi_{JJ}(x))' \,, &
 b^K_w & = (b_{K1}(w),\ldots,b_{KK}(w))' \,, \\
 \zeta_{\psi,J} & = \sup_{x \in [0,1]^d} \| G_{\psi,J}^{-1/2}  \psi^J_x \|_{\ell^2} \,, &
 \zeta_{b,J} & =  \sup_{w \in [0,1]^{d_w}}  \|  G_{b,J}^{-1/2} b^{K(J)}_w \|_{\ell^2} \,, \\
 {G}_{\psi,J} & = \E \big[ \psi^J_X  (\psi^J_X)'  \big] \,, &
 {G}_{b,J} & = \E \big[ b^{K(J)}_W (b^{K(J)}_W)' \big]  \,, \\
 {S}_{J} & = \E \big[ b^{K(J)}_W (\psi^J_X)' \big] \,, &
 {S}_{J}^o & = G_{b,J}^{-1/2} \E \big[ b^{K(J)}_W (\psi^J_X)' \big] G_{\psi,J}^{-1/2} \,.
\end{align*}
Let $s_J$ be the smallest singular value of $(G_{b,J})^{-1/2} S_J (G_{\psi,J})^{-1/2}$. By  Lemma A.1 of \cite{chen2018optimal}, under Assumptions \ref{a-data} and \ref{a-approx}(i) there is a finite positive constant $a_\tau$ such that
\begin{equation} \label{eq:eval_cc}
  a_\tau^{-1} s_J^{-1} \leq \tau_J \leq s_J^{-1}~~~\text{for all}~~J \in \T~.
\end{equation}

\subsection{B-splines}\label{ax:bspline}

The construction of univariate B-spline bases supported on $[0,1]$ follows Chapter 12.3 of \cite{devore1993constructive}. The basis is characterized by an \emph{order} $r \in \mb N$ (or \emph{degree} $r-1$) and a \emph{resolution level} $l \in \mb N \cup \{0\}$. Let $N_r$ denote the $r$-fold convolution of the indicator function of the unit interval, $N_r = \mathbbm{1}_{[0,1]} * \cdots * \mathbbm{1}_{[0,1]}$ ($r$-times). A dyadic\footnote{This basis is equivalent to a B-spline basis with interior knots at $2^{-l},\ldots,1-2^{-l}$. This knot placement ensures bases are nested across different $l$ (equivalently, $J$). For irregularly spaced data, interior knots can be placed at the $2^{-l},\ldots,1-2^{-l}$ quantiles of the distribution of $X$.} B-spline basis on $[0,1]$ with resolution level $l$ and order $r$ is
\[
 \psi_{J_1j}(x) = N_r(2^l x + r - j) \,, \quad j = 1,\ldots,2^l + r - 1 =: J_1\,.
\]
In the multivariate case we take tensor products of univariate bases. A B-spline basis supported on $[0,1]^d$ of order $r$ and resolution level $l$ has dimension $J = (2^l + r - 1)^d$. The set of possible sieve dimensions $J$  is therefore $\mc T = \{(2^l + r - 1)^d : l \in \mb N \cup \{0\}\}$.

We now review properties of B-spline bases that are used in the technical arguments below.
The following Lemma summarizes Lemma E.2 of \cite{chen2018optimal}.

\begin{lemma}\label{lem:spline}
Let Assumption \ref{a-data}(i) hold. Then for $\psi^{J}(x)$ formed from tensor product B-splines, there are constants $C_\psi,a_\zeta > 0$ depending only on $a_f$ such that
(i) $\sup_{x \in [0,1]^d } \|  \psi^J(x) \|_{\ell^1}  \leq C_\psi$;
(ii) $C_\psi^{-1} J^{-1} \leq \lambda_{\min}(G_{\psi,J}) \leq \lambda_{\max}(G_{\psi,J}) \leq C_\psi J^{-1}$;
(iii) $\sqrt J \leq \zeta_{\psi,J} \leq a_\zeta \sqrt{J} $.
\end{lemma}

\begin{corollary}
Let Assumption \ref{a-data}(ii) hold. Then for $b^{K(J)}(w)$ formed from tensor product B-splines and $J \leq K(J) \lesssim J$, there are constants $C_b,a_\zeta > 0$ depending only on $a_f$ such that
(i) $\sup_{w \in [0,1]^{d_w} } \|  b^{K(J)}(w) \|_{\ell^1}  \leq C_b$;
(ii) $C_b^{-1} J^{-1} \leq \lambda_{\min}(G_{b,J})) \leq \lambda_{\max}(G_{b,J})) \leq C_b J^{-1}$;
(iii) $\sqrt J \leq \zeta_{b,J} \leq a_\zeta \sqrt{J} $.
\end{corollary}

We also use some continuity properties of B-splines in the proofs. Note that $N_r(\cdot)$ is Lipschitz with $r = 2$ and $r-2$ times continuously differentiable when $r > 2$. Hence, $\|G_{\psi,J}^{-1/2} \big( [\psi^J(x_1)] - [\psi^J(x_2)] \big) \|_{\ell^2} \leq C J^\omega \| x_1 - x_2\|_{\ell^2}^{\omega'}$ holds for some positive constants $C,\omega,\omega'$. The B-spline basis also satisfies a Bernstein inequality (or inverse estimate): $\|\partial^a f\|_\infty \lesssim J^{|a|/d} \|f\|_\infty$ holds for any $f \in \Psi_J$ and multi-index $a$ with $|a| < r-1$.

\subsection{CDV Wavelets}

The construction of CDV wavelet bases supported on $[0,1]$ is reviewed in Appendix~E.2 of \cite{chen2018optimal} and follows \cite{cohen1993wavelets}; see also chapter 4.3.5 of \cite{gine2016mathematical}. The basis is characterized by an \emph{order} $N \in \mb N$. Let $L$ denote the smallest integer for which $2^{L} \geq 2N$. For each \emph{resolution level} $l \geq L$, there are a total of $2^l$ basis functions.
In the multivariate case we generate bases supported on $[0,1]^d$ by taking tensor products of univariate bases. The set of possible $J$  is therefore $\mc T = \{2^{ld} : l = L+1, L+2, \ldots\}$.

We say that the CDV wavelet basis is $S$-regular if it is $S$ times continuously differentiable. A $S$-regular basis can always be chosen by choosing the order $N$ such that $0.18(N-1) \geq S$ \cite[Theorem 4.2.10(e)]{gine2016mathematical}. The regularity $S$ of the basis for the endogenous variable $X$ should be chosen such that $S > \ol p$, where $\ol p$ is the maximal assumed degree of smoothness for $h_0$. Equivalently, our procedures deliver adaptivity over any smoothness range $[\ul p,\ol p]$ with $S > \ol p > \ul p > d/2$ when implemented with a $S$-regular CDV wavelet basis for $X$. As with choosing the order $r$ of B-splines, choosing $S$ is analogous to choosing the order of a kernel in kernel-based nonparametric estimation.

CDV wavelet bases for the $d_w$-dimensional instrumental variable $W$ are constructed similarly, using a basis of regularity $S+1$. Given the resolution level $l$ for the basis for $X$, the resolution level for the basis for $W$ is $l_w = \lceil (l + q) d/d_w \rceil$ for some $q \in \mb N$. Linking $l_w$ to $l$ in this manner again defines a mapping $K(J)$ between the two bases that satisfies $\lim_{J \to \infty} K(J)/J = c \in [1,\infty)$. As with B-splines, we recommend that $q$ should be the second- or third-smallest value for which $K(J) \geq J$ holds for all $J$.

We now review properties of CDV wavelet bases that are used in the proofs below. The following Lemma summarizes Lemma E.4 of \cite{chen2018optimal}.

\begin{lemma}\label{lem:wavelet}
Let Assumption \ref{a-data}(i) hold. Then with $\psi^{J}(x)$ formed from tensor product CDV wavelets, there are constants $C_\psi,a_\zeta > 0$ depending only on $a_f$ such that
(i) $\sup_{x \in [0,1]^d } \|  \psi^J_x \|_{\ell^1}  \leq C_\psi \sqrt J$;
(ii) $C_\psi^{-1} \leq \lambda_{\min}(G_{\psi,J}) \leq \lambda_{\max}(G_{\psi,J}) \leq C_\psi$;
(iii) $\sqrt J \leq \zeta_{\psi,J} \leq a_\zeta \sqrt{J} $.
\end{lemma}

\begin{corollary}
Let Assumption \ref{a-data}(ii) hold. Then with $b^{K(J)}(w)$ formed from tensor product CDV wavelets and $J \leq K(J) \lesssim J$, there are constants $C_b,a_\zeta > 0$ depending only on $a_f$ such that
(i) $\sup_{w \in [0,1]^{d_w} } \|  b^{K(J)}_w \|_{\ell^1}  \leq C_b \sqrt J$;
(ii) $C_b^{-1} \leq \lambda_{\min}(G_{b,J}) \leq \lambda_{\max}(G_{b,J}) \leq C_b $;
(iii) $\sqrt J \leq \zeta_{b,J} \leq a_\zeta \sqrt{J} $.
\end{corollary}

We also use some continuity properties of CDV wavelets in the proofs. As the Daubechies wavelet functions are $S$ times continuously differentiable on their supports, it follows by Lemma \ref{lem:wavelet}(ii) that the basis functions are H\"older continuous, in the sense that $\|G_{\psi,J}^{-1/2} \big( [\psi_{x_1}^J] - [\psi_{x_2}^J] \big) \|_{\ell^2} \leq C J^\omega \| x_1 - x_2\|_{\ell^2}^{\omega'}$ holds for some positive constants $C,\omega,\omega'$. This basis also satisfies a Bernstein inequality (or inverse estimate): $\|\partial^a f\|_\infty \lesssim J^{|a|/d} \|f\|_\infty$ holds for any $f \in \Psi_J$ and multi-index $a$ with $|a| < S$.

\subsection{H\"older Classes} \label{ax:besov}

Let $B^p_{\infty,\infty} = \{ h \in L^\infty([0,1]^d) : \|h\|_{B_{\infty,\infty}^p} < \infty\}$ denote the H\"older space of smoothness $p$ where $\|\cdot\|_{B_{\infty,\infty}^p}$ denotes the H\"older norm of smoothness $p > 0$ (see \cite{gine2016mathematical}, pp. 370-1), and let $B_{\infty,\infty}^{p}(M) = \{h \in B_{\infty,\infty}^{p}  : \|h \|_{B_{\infty,\infty}^{p}} \leq M\}$ denote the H\"older ball of smoothness $p$ and radius $M$. For $p \not \in \mb N$, we have $h \in B^p_{\infty,\infty}$ if and only if
\[
 \|h\|_{C^{\lfloor p \rfloor}} + \sum_{a : |a| = \lfloor p \rfloor} \sup_{\substack{x,y \in [0,1]^d: \\  x \neq y}} \frac{|\partial^a h(x) - \partial^a h(y)|}{|x - y|^{p - \lfloor p \rfloor}} < \infty\,,
\]
where
\[
 \|h\|_{C^{\lfloor p \rfloor}} = \|h\|_\infty + \sum_{|a| = \lfloor p \rfloor}\|\partial^a h\|_\infty \,.
\]
The space $B_{\infty,\infty}^{p}$ may equivalently be defined by the error in approximating a function using a linear B-spline basis (see \cite{devore1988interpolation} and \cite{devore1993constructive}). To do so, let $\Psi_J$ be a CDV wavelet space of regularity $S > p$ or
dyadic B-spline space of degree $r - 1 > p$ at resolution level $L_J$ that generates $J$. Let $d(h,\Psi_J) = \inf_{g \in \Psi_J} \|h-g\|_{\infty}$. We then have
\[
 h \in B_{\infty,\infty}^{p} \iff  \|h \|_\infty + \sup_{J : J \in \T}  J^{p/d} d(h,\Psi_J) < \infty \,,
\]
and, moreover, $\|h\|_{\infty} + \sup_{J : J \in \T} J^{p/d} d(h,\Psi_J)$ is equivalent to $\|h\|_{B_{\infty,\infty}^{p}}$; see, e.g., Theorem 12.3.3. of \cite{devore1993constructive} for the scalar case and Theorem 4.8 of \cite{devore1988interpolation} for the multivariate case. By Lebesgue's lemma \cite[p. 30]{devore1993constructive}, we have
\[
 d(h,\Psi_J) \leq \| h - \Pi_J h \|_{\infty} \leq (1+ \| \Pi_J \|_{\infty}) d(h,\Psi_J) \,,
\]
where $\| \Pi_J \|_{\infty} := \sup_{h : \|h\|_\infty \leq 1} \|\Pi_J h \|_\infty$ is the $L^\infty$ norm of the $L^2_X$ projection onto $\Psi_J$ (sometimes referred to as the Lebesgue constant). \cite{huang2003local} and \cite{chen2015optimal} established that $\|\Pi_J\|_\infty \lesssim 1$  under Assumption \ref{a-data}(i) when $\Psi_J$ is spanned by a (tensor product) B-spline or CDV wavelet basis, respectively. Hence,
\[
 h \in B_{\infty,\infty}^p \iff \|h \|_{\infty} + \sup_{J : J \in \T}  J^{p/d} \|h - \Pi_J h  \|_{\infty} < \infty \,,
\]
and $\|h\|_{\infty} +  \sup_{J : J \in \T}  J^{p/d} \|h - \Pi_J h \|_{\infty}$ is equivalent to $\|\cdot\|_{B_{\infty,\infty}^{p}}$.

\section{Technical Results and Proofs of Main Results}\label{ax:proofs}

In this Appendix we first introduce additional notation. We then present technical results and proofs of the main results from Sections~\ref{sec:4estimation} and~\ref{sec:4UCB}. We finally present technical results and the proofs of main results for Section~\ref{sec:4UCBde}.

\subsection{Notation}

By the discussion in Appendix \ref{ax:basis}, there are finite positive constants $a_\zeta$ and $a_b$ such that
\begin{align*}
 a_\zeta & \geq \zeta_{\psi,J}/\sqrt J \geq 1 \,, &
 a_\zeta & \geq \zeta_{b,J}/\sqrt {K(J)} \geq 1 \,, &
 a_b & \geq K(J)/J   \,.
\end{align*}
For any sequence $(Z_i)_{i=1}^n$ of random vectors and any function $g$, let $\E_n[g(Z)] = \frac{1}{n}\sum_{i=1}^n g(Z_i)$.
Estimators of the matrices defined at the beginning of Appendix \ref{ax:basis} and their orthogonalized versions are
\begin{align*}
 \widehat{G}_{\psi,J} & = \E_n \big[ \psi^J_X  (\psi^J_X)'  \big] \,, &
 \widehat{G}_{b,J} & = \E_n \big[ b^{K(J)}_W (b^{K(J)}_W)' \big] \,, \\
 \widehat{G}_{\psi,J}^o & = G_{\psi,J}^{-1/2} \E_n \big[ \psi^J_X  (\psi^J_X)'  \big] G_{\psi,J}^{-1/2} \,, &
 \widehat{G}_{b,J}^o & = G_{b,J}^{-1/2} \E_n \big[ b^{K(J)}_W (b^{K(J)}_W)' \big] G_{b,J}^{-1/2} \,, \\
 \widehat{S}_{J} & = \E_n \big[ b^{K(J)}_W (\psi^J_X)' \big] \,, &
 \widehat{S}_{J}^o & = G_{b,J}^{-1/2} \E_n \big[ b^{K(J)}_W (\psi^J_X)' \big] G_{\psi,J}^{-1/2} \,.
\end{align*}
Sieve variances and related terms are
\begin{align*}
 \|\hat{\sigma}_{x,J,J_2}  \|_{sd}^2 & \equiv n \hat \sigma^2_{J,J_2}(x) = \|\hat{\sigma}_{x,J} \|_{sd}^2  + \|\hat{\sigma}_{x,J_2} \|_{sd}^2  -2 \hat{\sigma}_{x,J,J_2} \,, &
 \|\hat{\sigma}_{x,J}\|^2_{sd} \equiv n \hat \sigma_J^2(x)  = \hat{\sigma}_{x,J,J} \,,\\
 \|{\sigma}_{x,J,J_2}  \|_{sd}^2 & = \|{\sigma}_{x,J} \|_{sd}^2  + \|{\sigma}_{x,J_2} \|_{sd}^2  -2 {\sigma}_{x,J,J_2} \,, &
 \|{\sigma}_{x,J}\|^2_{sd}  = {\sigma}_{x,J,J} \,,
\end{align*}
where
\begin{align*}
 \hat{\sigma}_{x,J,J_2} & \equiv n \tilde \sigma_{J,J_2}(x) =  \hat L_{J,x} \widehat{\Omega}_{J,J_2} (\hat L_{J_2,x})' \,, &
 \hat{L}_{J,x} & = [\psi_x^J]' [ \widehat{S}_J' \widehat{G}_{b,J}^{-1} \widehat{S}_J]^{-1} \widehat{S}_J' \widehat{G}_{b,J}^{-1} \,, \\
 {\sigma}_{x,J,J_2} & =  L_{J,x} {\Omega}_{J,J_2} (L_{J_2,x})' \,, &
 {L}_{J,x} & =  [\psi_x^J]' [ S_J'  G_{b,J}^{-1}  S_J]^{-1} S_J' G_{b,J}^{-1} \,,
\end{align*}
with $\hat \sigma^2_{J,J_2}(x)$ and $\tilde \sigma_{J,J_2}(x)$ given in (\ref{eq:hat-var-J2}), and
\begin{align*}
 \widehat{\Omega}_{J,J_2} & = \E_n \left[ \hat{u}_J \hat{u}_{J_2} b^{K(J)}_W b^{K(J_2)}_W \right]' \,, \quad \hat u_{i,J} = Y_i - \hat h_J(X_i) \,, &
 \widehat{\Omega}_J & = \widehat{\Omega}_{J,J} \,, \\
 {\Omega}_{J,J_2} & = \E \left[ {u}^2 b^{K(J)}_W b^{K(J_2)}_W \right]' \,, \quad u_{i} = Y_i - h_0(X_i) \,, &
 \Omega_J & = \Omega_{J,J}\,.
\end{align*}
Recall that $\Pi_J$ is the $L^2_X$ projection onto $\Psi_J$. We also define
\begin{align*}
 \Delta_J h_0 & = h_0 - \Pi_J h_0 \,, &
 \tilde h_J (x) & = \hat{L}_{J,x} \E_n[ b^{K(J)}_W h_0(X)] \,.
\end{align*}
For bootstrap and related processes, we use the notation
\begin{equation} \label{eq:z_star}
 \mathbb{Z}_n^*(x,J,J_2)  =  \frac{ 1}{\| \hat{\sigma}_{x,J,J_2} \|_{sd}} \left(\frac{1}{\sqrt n} \sum_{i=1}^n \bigg( \hat{L}_{J,x}  b^{K(J)}_{W_i} \hat u_{i,J} - \hat{L}_{J_2,x}  b^{K(J_2)}_{W_i} \hat u_{i,J_2} \bigg) \varpi_i  \right) \,,
\end{equation}
where $(\varpi_i)_{i=1}^n$ are IID $N(0,1)$ draws independent of the data, and
\begin{align}
 \mathbb{Z}_n^*(x,J) & \equiv \frac{D_J^*(x)}{\hat \sigma_J(x)} = \frac{1}{\|\hat{\sigma}_{x,J}\|_{sd}} \left(\frac{1}{\sqrt n} \sum_{i=1}^n \hat{L}_{J,x} b^{K(J)}_{W_i} \hat{u}_{i,J} \varpi_i \right) \,, \label{eq:Zstar} \\
 \widehat{\mathbb{Z}}_n (x,J) & = \frac{1}{\|\sigma_{x,J}\|_{sd}} \left(\frac{1}{\sqrt n} \sum_{i=1}^n {L}_{J,x} b^{K(J)}_{W_i} u_i \varpi_i \right) \,, \label{eq:Zhat} \\
 \mathbb Z_n(x,J) & = \frac{1}{\|\sigma_{x,J}\|_{sd}} \left(\frac{1}{\sqrt n} \sum_{i=1}^n {L}_{J,x} b^{K(J)}_{W_i} u_i \right) \,. \label{eq:Z}
\end{align}
The law of the processes $\mathbb{Z}_n^*(x,J)$  and $\widehat{{\mathbb Z}}_n(x,J)$ is determined from $(\varpi_i)_{i=1}^n$ conditional on the data $\mathcal{Z}^n:=(X_i,Y_i,W_i)_{i=1}^n$. We let $\mathbb P^*$ denote their probability measure (i.e., with respect to the $(\varpi_i)_{i=1}^n$ conditional on the data) and $\E^*$ denote expectation under $\mathbb P^*$.
We also shorten ``with $\mathbb{P}_{h_0}$ probability approaching $1$ (uniformly over $h_0 \in \mathcal{H})$'' to ``wpa1 $\mc H$-uniformly''. We write $\mc H^p = \mc H \cap B_{\infty,\infty}^p(M)$ and $\mc G^p = \mc G \cap B_{\infty,\infty}^p(M)$.

\subsection{Technical Results}

Here we present several technical results that are used in the proofs of the main results in Section~\ref{s:theory}. The proofs of these technical results are presented in our earlier working paper version \citep{cck}. The following Lemmas~\ref{lem:J0_mild} to~\ref{zorder} are labelled as Lemmas~D.1 to~D.7 in \cite{cck}, whereas the following Theorems~\ref{unifbiasvar} and~\ref{consistent} are labelled as Theorems~D.1 and~D.2 in \cite{cck}.

We first state two preliminary lemmas used in the proof of Theorem \ref{lepski2}. The first relates to resolution levels in the mildly ill-posed case. For any positive constant $R$, define
\begin{equation} \label{eq:J_bar_max}
 \bar J_{\max}(R) = \sup \bigg \{ J \in \T :   J \sqrt{\log J} \big[ (\log n)^4  \vee \tau_J \big] \leq R\sqrt{n}  \bigg \}\,.
\end{equation}
For $D > 0$ and $p \in [\ul p,\ol p]$, define
\begin{equation} \label{eq:J0_mild}
\begin{aligned}
 J_0(p,D) & = \sup \bigg \{ J \in \T : \tau_J \frac{\sqrt{J} \theta^*_{1-\hat \alpha}}{\sqrt{n}} \leq D J^{- \frac pd}    \bigg \} \,, \\
 J_0^+(p,D) & = \inf \{ J \in \T :  J > J_0(p,D)   \} \,.
\end{aligned}
\end{equation}

\begin{lemma}\label{lem:J0_mild}
Let Assumptions \ref{a-data}-\ref{a-var} hold and let $\tau_J \asymp J^{\varsigma/d}$ with $\varsigma \geq 0$. Then: with $\bar J_{\max}(R)$ as defined in (\ref{eq:J_bar_max}) for any $R >0$ and $J_0^+(p,D)$ as defined in (\ref{eq:J0_mild}) for any $D >0$, we have
\[
 \inf_{p \in [\ul p,\ol p]} \inf_{h_0 \in \mc H^p} \mathbb{P}_{h_0} ( J_0^+(p,D) < \bar{J}_{\max}(R) ) \to 1.
\]
\end{lemma}

The second lemma relates to resolution levels in the severely ill-posed case. For $R > 0$ and $p \in [\ul p,\ol p]$, define
\begin{align}
 \bar J_{\max}^*(R) & = \sup \left \{ J \in \T :  \tau_J J \sqrt{\log J}  \leq R\sqrt{n}  \right \} \,, \label{eq:Jmax*}\\
  M_0(p,R) & = \sup \{ J \in \T :  \tau_J J^{\frac pd + \frac{1}{2}} \sqrt{\log J} \leq R \sqrt{n}  \} \,, \label{eq:M0} \\
  M_0^+(p,R) & = \inf \{ J \in \T :  J > M_0(p,R) \} \,. \notag
\end{align}
Note that $M_0(p,R)$ is (weakly) decreasing in $p$. In particular, as $\ol p/d+1/2 \geq \ul p/d + 1/2 > 1$, we have $\bar{J}_{\max}^*(R) \geq M_0(\ul p,R) \geq M_0(p,R) \geq M_0(\ol p,R)$ for each $R$ and each $p \in [\ul p,\ol p]$.

\begin{lemma}\label{lem:J0_severe}
Let $\tau_J \asymp \exp(C J^{\varsigma/d})$ for some $C,\varsigma > 0$. Then for any $R > 0$, the inequality  $M_0^+(\ol p,R) \geq J_{\max}^*(R)$ holds for all $n$ sufficiently large.
\end{lemma}

\subsubsection{Uniform-in-$J$ Convergence Rates for $\hat h_J$}

Recall the definition of $\bar J_{\max}(R)$ from (\ref{eq:J_bar_max}) and that $\Delta_J h_0 = h_0 - \Pi_J h_0$.

\begin{theorem}\label{unifbiasvar}
Let Assumptions \ref{a-data}, \ref{a-residuals}(i), and \ref{a-approx} hold, and for any positive constant $R$ let $\bar J_{\max} \equiv \bar{J}_{\max} (R)$. Then: there exists a universal constant $C_{\ref{unifbiasvar}} > 0$ such that
\begin{align*}
 & (i) \; \; \;  \inf_{h_0 \in \mathcal{H}} \mathbb{P}_{h_0} \bigg(   \| \tilde{h}_J -  h_0  \|_{\infty} \leq C_{\ref{unifbiasvar}} \| \Delta_J h_0 \|_{\infty} \; \;  \; \forall \; \; J \in \T \cap [1, \bar J_{\max}] \bigg) \rightarrow 1 \, , \\ & (ii) \; \; \;
 \inf_{h_0 \in \mathcal{H}} \mathbb P_{h_0} \bigg(   \|\hat h_J - \tilde h_J\|_\infty \leq C_{\ref{unifbiasvar}} \tau_J \frac{\sqrt{J \log \bar{J}_{\max}}}{\sqrt{n}}  \; \; \; \forall \; \; J \in \T \cap [1, \bar J_{\max}] \bigg) \rightarrow 1 \,.
\end{align*}
\end{theorem}

\subsubsection{Uniform-in-$J$ Estimation of Sieve Variance Terms}

Recall the definition of $\bar J_{\max}(R)$ from (\ref{eq:J_bar_max}). In the remainder of this subsection, for any fixed $R > 0$, let $\bar J_{\max} \equiv \bar J_{\max}(R)$. Also let $J_{\min} \to \infty$ arbitrarily slowly. Given $\bar J_{\max}$ and $J_{\min}$, define
$\mc J_n = \{J \in \T : J_{\min} \leq J \leq \bar J_{\max}\}$,
\begin{equation} \label{eq:S-set}
\mc S_n = \{ (x,J,J_2) \in \mc X \times \mc J_n \times \mc J_n : J_2 > J\}
\end{equation}
and
\begin{equation} \label{eq:delta_n}
 \delta_n = \tau_{\bar{J}_{\max}} \sqrt{\frac{\bar{J}_{\max} \log \bar{J}_{\max}}{n}}  + \bigg(  \frac{\bar{J}_{\max}^2 \log \bar{J}_{\max}}{n}   \bigg)^{1/3}     + J_{\min}^{- \ul p/d} \,.
\end{equation}

\begin{lemma} \label{varest3}
Let Assumptions \ref{a-data}-\ref{a-var} hold. Then: there exists universal constants $C_{\ref{varest3}} > 0$ and $N_{\ref{varest3}} \in \mb N$ such that:
\begin{enumerate}[nosep]
\item[(i)] for every $x \in \mathcal{X}$ and $J ,J_2  \in \T$ with $J_2 > J \geq N_{\ref{varest3}}$, we have
\[
 C_{\ref{varest3}}^{-1} \| \sigma_{x,J_2}   \|_{sd} \leq  \| \sigma_{x,J,J_2}   \|_{sd} \leq C_{\ref{varest3}} \| \sigma_{x,J_2}   \|_{sd}\,;
\]
\item[(ii)] we have
\[
 \inf_{h_0 \in \mc H} \mathbb{P}_{h_0} \bigg( \sup_{(x,J,J_2) \in \mathcal{S}_n  } \left| \frac{\| \hat{\sigma}_{x,J,J_2}  \|_{sd}}{ \| \sigma_{x,J,J_2}   \|_{sd}} - 1   \right| \leq C_{\ref{varest3}}  \delta_n     \bigg) \rightarrow 1 \,.
\]
\end{enumerate}
\end{lemma}

\begin{lemma}\label{lem-varest2}
Let Assumptions \ref{a-data}-\ref{a-approx} hold.  Then: there is a universal constant $C_{\ref{lem-varest2}} > 0$ such that
\[
 \inf_{h_0 \in \mc H}  \mathbb{P}_{h_0} \bigg( \sup_{(x,J,J_2) \in \mathcal{S}_n} \frac{ \left|  \hat{\sigma}_{x,J,J_2} - \sigma_{x,J,J_2}    \right|}{\| \sigma_{x,J}  \|_{sd}  \| \sigma_{x,J_2}  \|_{sd} }  \leq C_{\ref{lem-varest2}} \delta_n \bigg)  \rightarrow 1 \,.
\]
In particular,
\[
 \inf_{h_0 \in \mc H}  \mathbb{P}_{h_0} \bigg( \sup_{(x,J) \in \mathcal{X} \times \mathcal{J}_n}  \left|  \frac{\|\hat{\sigma}_{x,J}\|^2_{sd}}{\|\sigma_{x,J} \|^2_{sd}} - 1   \right|   \leq C_{\ref{lem-varest2}} \delta_n \bigg)  \rightarrow 1 \,.
\]
\end{lemma}

\subsubsection{Uniform Consistency of $\hat J_{\max}$}

For the following lemma, recall $\hat J_{\max}$ from (\ref{eq:J_hat_max}) and $\bar J_{\max}(R)$ from (\ref{eq:J_bar_max}).

\begin{lemma}\label{lem:gridtest}
Let Assumptions \ref{a-data}-\ref{a-approx} hold. Then: replacing $10 \sqrt n$ with $M \sqrt n$ for any $M > 0$ in the definition of $\hat J_{\max}$ from (\ref{eq:J_hat_max}), there exists $R_1,R_2 > 0 $ which satisfy
\[
 \inf_{h_0 \in \mc H} \mathbb{P}_{h_0} \bigg( \bar{J}_{\max}(R_1) \leq  \hat{J}_{\max} \leq \bar{J}_{\max}(R_2) \bigg)  \to 1\,.
\]
\end{lemma}

\begin{remark} \label{rmk:J_max}
For any $R_2 \geq R_1 > 0$ there exists a finite positive constant $C$ for which
\[
 \bar{J}_{\max}(R_1) \leq \bar{J}_{\max}(R_2) \leq C \bar{J}_{\max}(R_1) \,.
\]
Lemma \ref{lem:gridtest} therefore provides an asymptotic rate of divergence for $\hat J_{\max}$.
\end{remark}

\subsubsection{Uniform-in-$J$ Bounds for the Bootstrap}

For the following Lemma, recall the critical value $\theta^*_{1-\hat \alpha}$ from Section \ref{sec:Jtilde}. 

\begin{lemma} \label{lepquant}
Let Assumptions \ref{a-data}-\ref{a-var} hold. Then: with $\bar J_{\max}(R)$ as defined in (\ref{eq:J_bar_max}) for any $R >0$, there exists constants $C_4,C_5 > 0$ for which
\[
 \inf_{h_0 \in \mathcal{H}} \mathbb{P}_{h_0} \bigg( C_4 \sqrt{\log \bar{J}_{\max}(R)} \leq  \theta^*_{1-\hat \alpha} \leq C_5 \sqrt{\log \bar{J}_{\max}(R)}  \bigg) \rightarrow 1.
\]
\end{lemma}

The second is a companion result concerning the critical value $z^*_{1-\alpha}$ from Section~\ref{sec:ucbs}:

\begin{lemma} \label{zorder}
Let Assumptions \ref{a-data}-\ref{a-var} hold. 
Then: with $\bar J_{\max}(R)$ as defined in (\ref{eq:J_bar_max}) for any $R >0$, there exists a constant $C_{\ref{zorder}} > 0$ for which
\[
 \inf_{h_0 \in \mathcal{H}} \mathbb{P}_{h_0} \bigg( z_{1- \alpha}^* \leq C_{\ref{zorder}} \sqrt{\log \bar{J}_{\max}(R)}  \bigg) \rightarrow 1\,.
\]
\end{lemma}

\subsubsection{Uniform Consistency for the Bootstrap}

Recall $\bar J_{\max} \equiv \bar J_{\max}(R)$ from (\ref{eq:J_bar_max}) and $\mc J_n$ and $\mc S_n$ from (\ref{eq:S-set}).

\begin{theorem} \label{consistent}
Let Assumptions \ref{a-data}-\ref{a-var} hold and let $J_{\min} \asymp (\log \bar J_{\max})^2$. Then: there exists a sequence $\gamma_n \downarrow 0$ for which the following inequalities hold wpa1 $\mc H$-uniformly:
\begin{align*}
 (i) \quad & \sup_{s \in \R} \left| \mathbb{P}_{h_0} \bigg( \sup_{(x,J) \in \mathcal{X} \times \mathcal{J}_n } \left | \sqrt{n} \frac{\hat{h}_J(x) - \tilde{h}_J(x)}{\| \hat \sigma_{x,J} \|_{sd}} \right| \leq s \right)    - \mathbb{P}^* \bigg( \sup_{(x,J) \in \mathcal{X} \times \mathcal{J}_n} \left| \mathbb{Z}_n^*(x,J) \right|  \leq s  \bigg)  \bigg |  \leq \gamma_n \,, \\
 (ii) \quad & \sup_{s \in \R} \left| \mathbb{P}_{h_0} \bigg( \sup_{(x,J,J_2) \in \mathcal{S}_n} \bigg | \sqrt{n} \frac{\hat{h}_J(x) - \hat{h}_{J_2}(x) -( \tilde{h}_J(x) - \tilde{h}_{J_2}(x) )  }{\| \hat{\sigma}_{x,J,J_2} \|_{sd}} \right| \leq s \bigg) \\
 & \; \; \; \; \; \; \;  \; \; \; \; - \mathbb{P}^* \bigg( \sup_{(x,J,J_2) \in \mathcal{S}_n} \left| \mathbb{Z}_n^*(x,J,J_2) \right| \leq s \bigg) \bigg | \leq \gamma_n \,.
\end{align*}
\end{theorem}

\subsection{Proofs of Main Results in Sections~\ref{sec:4estimation} and~\ref{sec:4UCB}}

\begin{proof}[Proof of Theorem~\ref{lepski2}]
We first list some constants that will be used throughout the proof. Fix $R_2 > 0$ in the definition of $\bar J_{\max}(R_2)$ from (\ref{eq:J_bar_max}) sufficiently large so that by Lemma~\ref{lem:gridtest} we have $\inf_{h_0 \in \mc H} \mathbb{P}_{h_0} ( \hat{J}_{\max} \leq \bar{J}_{\max}(R_2)) \to 1$. Let $\bar J_{\max} \equiv \bar J_{\max}(R_2)$ for the remainder of the proof. By Theorem~\ref{unifbiasvar}(i), there exists $C_{\ref{unifbiasvar}} > 0 $ which satisfies
\begin{equation} \label{eq:lepski2_rate}
 \inf_{h_0 \in \mathcal{H}} \mathbb{P}_{h_0} \bigg( \| \tilde{h}_J - \Pi_J h_0  \|_{\infty} \leq C_{\ref{unifbiasvar}} \| \Pi_J h_0 - h_0  \|_{\infty} \; \; \; \;  \forall \; J \in [1, \bar{J}_{\max}] \cap \T \bigg) \rightarrow 1 .
\end{equation}
For our choice of sieves, there exists $B_2 > 0 $ which satisfies
\begin{equation} \label{eq:lepski2_bias}
 \sup_{p \in [\ul p, \ol p]} \sup_{h_0 \in \mc H^p} J^{\frac pd} \| \Pi_J h_0 - h_0  \|_{\infty} \leq B_2  \; \; \; \; \; \forall \; \; J \in \T\,.
\end{equation}
Let $\hat{\mc S} = \{(x,J,J_2) \in \mc X \times \hat{\mc J} \times \hat{\mc J} : J_2 > J\}$. Lemmas~\ref{varest3} and~\ref{lem:gridtest}, Assumption~\ref{a-var}(i), and the fact that $\delta_n \downarrow 0$ (cf. (\ref{eq:delta_n})) imply that there exists $C_2, C_3 > 0 $ which satisfy
\begin{equation} \label{eq:lepski2_sd}
 \inf_{h_0 \in \mathcal{H}}  \mathbb{P}_{h_0}  \bigg(     \sup_{(x,J,J_2) \in \hat{\mathcal{S}}}  \frac{ \tau_{J_2} \sqrt{J}_2   }{\| \hat{\sigma}_{x,J,J_2}  \|_{sd}} \leq C_3     \bigg) \rightarrow 1 \, , \;
 \inf_{h_0 \in \mathcal{H}} \mathbb{P}_{h_0} \bigg( \sup_{(x,J,J_2) \in \hat{\mathcal{S}}}  \frac{\| \hat{\sigma}_{x,J,J_2}  \|_{sd}}{\tau_{J_2} \sqrt{J_2}} \leq C_2        \bigg) \rightarrow 1\,.
\end{equation}
Additionally, by Lemma~\ref{lepquant} there exists constants $C_4, C_5 >  0 $ which satisfy
\begin{equation} \label{eq:lepski2_jmax}
 \inf_{h_0 \in \mathcal{H}} \mathbb{P}_{h_0} \bigg( C_4 \sqrt{\log \bar{J}_{\max}} \leq  \theta^*_{1-\hat \alpha} \leq C_5 \sqrt{\log \bar{J}_{\max}}     \bigg) \to 1.
\end{equation}

\underline{Part (i), step 1:} We verify that $\hat{J}$ achieves the optimal rate under mild ill-posedness. Note by the procedure in Appendix~\ref{sec:regression} this is sufficient for adaptivity of $\tilde J$ for nonparametric regression. Fix $\xi > 1$ ($\xi = 1.1$ in the main text). Choose $D > 0$ such that $2B_2(C_1+1)D^{-1} C_3 < (\xi -1)$. Recall $J_0(p,D)$ and $J_0^+(p,D)$ from (\ref{eq:J0_mild}); we drop dependence of these quantities on $(p,D)$ hereafter to simplify notation. By Lemma~\ref{lem:J0_mild}, $\inf_{p \in [\ul p,\ol p]} \inf_{h_0 \in \mc H^p} \mathbb{P}_{h_0} ( J_0^+ < \bar{J}_{\max} ) \to 1$. It then follows from Lemmas~\ref{lem:J0_mild} and~\ref{lem:gridtest} that $\inf_{p \in [\ul p,\ol p]} \inf_{h_0 \in \mc H^p} \mathbb{P}_{h_0}  ( J_0^+ < \hat{J}_{\max} ) \to 1$. We therefore assume for the remainder of the proof of part (i) that $J_0^+ < \hat{J}_{\max},\bar J_{\max}$.

By Lemma~\ref{lem:gridtest}, $\hat{\mc J} \subseteq \mc J_n := \{J \in \T : 0.1 (\log \bar J_{\max})^2 \leq J \leq \bar J_{\max}\}$ wpa1 $\mc H$-uniformly. Then for all $J \in \hat{\mc J}$ with $J > J_0^+$, by the triangle inequality, displays (\ref{eq:lepski2_rate}) and (\ref{eq:lepski2_bias}), and definition of $J_0$, we may deduce that
\begin{align*}
 & \left| \| \hat{h}_J - \hat{h}_{J_0^+}  \|_{\infty} - \| \hat{h}_{J} - \hat{h}_{J_0^+} - (  \tilde{h}_J - \tilde{h}_{J_0^+} ) \|_{\infty} \right| \\
 & \leq  \| \tilde{h}_J - \Pi_J h_0  \|_{\infty} + \| \tilde{h}_{J_0^+} - \Pi_{J_0^+} h_0 \|_{\infty} + \| \Pi_{J_0^+} h_0  - h_0 \|_{\infty} + \| \Pi_{J} h_0 - h_0   \|_{\infty} \\
 & \leq 2 B_2 (1 + C_1)  (J_0^+)^{-p/d} \\
 & \leq 2 B_2 (1 + C_1) D^{-1}  \theta^*_{1-\hat \alpha} \tau_{J_0^+} \sqrt{J_0^+ / n}
\end{align*}
wpa1 uniformly for $h_0 \in \mc H^p$ and $p \in [\ul p,\ol p]$. By (\ref{eq:lepski2_sd}), we have that for all $J \in \hat{\mc J}$ with $J > J_0^+$
\[
 \tau_{J_0^+} \sqrt{J_0^+} \leq \tau_{J} \sqrt{J} \leq C_3 \| \hat{\sigma}_{x,J_0^+,J}    \|_{sd} \quad \quad \forall \quad x \in \mathcal{X}
\]
wpa1 uniformly for $h_0 \in \mc H^p$ and $p \in [\ul p,\ol p]$. Combining the preceding two inequalities and using the definition of $D$, we obtain that for all $J \in \hat{\mc J}$ with $J > J_0^+$,
\[
 \sup_{x \in \mathcal{X}} \sqrt{n} \frac{| \hat{h}_J(x) - \hat{h}_{J_0^+}(x) |}{\| \hat{\sigma}_{x,J_0^+,J} \|_{sd}}
 \leq \sup_{x \in \mathcal{X}}  \sqrt{n} \frac{| \hat{h}_{J}(x) - \hat{h}_{J_0^+}(x) - ( \tilde{h}_J(x) - \tilde{h}_{J_0^+}(x) ) |}{\| \hat{\sigma}_{x,J_0^+,J} \|_{sd}}  + (\xi-1) \theta^*_{1-\hat \alpha}
\]
wpa1 uniformly for $h_0 \in \mc H^p$ and $p \in [\ul p,\ol p]$.
It follows by definition of $\hat J$ that
\begin{align}
 & \sup_{p \in [\ul p,\ol p]} \sup_{h_0 \in \mc H^p} \mathbb{P}_{h_0} \big( \hat{J} >  J_0^+ \big) \notag \\
 & \leq \sup_{p \in [\ul p,\ol p]}  \sup_{h_0 \in \mc H^p} \mathbb{P}_{h_0} \bigg(  \sup_{J \in \hat{\mc J} : J >  J_0^+} \sup_{x \in \mathcal{X}} \frac{ \sqrt{n} | \hat{h}_{J_0^+} (x) - \hat{h}_{J} (x) |}{\| \hat{\sigma}_{x,J_0^+,J}  \|_{sd} }  > \xi \theta^*_{1-\hat \alpha}  \bigg) \notag \\
 & \leq  \sup_{h_0 \in \mathcal{H} } \mathbb{P}_{h_0} \bigg( \sup_{(x,J,J_2) \in \hat{\mc S}} \frac{\sqrt{n} | \hat{h}_{J}(x) - \hat{h}_{J_2}(x) - ( \tilde{h}_{J} (x) - \tilde{h}_{J_2} (x) )  |}{\| \hat{\sigma}_{x,J,J_2}  \|_{sd}}  \  > \theta^*_{1-\hat \alpha})     \bigg) + o(1)  \,.  \label{eq:lepski2_mild_prob}
\end{align}

To control the r.h.s. probability in (\ref{eq:lepski2_mild_prob}), let $\hat{\mc J}(\tilde J) = \{J \in \mc T : 0.1 (\log \tilde J)^2 \leq J \leq \tilde J\}$, $\hat{\mc S}(\tilde J) = \{(x,J,J_2) \in \mc X \times \hat{\mc J}(\tilde J) \times \hat{\mc J}(\tilde J) : J_2 > J\}$, and $\theta^*_{1-\hat \alpha; \tilde J}$ denote the $(1- 0.5 \wedge \sqrt{(\log \tilde J)/\tilde J})$ quantile of $\sup_{(x,J,J_2) \in \hat{\mathcal{S}}(\tilde J)} \left| \mathbb{Z}_n^*(x,J,J_2) \right|$. Then by Lemma~\ref{lem:gridtest} and Theorem~\ref{consistent}(ii), we have
\begin{align}
 & \sup_{h_0 \in \mathcal{H} } \mathbb{P}_{h_0} \bigg( \sup_{(x,J,J_2) \in \hat{\mc S}} \frac{\sqrt{n} | \hat{h}_{J}(x) - \hat{h}_{J_2}(x) - ( \tilde{h}_{J} (x) - \tilde{h}_{J_2} (x) )  |}{\| \hat{\sigma}_{x,J,J_2}  \|_{sd}}  \  > \theta^*_{1-\hat \alpha}  \bigg) \notag \\
 & \leq \sup_{h_0 \in \mathcal{H} } \sum_{\tilde J \in \mc T : \tilde J = \bar J_{\max}(R_1)}^{ \bar J_{\max}(R_2)}  \mathbb{P}_{h_0} \bigg( \sup_{(x,J,J_2) \in \hat{\mc S}(\tilde J)} \frac{\sqrt{n} | \hat{h}_{J}(x) - \hat{h}_{J_2}(x) - ( \tilde{h}_{J} (x) - \tilde{h}_{J_2} (x) )  |}{\| \hat{\sigma}_{x,J,J_2}  \|_{sd}}  \  > \theta^*_{1-\hat \alpha;\tilde J}  \bigg) \notag \\
 & \leq \sum_{\tilde J \in \mc T : \tilde J = \bar J_{\max}(R_1)}^{ \bar J_{\max}(R_2)}  \left( \sqrt{(\log \tilde J)/\tilde J} + \gamma_n  + o(1) \right) \to 0 \,, \label{eq:lepski2_mild_prob_2}
\end{align}
where the final line holds for all $n$ large, because $\bar J_{\max}(R_1) \to \infty$, $\gamma_n \downarrow 0$, and, by our choice of sieve and Remark~\ref{rmk:J_max}, for some constant $C > 0$ we have
\begin{align*}
 \#\{J \in \mc T : \bar J_{\max}(R_1) \leq J \leq \bar J_{\max}(R_2)\}
 & \leq \#\{J \in \mc T : \bar J_{\max}(R_1) \leq J \leq C \bar J_{\max}(R_1)\} \\
 & \leq \#\{l \in \mb N : \bar J_{\max}(R_1) \leq 2^{ld} \leq C \bar J_{\max}(R_1)\} \leq C \,.
\end{align*}
In view of (\ref{eq:lepski2_mild_prob}), this proves $\hat{J} \leq J_0^+ $ wpa1 uniformly for $h_0 \in \mc H^p$ and $p \in [\ul p,\ol p]$.

Whenever $\hat{J} \leq J_0^+ < \hat J_{\max},\bar J_{\max}$, it follows by definition of $\hat J$ and display (\ref{eq:lepski2_sd}) that wpa1 uniformly for $h_0 \in \mc H^p$ and $p \in [\ul p,\ol p]$, we have
\begin{align*}
 \| \hat{h}_{\hat{J}} -  h_0 \|_{\infty}
 & \leq \| \hat{h}_{\hat{J}} - \hat{h}_{J_0^+} \|_{\infty} + \| \hat{h}_{J_0^+} - h_0 \|_{\infty} \\
 & \leq  C_2 \xi  \theta^*_{1-\hat \alpha} \tau_{J_0^+} \sqrt{J_0^+ / n} + \|\hat{h}_{J_0^+} - \tilde{h}_{J_0^+} \|_{\infty} + \| \tilde{h}_{J_0^+} - h_0 \|_{\infty}.
\end{align*}
Then by Theorem~\ref{unifbiasvar}, definition of $J_0^+$, and the lower bound on $\theta^*_{1-\hat \alpha}$ in display (\ref{eq:lepski2_jmax}), we may deduce that there exists a constant $C > 0$ for which
\[
 \inf_{p \in [\ul p,\ol p]}   \inf_{h_0 \in \mathcal{H}^p}  \mathbb{P}_{h_0} \bigg( \| \hat{h}_{\hat{J}} - h_0 \|_{\infty}  \leq  C \theta^*_{1-\hat \alpha} \tau_{J_0^+} \sqrt{J_0^+ / n} \bigg) \to 1.
\]
As the model is mildly ill-posed, there exists a constant $C' > 0 $ for which $ \tau_{J_0^+} \sqrt{J_0^+} \leq C' \tau_{J_0} \sqrt{J_0}$. It then follows by definition of $J_0$  that
\begin{equation} \label{eq:lepski2_jhat_mild}
 \inf_{p \in [\ul p,\ol p]}  \inf_{h_0 \in \mc H^p} \mathbb{P}_{h_0} \left(  \| \hat{h}_{\hat{J}} -  h_0   \|_{\infty}  \leq  CC' D J_0^{- p/d}  \right) \rightarrow 1.
\end{equation}
By the upper bound on $\theta^*_{1-\hat \alpha}$ in display (\ref{eq:lepski2_jmax}) and because $\sqrt{\log \bar{J}_{\max}} \asymp \sqrt{\log n}$ (as the model is mildly ill-posed), there exists a constant $E > 0 $ such that by defining
\[
 J_n^*(p,E) = \sup \left\{   J \in \T : \tau_J \sqrt{(J \log n) / n} \leq E J^{-p/d} \right\}
\]
we have $\inf_{p \in [\ul p,\ol p]}   \inf_{h_0 \in \mc H^p}  \big( J_n^*(p,E) \leq J_0(p,D)  \big) \to 1$. Hence, as $\tau_J \asymp J^{\varsigma/d}$ we have $J_n^*(p,E) \asymp (n/\log n)^{d/(2(p + \varsigma) + d)}$. The desired result now follows from (\ref{eq:lepski2_jhat_mild}).

\underline{Part (i), step 2:} We verify that $\tilde J$ achieves the optimal rate under mild ill-posedness. By step 1, we have $ \inf_{p \in [\ul p,\ol p]} \inf_{h_0 \in \mc H^p} \mathbb{P}_{h_0} ( \hat{J} \leq J_0^+ ) \to 1$. If we can show that $ \hat J_n > J_0^+  $ wpa1 $\mc H$-uniformly, then $\tilde{J} = \hat{J}$ wpa1 $\mc H$-uniformly and the result follows by step 1.

By the lower bound on $\theta^*_{1-\hat \alpha}$ in display (\ref{eq:lepski2_jmax}) and the fact that $\sqrt{\log \bar{J}_{\max}} \asymp \sqrt{\log n}$ (as the model is mildly ill-posed), we may deduce that there exists a constant $E' > 0$ such that $\inf_{p \in [\ul p,\ol p]}   \inf_{h_0 \in \mc H^p}  \big( J_n^\dagger(p,E') \geq J_0^+(p,D)  \big) \to 1$
where
\[
 J_n^\dagger(p,E') = \inf \left\{   J \in \T : \tau_J \sqrt{(J \log n) / n} > E' J^{-p/d} \right\}.
\]
But note that $\max_{p \in [\ul p,\ol p]} J_0^\dagger(p,E') = J_0^\dagger(\ul p,E')$. The result now follows by Lemma~\ref{lem:gridtest}, noting that $\bar{J}_{\max}(R_1) /J_0^\dagger(\ul p,E') \to \infty$ when the model is mildly ill-posed because $\ul p > d/2$.

\underline{Part (ii), step 1:} We verify that $\hat J_n$ achieves the optimal rate under severe ill-posedness. To simplify notation we assume a CDV wavelet basis, though a similar argument applies (albeit with more complicated notation) for B-splines. Note that when the model is severely ill-posed, for any $R > 0$ we have $n^{\beta} \lesssim \tau_{\bar{J}_{\max}(R)}$ for some $\beta > 0$ and so $\tau_{\bar{J}_{\max}(R)} > (\log n)^4 $ for all sufficiently large $n$. Therefore $\bar J_{\max}(R) = \bar J_{\max}^*(R)$ for all $n$ sufficiently large, where $\bar J_{\max}^*(R)$ is defined in (\ref{eq:Jmax*}). By Theorem~\ref{unifbiasvar}, Lemma~\ref{lem:gridtest}, and Remark~\ref{rmk:J_max}, we may deduce that there exist constants $D,D' > 0 $ for which
\begin{align*}
 \| \hat{h}_{\hat J_n} - h_0 \|_{\infty} & \leq \| \hat{h}_{\hat J_n} - \tilde{h}_{\hat J_n}    \|_{\infty} + \| \tilde{h}_{\hat J_n} - h_0 \|_{\infty} \\
 & \leq D \bigg( (2^{-d}\bar J_{\max}^*(R_1))^{- \frac pd}  +  \tau_{2^{-d}\bar J_{\max}^*(R_2)} \sqrt{2^{-d}\bar J_{\max}^*(R_2) \log(2^{-d}\bar J_{\max}^*(R_2)) / n}    \bigg) \\
 & \leq D' \bigg( (2^{-d}\bar J_{\max}^*(R_2))^{- \frac pd}  +  \tau_{2^{-d}\bar J_{\max}^*(R_2)}  \sqrt{2^{-d}\bar J_{\max}^*(R_2) \log(2^{-d}\bar J_{\max}^*(R_2)) / n}    \bigg)
\end{align*}
wpa1 uniformly over $ \mc H^p$ and $p \in [\ul p,\ol p]$.

Recall the definition of $M_0(p,R_2)$ from (\ref{eq:M0}). By Lemma~\ref{lem:J0_severe}, for all $p \in [\ul p,\ol p]$ we have that $M_0(p,R_2) \geq M_0(\ol p,R_2) \geq 2^{-d} J_{\max}^*(R_2)$ holds for all $n$ sufficiently large, in which case by definition of $M_0(p,R_2)$ we must have
\[
  \tau_{2^{-d}\bar J_{\max}^*(R_2)}  \sqrt{2^{-d}\bar J_{\max}^*(R_2) \log(2^{-d}\bar J_{\max}^*(R_2)) / n}  \leq R_2 (2^{-d}\bar J_{\max}^*(R_2))^{- \frac pd} \,.
\]
Combining the preceding two inequalities then yields
\[
 \| \hat{h}_{\hat J_n} - h_0 \|_{\infty} \leq D'(1 + R_2) 2^p (\bar J_{\max}^*(R_2))^{- \frac pd}
\]
wpa1 uniformly over $ \mc H^p$ and $p \in [\ul p,\ol p]$.

It remains to show $ (\log n)^{d/\varsigma} \lesssim \bar J_{\max}^*(R_2)$ when $\tau_J \asymp \exp(C J^{\varsigma/d})$ for $C,\varsigma > 0$. Suppose $\liminf_{n \to \infty} \bar J_{\max}^*(R_2) / (\log n)^{d/\varsigma} = 0$. Then along a subsequence $\{n_k\}_{k\geq1}$ we have $ \bar{J}_{\max}^*(R_2) =  (2^{-\varsigma} C^{-1} u_{n_k} \log n_k)^{d/\varsigma} $ for some sequence $u_{n_k}  \downarrow 0$. Then $2^d \bar J_{\max}^*(R_2) \in \mc T$ satisfies
\[
 \tau_{2^d \bar J_{\max}^*(R_2)} 2^d \bar J_{\max}^*(R_2) \sqrt{\log (2^d \bar J_{\max}^*(R_2)) / n_k} \lesssim n_k^{u_{n_k} - \frac{1}{2}} (\log n_k)^{d/ \varsigma}  \sqrt{\log \log n_k} \xrightarrow[k \rightarrow \infty]{} 0\,,
\]
thereby contradicting the definition of $\bar J_{\max}^*(R_2)$ from (\ref{eq:Jmax*}) for all sufficiently large $k$.

\underline{Part (ii), step 2:} We verify that $\tilde{J}$ achieves the optimal rate under severe ill-posedness. For any constant $D > 0$, by definition of $\tilde J$ we have
\begin{align*}
 & \sup_{p \in [\ul p,\ol p]} \sup_{h_0 \in \mc H^p} \mathbb{P}_{h_0} \big( \| \hat {h}_{\tilde J} - h_0 \|_{\infty} > D (\log n)^{-p / \varsigma} \big)  \\
 & \leq \sup_{p \in [\ul p,\ol p]} \sup_{h_0 \in \mc H^p} \mathbb{P}_{h_0} \big( \| \hat{h}_{\hat J} - h_0 \|_{\infty} > D (\log n)^{-p / \varsigma} \mbox{ and } \hat{J} < \hat J_n \big) \\
 & \quad \quad + \sup_{p \in [\ul p,\ol p]} \sup_{h_0 \in \mc H^p} \mathbb{P}_{h_0} \big(  \| \hat{h}_{\hat J_n} - h_0    \|_{\infty} > D (\log n)^{-p / \varsigma}  \big)  .
\end{align*}
By part (ii), step 1, the constant $D$ can be chosen sufficiently large so that the second term on the r.h.s. is $o(1)$. For the first term, note that $\|\hat{h}_{\hat {J}} - h_0  \|_{\infty} \leq \| \hat{h}_{\hat{J}} -  \hat{h}_{\hat J_n}  \|_{\infty} + \| \hat{h}_{\hat J_n} - h_0  \|_{\infty}$, so  it suffices to show that there exists a constant $D > 0 $ for which
\[
 \sup_{p \in [\ul p,\ol p]} \sup_{h_0 \in \mc H^p} \mathbb{P}_{h_0} \big(  \| \hat{h}_{\hat{J}} - \hat{h}_{\hat J_n} \|_{\infty} > D (\log n)^{-p / \varsigma} \mbox{ and } \hat{J} < \hat J_n   \big) \rightarrow 0\,.
\]
But by definition of $\hat J$ and displays (\ref{eq:lepski2_sd}) and (\ref{eq:lepski2_jmax}), we have
\begin{align*}
 & \sup_{p \in [\ul p,\ol p]} \sup_{h_0 \in \mc H^p} \mathbb{P}_{h_0} \big(  \| \hat{h}_{\hat{J}} - \hat{h}_{\hat J_n} \|_{\infty} > D (\log n)^{-p / \varsigma} \mbox{ and } \hat{J} < \hat J_n   \big) \\
 & \leq   \sup_{p \in [\ul p,\ol p]} \sup_{h_0 \in \mc H^p} \mathbb{P}_{h_0} \left(  \xi C_2 \theta^*_{1-\hat \alpha} \tau_{\hat J_n} \sqrt{\hat J_n / n} > D (\log n)^{-p / \varsigma} \right) + o(1) \\
 & \leq   \sup_{p \in [\ul p,\ol p]}  \mathbbm{1} \left[ \xi C_2 C_5 \tau_{2^{-d}\bar J_{\max}^*(R_2)}  \sqrt{2^{-d}\bar J_{\max}^*(R_2) \log(2^{-d} \bar J_{\max}^*(R_2)) / n} > D (\log n)^{-p / \varsigma} \right] + o(1) \,.
\end{align*}
By step 1, we have $\tau_{2^{-d}\bar J_{\max}^*(R_2)}  \sqrt{2^{-d}\bar J_{\max}^*(R_2) \log(2^{-d} \bar J_{\max}^*(R_2)) / n} \lesssim (\log n)^{-p/\varsigma}$ uniformly for $p \in[\ul p,\ol p]$, so the constant $D$ can be chosen sufficiently large that the indicator function on the r.h.s. is zero uniformly for $p \in [\ul p,\ol p]$ for all $n$ sufficiently large.
\end{proof}

\begin{proof}[Proof of Corollary~\ref{cor:lepski2}]
\underline{Part (i):} Recall $J_0^+ \equiv J_0(p,D)^+$ from (\ref{eq:J0_mild}). We have
\[
 \| \partial^a \hat{h}_{\hat{J}} - \partial^a  h_0 \|_{\infty}
 \leq \| \partial^a  \hat{h}_{\hat{J}} - \partial^a  \hat{h}_{J_0^+} \|_{\infty} +  \| \partial^a \hat {h}_{J_0^+} - \partial^a \tilde h_{J_0^+} \|_{\infty}  + \| \partial^a \tilde {h}_{J_0^+} - \partial^a  h_0 \|_{\infty} \,.
\]
As $\hat{J} \leq J_0^+ < \hat J_{\max},\bar J_{\max}$ holds wpa1 uniformly for $h_0 \in \mc H^p$ and $p \in [\ul p,\ol p]$, by part (i), step 1 of the proof of Theorem~\ref{lepski2}, we may appeal to a Bernstein inequality (or inverse estimate) for our choice of basis to write
\[
 \| \partial^a \hat{h}_{\hat{J}} - \partial^a  h_0 \|_{\infty}
 \lesssim (J_0^+)^{|a|/d} \left( \| \hat{h}_{\hat{J}} - \hat{h}_{J_0^+} \|_{\infty} +  \| \hat {h}_{J_0^+} - \tilde h_{J_0^+} \|_{\infty} \right) + \| \partial^a \tilde {h}_{J_0^+} - \partial^a  h_0 \|_{\infty} \,.
\]
By similar arguments to the proof of Corollary 3.1 of \cite{chen2018optimal}, we may also deduce  $\| \partial^a \tilde {h}_{J_0^+} - \partial^a  h_0 \|_{\infty} \lesssim (J_0^+)^{(|a|-p)/d}$ and so
\[
 \| \partial^a \hat{h}_{\hat{J}} - \partial^a  h_0 \|_{\infty}
 \lesssim (J_0^+)^{|a|/d} \left( \| \hat{h}_{\hat{J}} - \hat{h}_{J_0^+} \|_{\infty} +  \| \hat {h}_{J_0^+} - \tilde h_{J_0^+} \|_{\infty} + (J_0^+)^{-p/d} \right)  \,.
\]
It now follows by similar arguments to part (i), step 1 of the proof of Theorem~\ref{lepski2} and definition of $J_0$ that there exists a constant $C > 0$ for which
\[
 \inf_{p \in [\ul p,\ol p]}  \inf_{h_0 \in \mc H^p} \mathbb{P}_{h_0} \left(  \| \partial^a \hat{h}_{\hat{J}} - \partial^a h_0  \|_{\infty}  \leq  C J_0^{(|a|- p)/d}  \right) \rightarrow 1.
\]
The result follows from noting, as in the proof of part (i), step 1 of the proof of Theorem~\ref{lepski2}, that
\[
 \inf_{p \in [\ul p,\ol p]}   \inf_{h_0 \in \mc H^p}  \mathbb{P}_{h_0}  \big( J_n^*(p,E) \leq J_0(p,D)  \big) \to 1\,,
\] where $J_n^*(p,E) \asymp (n/\log n)^{d/(2(p + \varsigma) + d)}$, and by part (i), step 2 of the proof of Theorem~\ref{lepski2} (which shows that $\tilde J = \hat J$ wpa1 $\mc H$-uniformly).

\underline{Part (ii):} Recall $\bar J_{\max}^*(R)$ from (\ref{eq:Jmax*}) and $\hat J_n$ from the definition of $\tilde J$. By similar arguments to part (ii), step 1 of the proof of Theorem~\ref{lepski2}, and the proof of Corollary 3.1 of \cite{chen2018optimal}, we may deduce
\begin{align*}
 & \|\partial^a \hat{h}_{\hat J_n} - \partial^a h_0 \|_{\infty} \\
 & \lesssim  (\bar J_{\max}^*(R_2))^{\frac{|a|}{d}} \bigg( (2^{-d}\bar J_{\max}^*(R_2))^{- \frac pd}  +  \tau_{2^{-d}\bar J_{\max}^*(R_2)}  \sqrt{2^{-d}\bar J_{\max}^*(R_2) \log(2^{-d}\bar J_{\max}^*(R_2)) / n}    \bigg)
\end{align*}
wpa1 uniformly over $ \mc H^p$ and $p \in [\ul p,\ol p]$. Hence, by part (ii), step 1 of the proof of Theorem~\ref{lepski2},
\[
 \|\partial^a \hat{h}_{\hat J_n} - \partial^a h_0 \|_{\infty} \lesssim (\log n)^{(|a|-p)/d}
\]
wpa1 uniformly over $\mc H^p$ and $p \in [\ul p,\ol p]$.

By similar arguments to part (ii), step 2 of the proof of Theorem~\ref{lepski2}, it suffices to show that there exists a constant $C > 0$ for which
\[
 \sup_{p \in [\ul p,\ol p]} \sup_{h_0 \in \mc H^p} \mathbb{P}_{h_0} \big(  \| \partial^a \hat{h}_{\hat{J}} - \partial^a  \hat{h}_{\hat J_n} \|_{\infty} > C (\log n)^{(|a|-p) / \varsigma} \mbox{ and } \hat{J} < \hat J_n   \big) \rightarrow 0\,.
\]
But for any $\hat J \leq \hat J_n$ by a Bernstein inequality (or inverse estimate) for our choice of basis, 
\[
  \| \partial^a \hat{h}_{\hat{J}} - \partial^a  \hat{h}_{\hat J_n} \|_{\infty} \lesssim (\hat J_n)^{|a|/d} \|  \hat{h}_{\hat{J}} - \hat{h}_{\hat J_n} \|_{\infty} \lesssim (\bar J_{\max}^*(R_2))^{|a|/d} \|  \hat{h}_{\hat{J}} - \hat{h}_{\hat J_n} \|_{\infty}
\]
wpa1 uniformly over $\mc H^p$ and $p \in [\ul p,\ol p]$, where the second inequality is because $\hat J_n \leq \hat J_{\max} \leq \bar J_{\max}(R_2)$ wpa1 $\mc H$-uniformly by Lemma~\ref{lem:gridtest} and because $\bar J_{\max}(R_2) = \bar J_{\max}^*(R_2)$ for all $n$ sufficiently large. But note by severe ill-posedness and definition of $\bar J_{\max}^*(R_2)$, we have that $C (\bar J_{\max}^*(R_2))^{\varsigma/d} \asymp \log \tau_{\bar J_{\max}^*(R_2)} \leq \log(R_2 \sqrt n) \asymp \log n$, and so $\bar J_{\max}^*(R_2) \lesssim (\log n)^{d/\varsigma}$. The result now follows by part (ii), step 2 of the proof of Theorem~\ref{lepski2}.
\end{proof}

\begin{proof}[Proof of Theorem~\ref{confmild}]
In some of what follows, we use the fact that the sieve dimensions for CDV wavelet bases are linked via $J^+ = 2^d J$ for $J \in \T$. We do so for notational convenience; a similar argument (with more complicated notation) applies for B-splines.

\underline{Part (i), step 1:} By part (i), step 2 of the proof of Theorem~\ref{lepski2}, we have $\hat{J} = \tilde{J}$ wpa1 $\mc H$-uniformly. It therefore suffices to prove the result for the band
\[
 C_n(x) = \bigg[ \hat{h}_{\hat{J}}(x) - \left( z_{1-\alpha}^* + \hat A \theta^*_{1-\hat \alpha} \right ) \hat \sigma_{\hat J}(x) , ~  \hat{h}_{\hat{J}}(x) + \left( z_{1-\alpha}^* + \hat A \theta^*_{1-\hat \alpha} \right) \hat \sigma_{\hat J}(x) \bigg],
\]
(cf. (\ref{band})). Note by Appendix~\ref{sec:regression} this implies the result holds for our UCBs for nonparametric regression as well. Fix $R_2 > 0$ in the definition of $\bar J_{\max}(R_2)$ from (\ref{eq:J_bar_max}) sufficiently large so that by Lemma~\ref{lem:gridtest} we have $\inf_{h_0 \in \mc H} \mathbb{P}_{h_0} ( \hat{J}_{\max} \leq \bar{J}_{\max}(R_2)) \to 1$. Let $\bar J_{\max} \equiv \bar J_{\max}(R_2)$ for the remainder of the proof. Recall the constants $C_{\ref{unifbiasvar}}$ from (\ref{eq:lepski2_rate}), $\ul B$ and $\ol B$ from the discussion preceding the statement of this theorem, and $C_4$ and $C_5$ from (\ref{eq:lepski2_jmax}). Also note that by Lemmas~\ref{varest3} and~\ref{lem:gridtest}, Assumption~\ref{a-var}(i), and the fact that $\delta_n \downarrow 0$ (cf. (\ref{eq:delta_n})) imply that there exists $C_2, C_3 > 0 $ which satisfy
\begin{equation} \label{eq:confmild_sd}
 \inf_{h_0 \in \mathcal{H}}  \mathbb{P}_{h_0}  \bigg(  \sup_{(x,J) \in \mc X \times \hat{\mc J}} \frac{ \tau_{J} \sqrt{J} }{\| \hat{\sigma}_{x,J} \|_{sd}} \leq C_3     \bigg) \rightarrow 1 \, , \;
 \inf_{h_0 \in \mathcal{H}} \mathbb{P}_{h_0} \bigg( \sup_{(x,J) \in \mc X \times \hat{\mc J}} \frac{\| \hat{\sigma}_{x,J} \|_{sd}}{\tau_{J} \sqrt{J}} \leq C_2 \bigg) \rightarrow 1\,.
\end{equation}

Let $v =  \inf_{J \in \T} (1+ \| \Pi_J \|_{\infty})^{-1} > 0$, where $\| \Pi_J \|_{\infty} \lesssim 1$ is the Lebesgue constant for $\Psi_J$ (see Appendix~\ref{ax:besov}). Choose $\beta \in (0,1) $ and $E > 0 $ such that $(v \ul B \beta^{-\ul p/d} - (C_{\ref{unifbiasvar}} + 1) \ol B) > 0 $ and $E^{-1} (v \ul B \beta^{-\ul p/d} - (C_{\ref{unifbiasvar}} + 1) \ol B) > C_2 (\xi+1)$, where $\xi>1 $ ($\xi = 1.1$ in the main text).

Define $J_0(p,E)$ as in (\ref{eq:J0_mild}). Part (i), step 1 of the proof of Theorem~\ref{lepski2} implies that $J_0(p,E) \gtrsim (n/\log n)^{d/(2(p + \varsigma) + d)}$. By Lemma~\ref{lem:gridtest} and mild ill-posedness, for any constant $C>0$ we have $J_0(p,E) / (\log \hat J_{\max})^2 \geq C$ wpa1 uniformly for $h_0 \in \mc H^p$ and $p \in [\ul p,\ol p]$. Hence, $\inf \{ J \in \T : J \geq \beta J_0(p,E) \} > \log \hat J_{\max}$ wpa1 uniformly for $h_0 \in \mc H^p$ and $p \in [\ul p,\ol p]$.

Fix any $J \in \hat{\mc J}$ with $J < \beta J_0(p,E)$ (this is justified wpa1 uniformly for $h_0 \in \mc H^p$ and $p \in [\ul p,\ol p]$ by the preceding paragraph) and note (dropping dependence of $J_0$ on $(p,E)$)
\begin{align*}
 \| \hat{h}_J - \hat{h}_{J_0} \|_{\infty} & =  \| \hat{h}_J - \hat{h}_{J_0} - \tilde{h}_J + \tilde{h}_{J} - \tilde{h}_{J_0} + \tilde{h}_{J_0} - h_0 + h_0  \|_\infty \\
  & \geq \| \tilde{h}_J - h_0   \|_\infty -   \| \tilde{h}_{J_0} - h_0 \|_\infty - \| \hat{h}_J - \tilde{h}_J - (\hat{h}_{J_0} - \tilde{h}_{J_0})   \|_\infty .
\end{align*}
For a given $h_0 \in \mc G^p$, let $h_{0,J} \in \arg \min_{h \in \Psi_J} \|h - h_0  \|_{\infty}$. Recall $\ul J$ from the definition of $\mc G^p$ and note that $\inf\{J : J \in \hat{\mc J}\} \geq \ul J$ holds wpa1 $\mc H$-uniformly by Lemma~\ref{lem:gridtest}. Recalling the Lebesgue constant $\|\Pi_J\|_\infty$ from Appendix~\ref{ax:besov}, we may then deduce
\[
 \| \tilde{h}_J - h_0   \|_\infty \geq \| h_{0,J} - h_0   \|_\infty \geq (1+ \|   \Pi_J \|_\infty)^{-1} \| h_0 - \Pi_J h_0   \|_\infty \geq v \ul B J^{-p/d}\,,
\]
for all $J \in \hat{\mc J}$ wpa1, uniformly for all $h_0 \in \mc G^p$ and all $p \in [\ul p, \ol p]$.
It follows by (\ref{eq:lepski2_rate}) and the discussion preceding the statement of this theorem that
\begin{align*}
 \| \hat{h}_J - \hat{h}_{J_0} \|_{\infty} & \geq v \ul B J^{-p/d} - (C_{\ref{unifbiasvar}}+1) \ol B J_0^{-p/d} - \| \hat{h}_J - \tilde{h}_J - (\hat{h}_{J_0} - \tilde{h}_{J_0}) \|_\infty \\
 & \geq (v \ul B \beta^{-\ul p/d} -  (C_{\ref{unifbiasvar}}+1) \ol B) J_0^{-p/d} - \| \hat{h}_J - \tilde{h}_J - (\hat{h}_{J_0} - \tilde{h}_{J_0}) \|_\infty \\
 & > C_2(\xi + 1) \tau_{J_0} \frac{\sqrt{J_0} \theta^*_{1-\hat \alpha}}{\sqrt{n}} - \| \hat{h}_J - \tilde{h}_J - (\hat{h}_{J_0} - \tilde{h}_{J_0}) \|_\infty  \, ,
\end{align*}
where the second line uses $J < \beta J_0$ and the third uses definition of $E$ and $J_0(p,E)$. It now follows by the preceding display and (\ref{eq:confmild_sd}) that
\begin{align*}
 & \sup_{p \in [\ul p, \ol p]}  \sup_{h_0 \in \mathcal{G}^p} \mathbb{P}_{h_0} \big( \hat{J} < \beta J_0(p,E) \big) \\
 & \leq  \sup_{p \in [\ul p, \ol p]} \sup_{h_0 \in \mathcal{G}^p} \mathbb{P}_{h_0} \bigg(  \inf_{J \in \hat{\mc J} : J < \beta J_0} \sup_{x \in \mathcal{X}}  \frac{\sqrt{n} | \hat{h}_J(x) - \hat{h}_{J_0}(x) |}{ \| \hat{\sigma}_{x,J,J_0} \|_{sd}} \leq \xi \theta^*_{1-\hat \alpha} \bigg) \\
 & \leq  \sup_{p \in [\ul p, \ol p]} \sup_{h_0 \in \mathcal{G}^p} \mathbb{P}_{h_0} \bigg( \sup_{(x,J,J_2) \in \hat{\mc S}} \frac{\sqrt{n} | \hat{h}_{J}(x) - \hat{h}_{J_2} (x) - ( \tilde{h}_{J} (x) - \tilde{h}_{J_2} (x) )|}{\| \hat{\sigma}_{x,J,J_2} \|_{sd}}  > \theta^*_{1-\hat \alpha} \bigg) + o(1) \\
 & \leq  \sup_{h_0 \in \mathcal{H}} \mathbb{P}_{h_0} \bigg( \sup_{(x,J,J_2) \in \hat{\mc S}} \frac{\sqrt{n} | \hat{h}_{J}(x) - \hat{h}_{J_2} (x) - ( \tilde{h}_{J} (x) - \tilde{h}_{J_2} (x) )|}{\| \hat{\sigma}_{x,J,J_2} \|_{sd}}  > \theta^*_{1-\hat \alpha} \bigg) + o(1) \to 0 \,,
\end{align*}
where the final line is by (\ref{eq:lepski2_mild_prob_2}).

\underline{Part (i), step 2:} Recall $J_0^+(p,D)$ from part (i), step 1 of the proof of Theorem~\ref{lepski2}. By the previous step of this proof and part (i), step 1 of the proof of Theorem~\ref{lepski2}, we have
\begin{equation}\label{eq:confmild_1}
 \inf \limits_{p \in [\ul p, \ol p]} \inf_{h_0 \in \mathcal{G}^p} \mathbb{P}_{h_0} \big( \beta J_0(p,E) \leq \hat{J} \leq J_0^+(p,D) \big)  \to 1\,.
\end{equation}
Therefore, by (\ref{eq:lepski2_rate}), (\ref{eq:confmild_sd}), (\ref{eq:confmild_1}), and definition of $\ol B$, for every $h_0 \in \mc G^p$ and $x \in \mathcal{X}$ we have
\begin{multline*}
 \frac{ | \tilde{h}_{\hat{J}}(x) - h_0(x)|}{ \| \hat{\sigma}_{x,\hat{J}} \|_{sd}}
  \leq (C_{\ref{unifbiasvar}}+1) C_3 \ol B \frac{\hat J^{-p/d}}{\tau_{\hat J} \sqrt{\hat J}}
  \leq (C_{\ref{unifbiasvar}}+1) C_3 \ol B \beta^{- \ol p/d} 2^p \frac{(2^d J_0(p,E))^{-p/d}}{ \tau_{\lceil \beta J_0(p,E) \rceil} \sqrt{\beta J_0(p,E)} } \,,
\end{multline*}
wpa1 uniformly for $h_0 \in \mc G^p$ and $p \in [\ul p,\ol p]$ and $x \in \mathcal{X}$, where $\tau_{\lceil \beta J_0(p,E) \rceil}$ denotes the ill-posedness at resolution level $\inf\{J \in \T : J \geq \beta J_0(p,E)\}$. It now follows from definition of $2^d J_0(p,E) \equiv J_0^+(p,E)$ from (\ref{eq:J0_mild}) that whenever the preceding inequality holds, we have
\begin{align*}
 \sup_{x \in \mathcal{X}} \sqrt{n} \frac{ | \tilde{h}_{\hat{J}}(x) - h_0(x) |}{ \| \hat{\sigma}_{x,\hat{J}} \|_{sd}}
 \leq  C_3(C_{\ref{unifbiasvar}}+1) \ol B \beta^{- \ol p/d - 1/2} 2^{\ol p + d/2} E^{-1} \frac{\tau_{2^d J_0(p,E)}}{\tau_{\lceil \beta J_0(p,E) \rceil}} \theta^*_{1-\hat \alpha} < A_0 \theta^*_{1-\hat \alpha} \,,
\end{align*}
where the final inequality holds uniformly for $h_0 \in \mc G^p$ and $p \in [\ul p,\ol p]$ for a constant $A_0 > 0$ because $\sup_{J \in \T} \tau_{2^d J}/\tau_{\lceil \beta J \rceil} < \infty$ by virtue of mild ill-posedness. Hence for any $A \geq A_0$,
\begin{align*}
 & \inf_{h_0 \in \mathcal{G}} \mathbb{P}_{h_0} \left(  h_0(x) \in C_n(x,A) \; \; \forall \; \; x  \in \mathcal{X} \right) \\
 & \geq \inf_{p \in [\ul p,\ol p]} \inf_{h_0 \in \mathcal{G}^p} \mathbb{P}_{h_0} \bigg( \sup_{x \in \mathcal{X}} \sqrt{n} \frac{ | \hat{h}_{\hat{J}}(x) - h_0(x) | }{\| \hat{\sigma}_{x,\hat{J}} \|_{sd}} \leq z_{1 - \alpha}^* + A \theta^*_{1-\hat \alpha} \bigg) + o(1) \\
 & \geq \inf_{p \in [\ul p,\ol p]} \inf_{h_0 \in \mathcal{G}^p} \mathbb{P}_{h_0} \bigg( \sup_{x \in \mathcal{X}} \sqrt{n} \frac{ | \hat{h}_{\hat{J}}(x) - \tilde{h}_{\hat{J}}(x) |}{\|  \hat{\sigma}_{x,\hat{J}} \|_{sd}} \leq z_{1 - \alpha}^* \bigg) + o(1) \\
 & \geq \inf_{p \in [\ul p,\ol p]} \inf_{h_0 \in \mathcal{G}^p} \mathbb{P}_{h_0} \bigg( \sup_{(x,J) \in \mathcal{X} \times \ul{\mc J}_n} \sqrt{n} \frac{ | \hat{h}_{J}(x) - \tilde{h}_{J}(x) |}{\|  \hat{\sigma}_{x,J} \|_{sd}} \leq z_{1 - \alpha}^* \bigg) + o(1) \,,
\end{align*}
where the final line is because $\hat J \in \ul{\mc J}_n := \{J \in \T : 0.1 (\log \bar J_{\max}(R_2))^2 \leq J \leq \bar J_{\max}^-(R_1)\}$ with $\bar J_{\max}^-(R_1) = \sup\{J \in \mc T : J < \bar J_{\max}(R_1)\}$ and $\hat{\mc J}_- \supseteq \ul{\mc J}_n$ both hold wpa1 uniformly for $h_0 \in \mc G^p$ and $p \in [\ul p,\ol p]$; the former holds by (\ref{eq:confmild_1}) and Lemma~\ref{lem:J0_mild} and the latter holds by Lemma~\ref{lem:gridtest} and the fact that $\hat J = \tilde J$ wpa1 $\mc H$-uniformly. Let $\ul z_{1-\alpha}^*$ denote the $1-\alpha$ quantile of $\sup_{(x,J) \in \mathcal{X} \times \ul{\mc J}_n} \left| \mathbb{Z}_n^*(x,J) \right|$. As $\ul z_{1-\alpha}^* \leq z_{1-\alpha}^*$ must hold whenever $\hat{\mc J} \supseteq \ul{\mc J}_n$, we therefore have
\begin{align*}
 & \inf_{h_0 \in \mathcal{G}} \mathbb{P}_{h_0} \left(  h_0(x) \in C_n(x,A) \; \; \forall \; \; x  \in \mathcal{X} \right) \\
 & \geq \inf_{p \in [\ul p,\ol p]} \inf_{h_0 \in \mathcal{G}^p} \mathbb{P}_{h_0} \bigg( \sup_{(x,J) \in \mathcal{X} \times \ul{\mc J}_n} \sqrt{n} \frac{ | \hat{h}_{J}(x) - \tilde{h}_{J}(x) |}{\|  \hat{\sigma}_{x,J} \|_{sd}} \leq \ul z_{1 - \alpha}^* \bigg) + o(1)  = (1-\alpha) + o(1) \,,
\end{align*}
where the last equality follows from Theorem~\ref{consistent}(i) and the definition of $\ul z_{1-\alpha}^*$.

\underline{Part (ii):} By Lemmas~\ref{lem-varest2}, \ref{lepquant}, and~\ref{zorder} and Assumption~\ref{a-var}(i), we have
\[
 \sup_{x \in \mathcal{X}} |C_n(x,A)	| \lesssim (1+A)\tau_{\hat J} \sqrt{(\hat J \log \bar{J}_{\max}) / n}
\]
wpa1 $\mc H$-uniformly. Then by (\ref{eq:confmild_1}) with $J_0 = J_0(p,D)$ and $\bar A = 1 + A$, we have that
\[
 \sup_{x \in \mathcal{X}} |C_n(x,A)|
 \lesssim \bar A \tau_{J_0^+} \sqrt{(J_0^+ \log \bar{J}_{\max}) / n}
 \lesssim \bar A \tau_{J_0} \sqrt{(J_0 \log \bar{J}_{\max}) / n}
 \lesssim \bar A \frac{\sqrt{\log \bar{J}_{\max}}}{\theta^*_{1-\hat \alpha}} J_0^{-p/d}
\]
holds wpa1 uniformly for $h_0 \in \mc G^p$ and $p \in [\ul p,\ol p]$ and for all $A > 0$, where the second inequality follows from the fact that the model is mildly ill-posed and the third is by definition (\ref{eq:J0_mild}). It follows by Lemma~\ref{lepquant} that there is a constant $C > 0$ (independent of $A$) for which
\[
 \inf_{p \in [\ul p,\ol p]} \inf_{h_0 \in \mc G^p} \mathbb{P}_{h_0} \left( \sup_{x \in \mathcal{X}} |C_n(x,A)| \leq C (1+A) (J_0(p,D))^{-p/d}    \right) \rightarrow 1 \,.
\]
The result now follows from part (i), step 2 of the proof of Theorem~\ref{lepski2}, which shows that  $\inf_{p \in [\ul p,\ol p]}   \inf_{h_0 \in \mc H^p}  \big( J_n^*(p,E) \leq J_0(p,D)  \big) \to 1$ with $J_n^*(p,E) \asymp (n/\log n)^{d/(2(p + \varsigma) + d)}$.
\end{proof}

\begin{proof}[Proof of Theorem~\ref{confsevere}]
In some of what follows, we use the fact that the sieve dimensions for CDV wavelet bases are linked via $J^+ = 2^d J$ for $J \in \T$. A similar argument (with more complicated notation) applies for B-spline bases.

\underline{Part (ii):} First note by Lemma \ref{lem:gridtest} and the fact that $\bar J_{\max}(R) = \bar J_{\max}^*(R)$ (see (\ref{eq:Jmax*})) holds for any $R > 0$ for all $n$ sufficiently large (see part (ii), step 1 of the proof of Theorem \ref{lepski2}), we have that $J_{\max}^*(R_1) \leq \hat J_{\max} \leq J_{\max}^*(R_2)$ wpa1 $\mc H$-uniformly.

Recall $M_0(p,R_2)$ from (\ref{eq:M0}).  By Lemma \ref{lem:J0_severe}, for all $p \in [\ul p,\ol p]$ we have that $M_0(p,R_2) \geq M_0(\ol p,R_2) \geq 2^{-d} J_{\max}^*(R_2)$ holds for all $n$ sufficiently large. Then by Lemmas \ref{lem-varest2}, \ref{lepquant}, and \ref{zorder} and Assumption \ref{a-var}(i), there exist constants $C,C' > 0$ for which
\[
 \sup_{x \in \mathcal{X}} |C_n(x,A)| \leq  C(1+A) \tau_{\tilde{J}} \sqrt{(\tilde{J} \log ( \bar{J}_{\max}^*(R_2) ) ) / n} + A \tilde{J}^{-\ul p/d} \leq C'(1+A)(J_{\max}^*(R_2))^{-p/d} + A \tilde J^{-\ul p/d}
\]
holds wpa1 uniformly for $h_0 \in \mc H^p$ and $p \in [\ul p,\ol p]$, where the second inequality is by definition of $M_0(p,R_2)$. The proofs of Theorem \ref{lepski2} and Corollary \ref{cor:lepski2} show that $\bar{J}_{\max}^*(R_2) \asymp (\log n)^{d/\varsigma}$ in the severely ill-posed case. Therefore, it suffices to show that there is a constant $c > 0$ for which $\hat{J} \geq c (\log n)^{d/\varsigma}$ holds wpa1 uniformly for $h_0 \in \mc G^p$ and $p \in [\ul p,\ol p]$.

Recall $\beta$ and $E$ from the proof of Theorem \ref{confmild} and $J_0(p,E)$ from (\ref{eq:J0_mild}). By similar arguments to Lemma \ref{lem:J0_severe}, we may deduce that $\inf \{ J \in \T : J \geq \beta J_0(p,E) \} > \log \hat J_{\max}$ wpa1 uniformly for $h_0 \in \mc H^p$ and $p \in [\ul p,\ol p]$. It then follows by the same argument as part (i), step 1 of the proof of Theorem \ref{confmild} that $\hat{J} \geq \beta J_0(p,E)$ holds wpa1 uniformly for $h_0 \in \mc G^p$ and $p \in [\ul p,\ol p]$. But by Lemma \ref{lepquant} and the fact that $\log \bar J_{\max}^*(R_2) \asymp \log \log n$ for severely ill-posed models, it follows that there is a constant $C'' > 0$ for which, by defining
\[
 J^*(p,C'') = \sup \bigg \{ J \in \T : \tau_J \sqrt{(J \log \log  n) / n} \leq C'' J^{-p/d}  \bigg \} \,,
\]
we have $ \inf_{p \in [\ul p,\ol p]} \inf_{h_0 \in \mc H^p} \mathbb{P}_{h_0} ( J_0(p,E) \geq  J^*(p,C'') ) \to 1$. Finally, we may deduce by a similar argument to part (ii), step 1 of the proof of Theorem \ref{lepski2} that $J^*(p,C'') \gtrsim (\log n)^{d/\varsigma}$ for all $p \in [\ul p, \ol p]$, which establishes the desired behavior of $\hat J$.

\underline{Part (i):} By Theorem \ref{unifbiasvar} and Lemma \ref{lem:gridtest}, there exists a constant $A_0 > 0 $ for which
\[
 | \hat{h}_{\tilde{J}}(x) - h_0(x)| \leq | \hat{h}_{\tilde{J}}(x) - \tilde{h}_{\tilde{J}}(x) | + A_0 \tilde{J}^{-\ul p/d}
\]
holds for all $x \in \mc X$ wpa1 $\mc H$-uniformly. Then for any $A \geq A_0$, we have
\[
 \inf_{h_0 \in \mathcal{G} } \mathbb{P}_{h_0} \big( h_0(x) \in C_n(x,A) \; \; \forall \; x \in \mathcal{X} \big)
 \geq \inf_{h_0 \in \mathcal{G}} \mathbb{P}_{h_0} \bigg( \sup_{x \in \mathcal{X}} \left|\sqrt{n} \frac{ \hat{h}_{\tilde{J}}(x) - \tilde{h}_{\tilde{J}}(x)}{\| \hat \sigma_{x,\tilde{J}} \|_{sd}} \right| \leq z_{1- \alpha}^* \bigg) + o(1) \,.
\]

Suppose that $J_{\max}^*(R_2) \geq 2^{2d} J_{\max}^*(R_1) \in \T$. Then by definition of $\bar J_{\max}^*(R)$ and Remark \ref{rmk:J_max}, we have
\begin{equation}\label{eq:confsevere_1}
 \frac{\tau_{J_{\max}^*(R_2)}}{\tau_{2^{2d} J_{\max}^*(R_1)} } \asymp \frac{\tau_{J_{\max}^*(R_2)}}{\tau_{2^{2d} J_{\max}^*(R_1)} } \frac{J_{\max}^*(R_2) \sqrt{\log J_{\max}^*(R_2)}}{2^{2d} J_{\max}^*(R_1) \sqrt{\log J_{\max}^*(R_1)}} \leq \frac{R_2}{R_1} \,.
\end{equation}
But note that if $J_{\max}^*(R_2) \geq 2^{2d} J_{\max}^*(R_1)$ then by severe ill-posedness we have
\[
 \frac{\tau_{J_{\max}^*(R_2)}}{\tau_{2^d J_{\max}^*(R_1)} } \geq \frac{\tau_{2^{2d} J_{\max}^*(R_1)}}{\tau_{2^d J_{\max}^*(R_1)} } \asymp e^{C(( 2^{2d} J_{\max}^*(R_1))^{\varsigma/d} - (2^d J_{\max}^*(R_1))^{\varsigma/d})} = e^{C 2^\varsigma (2^{\varsigma } - 1) (J_{\max}^*(R_1))^{\varsigma/d}} \to +\infty\,,
\]
which contradicts (\ref{eq:confsevere_1}). Therefore, $\bar J_{\max}^*(R_1) \in \{2^{-d}\bar J_{\max}^*(R_2),\bar J_{\max}^*(R_2)\}$ holds for all $n$ sufficiently large, from which it follows by Lemma \ref{lem:gridtest} that $\hat J_{\max} \in \{2^{-d}\bar J_{\max}^*(R_2),\bar J_{\max}^*(R_2)\}$ wpa1 $\mc H$-uniformly. Therefore, $\tilde J \leq 2^{-d} \bar J_{\max}^*(R_2)$ holds wpa1 $\mc H$-uniformly. But by part (ii) we also have that $\tilde J \geq c \bar J_{\max}^*(R_2)$ holds for a sufficiently small  $c > 0$ wpa1 uniformly  $h_0 \in \mc G^p$ and $p \in [\ul p,\ol p]$. Therefore,  $\tilde J \in \ul{\mc J}_n := \{J \in \T : c \bar J_{\max}^*(R_2) \leq J \leq 2^{-d} \bar J_{\max}(R_2)\}$ and $\hat{\mc J} \supseteq \ul{\mc J}_n$ both hold wpa1 uniformly for $h_0 \in \mc G^p$ and $p \in [\ul p,\ol p]$.

Let $\ul z_{1-\alpha}^*$ denote the $1-\alpha$ quantile of $\sup_{(x,J) \in \mathcal{X} \times \ul{\mc J}_n} \left| \mathbb{Z}_n^*(x,J) \right|$. As $\ul z_{1-\alpha}^* \leq z_{1-\alpha}^*$ must hold whenever $\hat{\mc J} \supseteq \ul{\mc J}_n$, we therefore have
\begin{align*}
 & \inf_{h_0 \in \mathcal{G}} \mathbb{P}_{h_0} \left(  h_0(x) \in C_n(x,A) \; \; \forall \; \; x  \in \mathcal{X} \right) \\
 & \geq \inf_{p \in [\ul p,\ol p]} \inf_{h_0 \in \mathcal{G}^p} \mathbb{P}_{h_0} \bigg( \sup_{(x,J) \in \mathcal{X} \times \ul{\mc J}_n} \sqrt{n} \frac{ | \hat{h}_{J}(x) - \tilde{h}_{J}(x) |}{\|  \hat{\sigma}_{x,J} \|_{sd}} \leq \ul z_{1 - \alpha}^* \bigg) + o(1) = (1-\alpha) + o(1) \,,
\end{align*}
where the last equality follows from Theorem \ref{consistent}(i) and the definition of $\ul z_{1-\alpha}^*$.
\end{proof}

\subsection{Supplemental Results: UCBs for Derivatives}

Here we present supplemental results for the proofs of Theorems~\ref{confmild-derivative} and~\ref{confsevere-derivative}. Throughout this subsection, for any fixed $R > 0$, let $\bar J_{\max} \equiv \bar J_{\max}(R)$. Also let $J_{\min} \to \infty$ as $n \to \infty$ with $J_{\min} \leq \bar J_{\max}$. Define  $\mc J_n = \{J \in \T : J_{\min} \leq J \leq \bar J_{\max}\}$. Also recall $\delta_n$ from (\ref{eq:delta_n}). We introduce the bootstrap process for the derivatives:
\[
 \mathbb{Z}_n^{a*}(x,J)\equiv \frac{D_J^{a*}(x)}{\hat \sigma_J^{a}(x)}= \frac{1}{\|\hat{\sigma}_{x,J}^a\|_{sd}} \left(\frac{1}{\sqrt n} \sum_{i=1}^n \hat{L}_{J,x}^a b^{K(J)}_{W_i} \hat{u}_{i,J} \varpi_i \right) ,
\]
where $\|\hat{\sigma}_{x,J}^a\|^2_{sd} \equiv n\hat{\sigma}_{J}^{a2}(x) = \hat{L}_{J,x}^a  \widehat{\Omega}_{J,J} (\hat{L}_{J,x}^a )'$ and $\hat{L}_{J,x}^a  =  (\partial^a \psi_x^J)' [\widehat S_J' \widehat G_{b,J}^{-1} \widehat S_J ]^{-1} \widehat S_J' \widehat G_{b,J}^{-1}$ with $\partial^a \psi^J_x$ denoting the derivative applied element-wise: $\partial^a \psi^J_x = (\partial^a \psi_{J1}(x),\ldots,\partial^a \psi_{JJ}(x))'$. Proofs of these supplemental results are presented in our earlier working paper version \cite{cck}, where they are labelled as Lemmas E.12, E.13, and E.14, respectively.

\begin{lemma}\label{lem:var-derivative}
Let Assumptions~\ref{a-data}-\ref{a-approx} hold.  Then: there is a universal constant $C_{\ref{lem:var-derivative}} > 0$ such that
\[
 \inf_{h_0 \in \mc H}  \mathbb{P}_{h_0} \bigg( \sup_{(x,J) \in \mathcal{X} \times \mathcal{J}_n}  \left|  \frac{\|\hat{\sigma}_{x,J}^a\|^2_{sd}}{\|\sigma_{x,J}^a \|^2_{sd}} - 1   \right|   \leq C_{\ref{lem-varest2}} \delta_n \bigg)  \rightarrow 1 \,.
\]
\end{lemma}

\begin{lemma} \label{zorder-derivative}
Let Assumptions~\ref{a-data}-\ref{a-var} hold. For a given $\alpha \in (0,1)$, let $z_{1- \alpha}^{a*}$ denote the $1-\alpha$ quantile of $\sup_{ (x,J) \in \mathcal{X} \times \hat{\mc J} } | \mathbb{Z}_n^{a*}(x,J) |$. Then: with $\bar J_{\max}(R)$ as defined in (\ref{eq:J_bar_max}) for any $R >0$, there exists a constant $C_{\ref{zorder-derivative}} > 0$ for which
\[
 \inf_{h_0 \in \mathcal{H}} \mathbb{P}_{h_0} \bigg( z_{1- \alpha}^{a*} \leq C_{\ref{zorder-derivative}} \sqrt{\log \bar{J}_{\max}(R)}  \bigg) \rightarrow 1\,.
\]
\end{lemma}

\begin{lemma} \label{consistent-derivative}
Let Assumptions~\ref{a-data}-\ref{a-var} hold and let $J_{\min} \asymp (\log \bar J_{\max})^2$. Then: there exists a sequence $\gamma_n \downarrow 0$ for which
\[
 \sup_{s \in \R} \left| \mathbb{P}_{h_0} \bigg( \sup_{(x,J) \in \mathcal{X} \times \mathcal{J}_n } \left | \sqrt{n} \frac{\partial^a \hat{h}_J(x) - \partial^a  \tilde{h}_J(x)}{\| \hat \sigma_{x,J}^a \|_{sd}} \right| \leq s \right)    - \mathbb{P}^* \bigg( \sup_{(x,J) \in \mathcal{X} \times \mathcal{J}_n} \left| \mathbb{Z}_n^{a*}(x,J) \right|  \leq s  \bigg)  \bigg |  \leq \gamma_n
\]
holds wpa1 $\mc H$-uniformly.
\end{lemma}

\subsection{Proofs of Theorems~\ref{confmild-derivative} and~\ref{confsevere-derivative} on UCBs for Derivatives}

\begin{proof}[Proof of Theorem~\ref{confmild-derivative}]
The proof follows similar arguments to the proof of Theorem~\ref{confmild}. Here we state the necessary modifications.

\underline{Part (i), step 1:} Identical to part (i), step 1 of the proof of Theorem~\ref{confmild}.

\underline{Part (i), step 2:} Note that by Theorem~\ref{unifbiasvar} and a similar argument to the proof of Corollary 3.1 of \cite{chen2018optimal}, we have
\[
 \inf_{h_0 \in \mathcal{H}} \mathbb{P}_{h_0} \bigg( \| \partial^a \tilde{h}_J - \partial^a h_0  \|_{\infty} \leq C_6 J^{(|a|-p)/d} \; \; \; \;  \forall \; J \in [1, \bar{J}_{\max}] \cap \T \bigg) \rightarrow 1
\]
for some constant $C_6 > 0$. Moreover, by Lemma~\ref{lem:var-derivative} and Assumption~\ref{a-var}(iii) there is a constant $C_7 > 0$ for which
\[
 \inf_{h_0 \in \mathcal{H}}  \mathbb{P}_{h_0}  \bigg(  \sup_{(x,J) \in \mc X \times \hat{\mc J}} \frac{ \tau_{J} J^{1/2 + |a|/d} }{\| \hat{\sigma}_{x,J}^{a} \|_{sd}} \leq C_7     \bigg) \rightarrow 1 \,.
\]
It now follows by (\ref{eq:confmild_1}) that
\[
 \frac{ | \partial^a \tilde{h}_{\hat{J}}(x) - \partial^a h_0(x)|}{ \| \hat{\sigma}_{x,\hat{J}}^a \|_{sd}}
  \leq C_6 C_7 \frac{\hat J^{-p/d}}{ \tau_{\hat J} \sqrt{\hat J}}
  \leq C_6 C_7 \beta^{- \ol p/d} 2^{\ol p} \frac{(2^d J_0(p,E))^{-p/d}}{ \tau_{\lceil \beta J_0(p,E) \rceil} \sqrt{\beta J_0(p,E)} } \,,
\]
wpa1 uniformly for $h_0 \in \mc G^p$ and $p \in [\ul p,\ol p]$ and $x \in \mathcal{X}$. The remainder of the proof of this part now follows by identical arguments to part (i), step 2 of the proof of Theorem~\ref{confmild}, using Lemma~\ref{consistent-derivative} in place of Theorem~\ref{consistent}(i).

\underline{Part (ii):} By Lemma~\ref{lepquant}, Lemmas~\ref{lem:var-derivative} and~\ref{zorder-derivative} and Assumption~\ref{a-var}(iii), we have
\[
 \sup_{x \in \mathcal{X}} |C_n^a(x,A)	| \lesssim  (1+A)\tau_{\hat J} \hat J^{1/2 + |a|/d} \sqrt{(\log \bar{J}_{\max}) / n}
\]
wpa1 $\mc H$-uniformly. Then by display (\ref{eq:confmild_1}), with $J_0 = J_0(p,D)$  we have that
\begin{align*}
 \sup_{x \in \mathcal{X}} |C_n^a(x,A)|
 & \lesssim  (1+A)\tau_{J_0^+} (J_0^+)^{1/2 + |a|/d} \sqrt{(\log \bar J_{\max}) / n} \\
 & \lesssim  (1+A)\tau_{J_0} J_0^{|a|/d} \sqrt{(J_0 \log \bar{J}_{\max}) / n}
 \lesssim (1+A)\frac{\sqrt{\log \bar{J}_{\max}}}{\theta^*_{1-\hat \alpha}} J_0^{(|a|-p)/d}
\end{align*}
holds wpa1 uniformly for $h_0 \in \mc G^p$ and $p \in [\ul p,\ol p]$, where the second inequality follows from the fact that the model is mildly ill-posed and the third is by definition (\ref{eq:J0_mild}). The result now follows by similar arguments to part (ii) of the proof of Theorem~\ref{confmild}.
\end{proof}

\begin{proof}[Proof of Theorem \ref{confsevere-derivative}]
The proof follows similar arguments to the proof of Theorem~\ref{confsevere}. Here we state the necessary modifications.

\underline{Part (i):} By Lemma \ref{lem:gridtest}, Theorem~\ref{unifbiasvar}, and similar arguments to the proof of Corollary 3.1 of \cite{chen2018optimal}, there exists a constant $A_0 > 0 $ for which
\[
 |\partial^a \hat{h}_{\tilde{J}}(x) - \partial^a h_0(x)| \leq | \partial^a \hat{h}_{\tilde{J}}(x) - \partial^a \tilde{h}_{\tilde{J}}(x) | + A_0 \tilde{J}^{(|a|-\ul p)/d}
\]
holds for all $x \in \mc X$ wpa1 $\mc H$-uniformly. Then for any $A \geq A_0$, we have
\[
 \inf_{h_0 \in \mathcal{G} } \mathbb{P}_{h_0} \big( \partial^a h_0(x) \in C_n^a(x,A) \; \; \forall \; x \in \mathcal{X} \big)
 \geq \inf_{h_0 \in \mathcal{G}} \mathbb{P}_{h_0} \bigg( \sup_{x \in \mathcal{X}} \left|\sqrt{n} \frac{ \partial^a \hat{h}_{\tilde{J}}(x) - \partial^a \tilde{h}_{\tilde{J}}(x)}{\| \hat \sigma_{x,\tilde{J}}^a \|_{sd}} \right| \leq z_{1- \alpha}^{a*} \bigg) + o(1) \,.
\]
The remainder of the proof now follows similarly to the proof of Theorem \ref{confsevere}, using Lemma \ref{consistent-derivative} in place of Theorem \ref{consistent}(i).

\underline{Part (ii):} By Lemmas~\ref{lem:J0_severe}, \ref{lepquant}, \ref{lem:var-derivative}, and~\ref{zorder-derivative} and Assumption \ref{a-var}(iii), there exist constants $C,C' > 0$ for which
\begin{align*}
 \sup_{x \in \mathcal{X}} |C_n^a(x,A)|
 & \leq  C (1 + A) \tau_{\tilde{J}} \tilde{J}^{1/2 + |a|/d}  \sqrt{\log ( \bar{J}_{\max}^*(R_2) ) / n} + A \tilde{J}^{(|a|-\ul p)/d} \\
 & \leq C' (1 + A)(J_{\max}^*(R_2))^{(|a|-p)/d} + A \tilde J^{(|a|-\ul p)/d}
\end{align*}
holds wpa1 uniformly for $h_0 \in \mc H^p$ and $p \in [\ul p,\ol p]$. The remainder of the proof now follows similarly to the proof of Theorem \ref{confsevere}.
\end{proof}

\let\oldbibliography\thebibliography
\renewcommand{\thebibliography}[1]{\oldbibliography{#1}
\setlength{\itemsep}{2pt}}

{\singlespacing
\putbib
}

\end{bibunit}

\end{document}